\documentclass[12pt]{article}
\usepackage{amsmath}
\usepackage{graphicx,psfrag,epsf}
\usepackage{enumerate}
\usepackage{natbib}
\usepackage{url} 
\usepackage{amsthm}

\usepackage{amsfonts}
\usepackage{multirow, booktabs}
\usepackage[table]{xcolor}
\usepackage{cleveref}
\usepackage[T1]{fontenc}
\usepackage[utf8]{inputenc}
\usepackage{authblk}
\usepackage[multiple]{footmisc}
\usepackage{blindtext,titlefoot}
\usepackage{sectsty}
\usepackage{bbm}
\usepackage{extarrows}
\usepackage{amsmath,graphicx,centernot}
\usepackage{mathtools}

\newcommand{\bigCI}{\mathrel{\text{\scalebox{1.07}{$\perp\mkern-10mu\perp$}}}}
\newcommand{\nbigCI}{\centernot{\bigCI}}

\usepackage{amsmath, amssymb}
\usepackage{graphicx,psfrag,epsf}
\usepackage{enumerate}
\usepackage{natbib}
\usepackage{url} 
\usepackage{amsthm}
\usepackage{subcaption}
\usepackage{xcolor}
\usepackage{algorithm,algorithmic,refcount} 
\usepackage{graphicx}
\usepackage{caption}
\usepackage{sidecap}
\usepackage{multirow, booktabs}
\usepackage{algorithm,algorithmic,refcount} 
\newcommand{\indep}{\rotatebox[origin=c]{90}{$\models$}}

\newtheorem{Lemma}{Lemma}
\newtheorem*{Lemma*}{Lemma}
\newtheorem{Definition}{Definition}
\newtheorem{Theorem}{Theorem}
\newtheorem*{Theorem*}{Theorem}
\newtheorem{Example}{Example}
\newtheorem{Remark}{Remark}

\newtheorem{Assumption}{Assumption}
\newtheorem{Simulation}{Simulation}

\title{\Large \bf Re-Evaluating Strengthened-IV Designs: Asymptotic Efficiency, Bias Formula, and the Validity and Power of Sensitivity Analyses}
\date{}

\author{$\text{Siyu Heng}^{1, ^*}$, $\text{Bo Zhang}^{1, ^\ast}$, $\text{Xu Han}^{2}$, $\text{Scott A. Lorch}^{1}$, and $\text{Dylan S. Small}^{1, \dagger}$}

\vspace{-0.3 cm}
\affil{$\textit{University of Pennsylvania}^{1} \textit{ and } \textit{Temple University}^{2}$}

\begin{document}

\sectionfont{\bfseries\large\sffamily}%
%

\subsectionfont{\bfseries\sffamily\normalsize}%
%


\def\spacingset#1{\renewcommand{\baselinestretch}%
{#1}\small\normalsize} \spacingset{1}

\maketitle

\let\thefootnote\relax

\let\thefootnote\relax\footnotetext{$^\ast$ The first two authors contributed equally to this work.}

\let\thefootnote\relax\footnotetext{$\dagger$ \textit{Address for correspondence:} Dylan S. Small, Department of Statistics, The Wharton School, University of Pennsylvania, Philadelphia, PA 19104 (e-mail: \textsf{dsmall@wharton.upenn.edu}).}

\vspace{-0.9 cm} 
\begin{abstract}
Instrumental variables (IVs) are extensively used to estimate treatment effects when the treatment and outcome are confounded by unmeasured confounders; however, weak IVs are often encountered in empirical studies and may cause problems. Many studies have considered building a stronger IV from the original, possibly weak, IV in the design stage of a matched study at the cost of not using some of the samples in the analysis. It is widely accepted that strengthening an IV tends to render nonparametric tests more powerful and will increase the power of sensitivity analyses in large samples. In this article, we re-evaluate this conventional wisdom to bring new insights into this topic. We consider matched observational studies from three perspectives. First, we evaluate the trade-off between IV strength and sample size on nonparametric tests assuming the IV is valid and exhibit conditions under which strengthening an IV increases power and conversely conditions under which it decreases power. Second, we derive a necessary condition for a valid sensitivity analysis model with continuous doses. We show that the $\Gamma$ sensitivity analysis model, which has been previously used to come to the conclusion that strengthening an IV increases the power of sensitivity analyses in large samples, does \textit{not} apply to the continuous IV setting and thus this previously reached conclusion may be invalid. Third, we quantify the bias of the Wald estimator with a possibly invalid IV under an oracle and leverage it to develop a valid sensitivity analysis framework; under this framework, we show that strengthening an IV may amplify or mitigate the bias of the estimator, and may or may not increase the power of sensitivity analyses. We also discuss how to better adjust for the observed covariates when building an IV in matched studies.
\end{abstract}


\thispagestyle{empty}

\spacingset{1.45} 

\section{Introduction}
\label{sec: introduction}
\subsection{Instrumental variable methods in matched observational studies}
\label{subsec: encouragement design example}
An instrumental variable (IV) can be thought of as a haphazard encouragement to take some treatment whose only effect on the outcome is through its effect on the treatment. A randomized encouragement design is an ideal prototype (\citealp{holland1988causal}). In a randomized encouragement design, if the encouragement is associated with the treatment, affects the outcome only through the treatment, and is independent of unmeasured confounders, the encouragement is said to be a valid IV. A valid IV can be used to obtain bounds on the average treatment effect (\citealp{robins1994correcting}; \citealp{balke1997bounds}) and, under additional assumptions, to obtain a consistent estimate of the average treatment effect (for a certain subpopulation or the whole population depending on the assumptions) (\citealp{holland1988causal}; \citealp{AIR1996}; \citealp{hernan2006instruments}; \citealp{swanson2018partial}). This ``randomized encouragement'' perspective of an IV can be further combined with matching, a nonparametric technique of adjusting for observed covariates, when the IV is believed to be randomized conditional on a set of observed covariates. In a matched observational study, subjects with similar observed covariates are put in the same matched set, and comparisons are made within these matched sets (\citealp{rubin1973matching}; \citealp{rosenbaum2002observational,rosenbaum2010design}; \citealp{hansen2004full}; \citealp{rubin2008objective}; \citealp{stuart2010matching}; \citealp{zubizarreta2012using}; \citealp{pimentel2015large}). This study design approach to IV analysis has at least three nice features. First, it facilitates blinding like in a randomized trial as the study is designed before looking at the outcomes. Second, it facilitates nonparametric, randomization-based inference under the assumption that the IV is valid, and associated sensitivity analysis to examine how sensitive a conclusion is to a putative IV not being effectively randomly assigned. Third, it provides a unified framework to deal with continuous and binary outcomes; see \cite{baiocchi2012near} for a more detailed introduction.

An example of this matching-based study design approach to IV analysis concerns the following: does delivery of a preterm infant at a high-level neonatal intensive care unit (NICU) compared to a low-level NICU increase infants' length of hospital stay? A high-level NICU is one with a high level of technology, particularly resuscitative capacity, and high volume, whereas a low-level NICU lacks at least one of these features (\citealp{lorch2012differential}). Following \cite{baiocchi2010building}, we consider data on all premature births in Pennsylvania from 1995 to 2004 plus the first six months of 2005. The putative IV leveraged in the study is the excess travel time, defined as the travel time (in minutes) to the nearest high-level NICU minus the time to the nearest low-level NICU. Such an excess distance IV has been widely used in health studies, e.g., \citet{mcclellan1994does}.

Excess travel time would be a valid IV if a mother's risk of having a long length of stay (e.g., due to complications of the birth or the baby experiencing problems) is not related to whether the mother lives close to a high-level NICU. However, high-level NICUs tend to be in urban areas and mothers living in urban areas are on average different from those living in rural areas in ways that might be related to the risk of a long length of stay. To control for confounders of the relationship between living near a high-level NICU and length of stay, we follow \citet{baiocchi2010building} and use optimal \textit{non-bipartite} matching (\citealp{lu2001matching,lu2011optimal}) to pair mothers with similar observed covariates $\mathbf{X}$, including variables related to mother's pregnancy and socioeconomic status, but different excess travel times. Continuous IVs are sometimes dealt with by dichotomizing them into binary IVs according to some specified cut-off (e.g., dichotomizing excess travel time by above or below the median), but this sacrifices information. Non-bipartite matching preserves the continuous nature of an IV since it does not dichotomize the continuous IV into binary IVs. Non-bipartite matching typically divides $2I$ individuals into $I$ non-overlapping pairs of two individuals through minimizing the sum of distances within the $I$ pair. Typically the distance within each pair is defined in a way such that it gets smaller as the observed covariates of two paired individuals get more similar or the IV doses of two paired individuals get more \textit{distinct} (\citealp{lu2001matching,lu2011optimal}). Therefore, through minimizing the sum of within-pairs distances, non-bipartite matching tends to form pairs that are similar in observed covariates but markedly different in IV doses. After applying non-bipartite matching, any mother can in principle be paired with any other mother with similar observed covariates but different IV dose. After matching, each matched pair consists of one mother who lives ``near'' to a high-level NICU (i.e., excess travel time smaller) and another mother who lives ``far'' from a high-level NICU (i.e., excess travel time larger). Here ``near" and ``far" refer to the relative magnitude of two IV doses in each matched pair: if a mother with excess travel time of 30 minutes is paired with a mother with excess travel time of 45 minutes, then she is in the ``near group". If she is instead paired with another mother with excess travel time of 15 minutes, she is in the ``far" group. See Supplementary Material A.1 for more details on statistical matching.

The first three columns of Table~\ref{tbl: balance table} summarize the covariate balance of matched pairs in the ``near group'' and the ``far group'' of the study. After matching, the covariates are well-balanced with standardized differences near zero. One can then apply permutation-based inferential methods to the matched samples and estimate the treatment effects (\citealp{baiocchi2010building}). This matching-based study design approach to continuous IV analyses is also referred to as ``near/far matching'' (\citealp{baiocchi2012near}), and optimal non-bipartite matching is one of the most widely used statistical matching algorithms to implement it. See \citet{zubizarreta2013stronger} and \citet{keele2016strong,keele2018stronger} for other matching algorithms that implement the near/far matching design. Notable empirical studies that use near/far matching include \citet{lorch2012differential}, \citet{goyal2013length}, \citet{neuman2014anesthesia}, \citet{santana2015cisplatin}, \citet{berkowitz2017supplemental, berkowitz2019association}, \citet{lum2017causal}, and \citet{grieve2019analysis}.

\begin{table}[ht]
\caption{Covariate balance and average excess travel time of the near/far matching design $\mathcal{M}_{0}$ that uses all samples and the strengthening-IV design $\mathcal{M}_{1}$ that uses half of the samples. The means of each covariate in ``near'' and ``far'' groups are reported in columns ``Near'' and ``Far''. Std.dif is an abbreviation for ``absolute standardized difference'', i.e., the absolute value of weighted difference in means divided by the pooled standard deviation between the encouraged and control groups before matching; see \citet{rosenbaum2010design}.}
\begin{tabular}{lcccccc}
\hline
 & \multicolumn{3}{c}{Unstrengthened IV} & \multicolumn{3}{c}{Strengthened IV}\\
 Number of matched pairs  & \multicolumn{3}{c}{95,945} & \multicolumn{3}{c}{47,936}\\
 Estimated compliance rate  & \multicolumn{3}{c}{0.21} & \multicolumn{3}{c}{0.42}\\
& Near & Far & Std.dif & Near & Far  & Std.dif \\ \hline
\multirow{2}{*}{\begin{tabular}[c]{@{}l@{}}Excess travel time to\\ high-level NICU, minutes\end{tabular}} &           &          &    &   &     &         \\ 
 & 5.88      & 21.26    & 0.97  &  2.90 &  38.13  &  2.78    \\ 
Covariates &           &          &          \\ 
\hspace{0.5 cm}Birth weight, g & 2586     & 2585     & 0.00 & 2586  &  2585 &  0.00     \\
\hspace{0.5 cm}Gestational age, weeks & 35.13      & 35.13     & 0.00  &  35.16 &  35.15 & 0.00     \\
\hspace{0.5 cm}Gestational diabetes, 1/0 & 0.05     & 0.05    & 0.00  &  0.03  & 0.04 & 0.04     \\ 
\hspace{0.5 cm}Single birth, 1/0  & 0.83     & 0.83    & 0.00  & 0.86 & 0.85  & 0.02     \\ 
\hspace{0.5 cm}Parity   &2.12       & 2.12      & 0.00  & 2.01 & 2.03  &  0.02 \\
\hspace{0.5 cm}Mother's age, years    &28.06       &28.04      &0.00  & 27.53  &  27.32  &  0.04 \\
\hspace{0.5 cm}Mother's education (scale)    &3.69       &3.69      &0.01  & 3.64  &  3.58 & 0.05 \\
\hspace{0.5 cm}Mother's race (white or not), 1/0   &0.70      &0.71      &0.00  & 0.85 & 0.87 & 0.04 \\
\hspace{0.5 cm}Mother's race missing, 1/0  &0.09   &0.09  &0.00 & 0.06  & 0.07  &  0.02 \\
\hspace{0.5 cm}Neighborhood below poverty (fr)  &0.13   &0.12  &0.04  & 0.12 & 0.11 & 0.11 \\
\hline
\end{tabular}
\label{tbl: balance table}
\end{table}

\subsection{Strength of IVs and problems with weak IVs}
\label{subsec: problems with weak IV}
While encouragement creates incentives for subjects to accept the treatment, some may fail to comply with this encouragement. An encouragement is said to be a strong IV if it is strongly associated with the treatment and weak if it is only weakly associated (\citealp{Bound1995}). When the IV is binary, the strength of the IV is measured by the \emph{compliance rate}, defined to be the proportion of individuals who would accept treatment if encouraged but would not accept treatment if not encouraged (\citealp{AIR1996, Imbens_Rosenbaum2005}).

It is well recognized that studies leveraging weak IVs face several problems. With weak IVs, the usual two-stage least squares (2SLS) estimation method leads to invalid inference (\citealp{staiger1994instrumental}; \citealp{Bound1995}; \citealp{stock2002survey}). This problem can be fixed by using a permutation-based method of inference. Second, even valid, permutation-based inference suffers from low finite sample power and excessively long and non-informative confidence intervals (\citealp{Imbens_Rosenbaum2005}). This problem can be fixed by increasing the sample size. A third, and perhaps more worrisome, problem is that weak IVs are invariably sensitive to small or moderate biases, no matter the sample size (\citealp{staiger1994instrumental}; \citealp{small2008war}; \citealp{ertefaie2018quantitative}). With a weak IV, even a small correlation between the IV and unmeasured confounders greatly increases the bias of IV estimators (\citealp{Bound1995}; \citealp{small2008war}). Moreover, the power of a sensitivity analysis, defined as the probability that a study rejects a false null hypothesis when a specified magnitude of hidden bias in the IV is allowed for (\citealp{rosenbaum2004design, Rosenbuam_heter_causality2005}), drastically decreases for a weak IV even if it is only slightly biased; see \cite{small2008war} and \cite{ertefaie2018quantitative} for a detailed account.

\subsection{A strong IV is good, but strengthening an IV may not be}
Weak IVs can be problematic; a strong IV, when available, is preferable. In many practical situations, however, strong IVs are not available. To overcome this conundrum, many works have proposed to strengthen an existing, possibly weak, IV in the design stage of matched observational studies (e.g., \citealp{baiocchi2010building}; \citealp{zubizarreta2013stronger}; \citealp{yang2014dissonant}; \citealp{keele2016strong,keele2018stronger}; \citealp{lehmann2017strengthening}; \citealp{ertefaie2018quantitative}; \citealp{fogarty2019biased}).

Intuitively, ceteris paribus, a larger dose of encouragement creates stronger incentives for subjects to accept treatment (e.g., living nearer to a high-level NICU creates a stronger incentive to attend a high-level NICU). One strategy that exploits this intuition builds a stronger IV from an existing one by making the average matched pair difference in the IV larger, at the cost of throwing away some of the samples. Examples of methodological development include \citet{baiocchi2010building}, \citet{zubizarreta2013stronger}, \citet{yang2014dissonant}, \citet{keele2016strong, keele2018stronger}, and \citet{ertefaie2018quantitative}. Examples of applications that adopt this strategy include \citet{lorch2012differential}, \citet{goyal2013length}, \citet{neuman2014anesthesia}, \citet{santana2015cisplatin}, and \citet{grieve2019analysis}. To illustrate, we followed \citet{baiocchi2010building} and built a stronger IV by forming only half as many matched pairs as the data allows. The last three columns of Table \ref{tbl: balance table} summarize the covariate balance of this strengthened IV.

We will demonstrate there is a fundamental difference between a naturally strong IV and building a stronger IV. We have four objectives in this article. First, we would like to clarify the potential outcomes framework and the related randomization-based inferential procedures in matching-based study design approaches to IV analysis with continuous IVs. Second, we exhibit when strengthening a valid IV improves the finite sample power of some popular nonparametric test statistics, and when it does not. Third, we discuss what constitutes a valid randomization-based sensitivity analysis that examines how the IV (encouragement) assignment deviating from randomization would materially affect the causal conclusion. A surprising consequence of our discussion is that the widely used $\Gamma$ sensitivity analysis model, sometimes called the Rosenbaum bounds, \emph{cannot} be directly applied in the continuous IV setting. We show that, contrary to previously reported conclusions that strengthening an IV increases the power of a sensitivity analysis in large samples, strengthening an IV may or may not increase the power of a sensitivity analysis under our framework and we illuminate the factors which determine whether it increases the power of a sensitivity analysis. \textsf{R} code necessary to reproduce the main results in this article is available via \url{https://github.com/siyuheng/Code-for-IV-to-Strengthen-or-not-to-Strengthen-}.
\section{Review: IV methods in pair-matched studies}
\label{sec: review of binary and continuous IV}
\subsection{The potential outcomes framework and IV assumptions for a binary IV}\label{subsec: review of potential outcome framework}
\vspace{-0.155 cm}
    We first consider a setting with a binary IV. An example of a study with a binary IV is \citet{bronars1994economic} who, in a study of the effect of out-of-wedlock fertility on labor supply, used whether an unwed mother's first birth was to a singleton baby vs. twins (mothers who gave birth to triplets or more were not considered). Suppose there are $I$ matched pairs, each consisting of $2$ individuals. For a study with a binary IV, individual $j$ ($j=1,2$) in matched pair $i$ ($i=1,\dots, I$) is associated with a binary IV (i.e., an encouragement) $Z_{ij}$, a vector of observed covariates $\mathbf{x}_{ij}$, a binary treatment indicator $D_{ij}$, and an outcome of interest $R_{ij}$. Let $\mathbf{Z}=(Z_{11}, \dots, Z_{I2})$, $\mathbf{D}=(D_{11}, \dots, D_{I2})$ and $\mathbf{R}=(R_{11}, \dots, R_{I2})$ denote the vectors of encouragement assignments, treatments, and outcomes for $2I$ subjects, respectively. Let $\mathcal{Z}$ be the collection of all encouragement indicator vectors $\mathbf{Z}$ such that $Z_{i1}+Z_{i2}=1$ for all $i$. 
    
    We consider the potential outcomes framework (\citealp{neyman1923application}; \citealp{rubin1974estimating}) for the IV setting as in \citet{AIR1996}. Let $D_{ij}(\mathbf{Z})$ be the indicator for whether subject $ij$ would receive the treatment or not if the encouragement assignment vector is set to $\mathbf{Z}$, and $R_{ij}(\mathbf{Z}, \mathbf{D})$ the outcome of $ij$ if the encouragement assignment vector is $\mathbf{Z}$ and the treatment assignment vector is $\mathbf{D}$. We assume the stable unit treatment value assumption (SUTVA): if $Z_{ij}=Z_{ij}^{\prime}$, then $D_{ij}(\mathbf{Z})=D_{ij}(\mathbf{Z}^{\prime})$; if $Z_{ij}=Z_{ij}^{\prime}$ and $D_{ij}=D_{ij}^{\prime}$, then $R_{ij}(\mathbf{Z}, \mathbf{D})=R_{ij}(\mathbf{Z}^{\prime}, \mathbf{D}^{\prime})$. SUTVA says that encouragement affects only the subject being encouraged (no interference among units) and that there are no different versions of the encouragement. In addition to SUTVA, we assume that the exclusion restriction holds: $R_{ij}(\mathbf{Z}, \mathbf{D})=R_{ij}(\mathbf{Z}^{\prime}, \mathbf{D})$ for all $\mathbf{Z}, \mathbf{Z}^{\prime}$ and $\mathbf{D}$. Under SUTVA and the exclusion restriction, the following counterfactuals are well defined: $d_{Tij}$ is the treatment subject $ij$ would have if encouraged, $d_{Cij}$ is the treatment the subject would have if not encouraged, $r_{Tij}$ is the outcome the subject would have if encouraged, and $r_{Cij}$ is the outcome the subject would have if not encouraged. The observed treatment $D_{ij}$ and observed outcome $R_{ij}$ satisfy $D_{ij} = Z_{ij}d_{Tij} + (1-Z_{ij})d_{Cij}$ and $R_{ij} = Z_{ij}r_{Tij} + (1-Z_{ij})r_{Cij}$. Write $\mathcal{F}_{0}=\{(\mathbf{x}_{ij}, d_{Tij}, d_{Cij}, r_{Tij}, r_{Cij}): i=1,\dots, I, j=1,2 \}$.
 
        In addition to SUTVA and the exclusion restriction, there are three more assumptions commonly used in a randomized encouragement design: (1) IV Random Assignment: $P(Z_{i1}=1\mid \mathcal{F}_{0}, \mathcal{Z})=P(Z_{i2}=1\mid \mathcal{F}_{0}, \mathcal{Z})=1/2$ for all $i$, possibly conditional on $\mathbf{x}_{i1}=\mathbf{x}_{i2}$; (2) Positive Correlation (Between IV and Treatment): $E(D\mid Z=1,\mathbf{X}=\mathbf{x})>E(D\mid Z=0, \mathbf{X}=\mathbf{x})$ for all $\mathbf{x}$; (3) Monotonicity: $d_{Tij}\geq d_{Cij}$ for all $i,j$. The IV random assignment assumption can be satisfied if the IV is physically randomly assigned as in a randomized encouragement study or if it is independent of any unmeasured confounders conditional on the covariates $\mathbf{X}$ (so effectively randomly assigned). In an IV analysis, a subject belongs to one of the following four classes: 1) an always-taker if $(d_{Tij}, d_{Cij})=(1,1)$; 2) a complier if $(d_{Tij}, d_{Cij})=(1,0)$; 3) a never-taker if $(d_{Tij}, d_{Cij})=(0,0)$; 4) a defier if $(d_{Tij}, d_{Cij})=(0,1)$. The monotonicity assumption excludes defiers. For a detailed discussion of these assumptions, see \citet{AIR1996} and \citet{Baiocchi_ivtutorial2014}. An IV that satisfies SUTVA, exclusion restriction, IV random assignment, and positive correlation is called a valid IV; otherwise, it is invalid.

\subsection{Randomization inference with a binary IV in pair-matched studies}\label{subsec: randomization inference with a binary IV}
   In a randomization inference for a randomized encouragement design, the only probability distribution that enters the inference is the distribution of encouragement assignments that describes the encouragement assignment mechanism; potential outcomes under encouragement or control are held fixed (\citealp{fisher1937design}; \citealp{Imbens_Rosenbaum2005}). A randomization-based inferential procedure can be conducted by looking at the probability that a test statistic $T$ is greater than or equal to the observed value $t$:
\begin{equation*}
	P(T \geq t\mid \mathcal{F}_{0}, \mathcal{Z})=\sum_{\mathbf{z}\in \mathcal{Z}}\mathbbm{1}(T(\mathbf{z}, \mathbf{R}) \geq t) \cdot P(\mathbf{Z}=\mathbf{z}\mid \mathcal{F}_{0}, \mathcal{Z})=\frac{|\{\mathbf{z}\in \mathcal{Z}: T(\mathbf{z}, \mathbf{R}) \geq t\}|}{|\mathcal{Z}|},
\end{equation*}
where $T$ is a statistic that tests a null hypothesis concerning the counterfactuals $(d_{Tij}, d_{Cij}$, $r_{Tij}, r_{Cij})$, and $P(\mathbf{Z}=\mathbf{z}\mid \mathcal{F}_{0}, \mathcal{Z})=1/|\mathcal{Z}|=1/2^{I}$ for all $\mathbf{z} \in \mathcal{Z}$ with a valid binary IV.

\citet{small2008war} considered testing the null hypothesis $H_0: \beta = \beta_0$ in the following model: 
\begin{equation}\label{model: prop model}
    r_{Tij} - r_{Cij} = \beta (d_{Tij} - d_{Cij}).
\end{equation}
Model (\ref{model: prop model}) is called a proportional treatment effect model because the effect of the encouragement on the outcome is proportional to its effect on the treatment. Under model (\ref{model: prop model}), we have $r_{Tij}-\beta d_{Tij}=r_{Cij}-\beta d_{Cij}\overset{\Delta}{=}\epsilon_{ij}$. Let $Y_{i}=(Z_{i1}-Z_{i2})(R_{i1}-R_{i2})$ and $S_{i}=(Z_{i1}-Z_{i2})(D_{i1}-D_{i2})$. We have $Y_{i}-\beta S_{i} = (Z_{i1}-Z_{i2})(\epsilon_{i1}-\epsilon_{i2})\overset{\Delta}{=} \epsilon_{i}$. Under the IV random assignment assumption, $\epsilon_{i}$ takes on value $|\epsilon_{i1}-\epsilon_{i2}|$ or $-|\epsilon_{i1}-\epsilon_{i2}|$ with equal probability $1/2$, and in order to test $H_{0}: \beta=\beta_{0}$, it suffices to test whether $Y_{i}-\beta_{0} S_{i}=\epsilon_{i}+(\beta-\beta_{0})S_{i}$ is symmetrically distributed about zero, which can be done, for instance, by applying the Wilcoxon signed rank test or the sign test to $Y_{i}-\beta_{0}S_{i}$ (\citealp{lehmann2004elements}). Randomization-based inferential procedures could also use the randomization distribution of a combined quantile average (\citealp{rosenbaum1999using}) or the sample mean of $Y_{i}-\beta_{0} S_{i}$ (\citealp{Imbens_Rosenbaum2005}), yielding a Wald estimator $\widehat{\beta}_{\text{IV}}$ (\citealp{wald1940fitting}; \citealp{small2008war}) through solving the estimating equation $\mathbb{E}(Y_{i}-\beta S_{i})=(\mathbb{E}[R_{ij}\mid Z_{ij}=1]-\mathbb{E}[R_{ij}\mid Z_{ij}=0])-\beta(\mathbb{E}[D_{ij}\mid Z_{ij}=1]-\mathbb{E}[D_{ij}\mid Z_{ij}=0])=0$:
\begin{align*}
  &  \widehat{\beta}_{\text{IV}} = \frac{\widehat{\mathbb{E}}[\ R_{ij}\mid Z_{ij}=1\ ]-\widehat{\mathbb{E}}[\ R_{ij}\mid Z_{ij}=0\ ]}{\widehat{\mathbb{E}}[\ D_{ij}\mid Z_{ij}=1\ ]-\widehat{\mathbb{E}}[\ D_{ij}\mid Z_{ij}=0 \ ]}= \frac{\sum_{i=1}^{I}(Z_{i1}-Z_{i2})(R_{i1}-R_{i2})}{\sum_{i=1}^{I}(Z_{i1}-Z_{i2})(D_{i1}-D_{i2})}.
\end{align*}
More recently, \citet{baiocchi2010building} developed a randomization-based inferential procedure for IV analysis that allows for heterogeneous treatment effects in pair-matched studies, and \citet{kang2016full} further generalized it to the full matching design, allowing for matching with different numbers of controls across matched strata. 

\subsection{Near/far matching: embedding continuous IVs into pair-matched studies}
\label{subsec: review of continuous IV}
The NICU study described in Section~\ref{sec: introduction} involves a continuous IV, excess travel time. We clarify how to adapt the randomization-based IV analysis for binary IVs to the setting of continuous IVs in this section. A matching-based study design approach to IV analysis (e.g., near/far matching) typically pairs subjects with similar observed covariates $\mathbf{X}$, but different continuous IVs $\widetilde{Z}$. For instance, in the NICU study, two mothers with similar socioeconomic and health status, but different excess travel times to the nearest high-level NICU are paired together. Let $\widetilde{Z}_{ij}$ denote the continuous IV dose of individual $j$ in matched pair $i$. In each matched pair $i$, a binary encouragement indicator $Z_{ij}$ is constructed out of $\widetilde{Z}_{ij}$: $Z_{i1}=1$ (encouraged) and $Z_{i2}=0$ (not encouraged, i.e., control) if $\widetilde{Z}_{i1} < \widetilde{Z}_{i2}$ and vice versa. As emphasized in Section~\ref{subsec: encouragement design example}, the binary encouragement indicator $Z_{ij}$ is constructed based on the relative magnitude of $\widetilde{Z}_{i1}$ and $\widetilde{Z}_{i2}$ after matching, not according to some prespecified cut-off before matching. In the NICU study, the continuous IV $\widetilde{Z}$ is the excess travel time to the nearest high-level NICU, and the individual with a \textit{smaller} $\widetilde{Z}$ in each matched pair is seen as being encouraged to deliver at a hospital with a high-level NICU and assigned $Z = 1$.

To develop a valid randomization-based inference, one needs to carefully define the encouragement being manipulated/randomized and the fixed potential outcomes under encouragement/control. In their original paper that leveraged excess travel time as an IV to study high-level NICU's effect on neonatal death, \citet{baiocchi2010building} asked:
\begin{quote}
    What would have happened to a mother and her newborn had she lived either close to or far from a high-level NICU? Here there are two responses, $(r_{Tij}, r_{Cij})$ or $(d_{Tij}, d_{Cij})$, where $r_{Tij}$ and $d_{Tij}$ are observed from the $j$-th subject in pair $i$ under [encouragement], $Z_{ij} = 1$, whereas $r_{Cij}$ and $d_{Cij}$ are observed from the $j$-th subject in pair $i$ under control [no encouragement], $Z_{ij} = 0$.
\end{quote}
This notion of encouragement and control has been adopted in much subsequent work (\citealp{baiocchi2010building}; \citealp{zubizarreta2013stronger}; \citealp{keele2016strong,keele2018stronger}; \citealp{lehmann2017strengthening}; \citealp{ertefaie2018quantitative}; \citealp{fogarty2019biased}). One way to interpret this statement is the following: for each subject $ij$, the encouragement that may be manipulated/randomized is the binary encouragement indicator $Z_{ij} = 1$, corresponding to the subject living close to a high-level NICU, and $Z_{ij} = 0$, corresponding to the subject living far from a high-level NICU. However, the potential outcomes $(d_{Tij}, d_{Cij}, r_{Tij}, r_{Cij})$ are not fixed under this definition of encouragement and control. For example, say mother $ij$ is paired to a mother who has an excess travel time of 30 minutes. Consider whether mother $ij$ has an excess travel time of 5 or 15 minutes. In either case, mother $ij$ will be the encouraged mother in the pair. Mother $ij$ might be willing to travel 5 more minutes to go to a high-level NICU but not 15 more minutes. Say mother $ij$ would have a longer length of stay if she went to a high-level NICU rather than a low-level NICU. Then SUTVA (\citealp{rubin1980discussion}) is violated because there are versions of the encouragement which lead to different potential outcomes.

To have well-defined potential outcomes and SUTVA hold, the encouragement (control) assignment in a pair needs to define an exact dose of the IV that is received, e.g., exact excess travel time in the NICU study.  This can be achieved by considering a randomized encouragement assignment in which in pair $i$, if the excess travel times are $\widetilde{Z}_{i1}$, $\widetilde{Z}_{i2}$, a fair coin is flipped and if the coin lands heads, subject $i1$ is assigned IV $\widetilde{Z}_{i1}$ and subject $i2$ is assigned IV $\widetilde{Z}_{i2}$ and vice versa if the coin lands tails. Let $\widetilde{\mathbf{Z}}=(\widetilde{Z}_{11}. \dots, \widetilde{Z}_{I2})$ be the dose assignment (i.e., the continuous IV) vector, $D_{ij}(\widetilde{\mathbf{Z}})$ the indicator for whether subject $ij$ receives the treatment or not given $\widetilde{\mathbf{Z}}$, and $R_{ij}(\widetilde{\mathbf{Z}}, \mathbf{D})$ the outcome of subject $ij$ under $\widetilde{\mathbf{Z}}$ and $\mathbf{D}$. Similar to the binary IV case, SUTVA and the exclusion restriction (\citealp{holland1988causal}; \citealp{rosenbaum1989sensitivity}) state that: (1). $\widetilde{Z}_{ij}=\widetilde{Z}_{ij}^{\prime}$ implies $D_{ij}(\widetilde{\mathbf{Z}})=D_{ij}(\widetilde{\mathbf{Z}}^{\prime})$, and $\widetilde{Z}_{ij}=\widetilde{Z}_{ij}^{\prime}$ and $D_{ij}=D_{ij}^{\prime}$ together imply $R_{ij}(\widetilde{\mathbf{Z}}, \mathbf{D})=R_{ij}(\widetilde{\mathbf{Z}}^{\prime}, \mathbf{D}^{\prime})$; (2). $R_{ij}(\widetilde{\mathbf{Z}}, \mathbf{D})=R_{ij}(\widetilde{\mathbf{Z}}^{\prime}, \mathbf{D})$ for all $\widetilde{\mathbf{Z}}, \widetilde{\mathbf{Z}}^{\prime}$ and $\mathbf{D}$. Let $\widetilde{\mathbf{Z}}_{\vee}=(\widetilde{Z}_{11}\vee \widetilde{Z}_{12}, \widetilde{Z}_{21}\vee \widetilde{Z}_{22}, \dots, \widetilde{Z}_{I1}\vee \widetilde{Z}_{I2})$ and $\widetilde{\mathbf{Z}}_{\wedge}=(\widetilde{Z}_{11}\wedge \widetilde{Z}_{12}, \widetilde{Z}_{21}\wedge \widetilde{Z}_{22}, \dots, \widetilde{Z}_{I1}\wedge \widetilde{Z}_{I2})$, where $a \vee b= \max(a,b)$ and $a \wedge b = \min(a,b)$, represent the maximum and minimum of two doses in each matched pair. Under SUTVA and the exclusion restriction assumption, potential outcomes $(d_{Tij}, d_{Cij}, r_{Tij}, r_{Cij})$ are now well-defined and fixed after matching:
\begin{equation}
\label{def: potential outcome with continuous IV}
\begin{split}
    &d_{Tij}\overset{\Delta}{=} (D_{ij} \mid \widetilde{Z}_{ij}=\widetilde{Z}_{ij}\wedge \widetilde{Z}_{ij^{\prime}}), \quad d_{Cij}\overset{\Delta}{=} (D_{ij} \mid \widetilde{Z}_{ij}=\widetilde{Z}_{ij}\vee \widetilde{Z}_{ij^{\prime}}), \\
    &r_{Tij}\overset{\Delta}{=} (R_{ij} \mid \widetilde{Z}_{ij}=\widetilde{Z}_{ij}\wedge \widetilde{Z}_{ij^{\prime}}), \quad r_{Cij}\overset{\Delta}{=} (R_{ij} \mid \widetilde{Z}_{ij}=\widetilde{Z}_{ij}\vee \widetilde{Z}_{ij^{\prime}}),
\end{split}
\end{equation}
for $j \neq j^{\prime}$. That is, instead of defining the potential outcomes of subject $ij$ with respect to whether $Z_{ij}=\mathbbm{1}(\widetilde{Z}_{ij}<\widetilde{Z}_{ij^{\prime}}) = 1$ or $0$, which violates SUTVA as discussed above, we define potential outcomes with respect to the continuous dose. We still write $\mathcal{F}_{0}=\{(\mathbf{x}_{ij}, d_{Tij}, d_{Cij}, r_{Tij}, r_{Cij}): i=1,\dots, I, j=1,2 \}$, where $d_{Tij}$, $d_{Cij}$, $r_{Tij}$, $r_{Cij}$ are defined as in (\ref{def: potential outcome with continuous IV}). In a randomization inference with the continuous IV $\widetilde{Z}$, the only probability distribution that enters statistical inference is the conditional probability $P(\widetilde{Z}_{i1} = \widetilde{Z}_{i1} \wedge \widetilde{Z}_{i2}, \widetilde{Z}_{i2} = \widetilde{Z}_{i1}\vee \widetilde{Z}_{i2}\mid \mathcal{F}_{0}, \widetilde{\mathbf{Z}}_{\vee}, \widetilde{\mathbf{Z}}_{\wedge})$ that characterizes the underlying dose assignment mechanism in each matched pair $i$. By conditioning on $\widetilde{\mathbf{Z}}_{\vee}$ and $\widetilde{\mathbf{Z}}_{\wedge}$, we only know the maximum and minimum of $(\widetilde{Z}_{i1}, \widetilde{Z}_{i2})$ for each $i$, but not which individual has which value. When $\widetilde{Z}$ is a valid IV that is physically randomized as in a randomized encouragement design, then this dose assignment mechanism is created by the experimenter and known to us: $P(\widetilde{Z}_{i1} = \widetilde{Z}_{i1} \wedge \widetilde{Z}_{i2}, \widetilde{Z}_{i2} = \widetilde{Z}_{i1}\vee \widetilde{Z}_{i2}\mid \mathcal{F}_{0}, \widetilde{\mathbf{Z}}_{\vee}, \widetilde{\mathbf{Z}}_{\wedge}) = P(\widetilde{Z}_{i1} = \widetilde{Z}_{i1} \vee \widetilde{Z}_{i2}, \widetilde{Z}_{i2} = \widetilde{Z}_{i1}\wedge \widetilde{Z}_{i2}\mid \mathcal{F}_{0}, \widetilde{\mathbf{Z}}_{\vee}, \widetilde{\mathbf{Z}}_{\wedge}) = 1/2$, and forms the ``reasoned basis for inference'' in Fisher's phrase. 

\section{Strengthening a valid IV: strength versus sample size}
\label{sec: strengthen valid IV}
 \citet{small2008war} leveraged a matching-based study design approach to reanalyze a study by \citet{angrist1994world} concerning the effects of military services during World War II on subsequent earnings, in which men from different birth cohorts who are otherwise similar are ``encouraged'' or ``discouraged'' to serve in the military depending on when they turned $18$, e.g., men who turned $18$ in $1944$ were more encouraged than men who turned $18$ in $1946$. By selecting men from different birth cohorts, IVs of different strength and with unequal sample sizes can be created. Such a trade-off between IV strength and sample size also naturally emerges when part of the sample is discarded in order to forge a stronger IV in the design stage of matched observational studies. Given valid IVs of different sample sizes and distinct strength, which should one prefer in the context of nonparametric testing in matched studies? 
 
\subsection{Quantitative evaluation of the trade-off between strength and sample size for a valid IV}\label{subsec: Quantitative evaluation of the trade-off between IV strength and sample size}
We consider testing the null hypothesis $H_0: \beta = \beta_0$ in model (\ref{model: prop model}) in a matched pair study using the Wilcoxon signed rank test or the sign test as introduced in Section~\ref{sec: review of binary and continuous IV}. Recall that the error term $\epsilon_{i}=Y_{i}-\beta S_{i}$ takes on value $|\epsilon_{i1}-\epsilon_{i2}|$ or $-|\epsilon_{i1}-\epsilon_{i2}|$ with equal probability. We assume that $\epsilon_{i}, i = 1,...,I$ are i.i.d. realizations from a distribution $F$ where $F$ is symmetric about zero given that the IV is valid. Following \citet{ertefaie2018quantitative}, we consider a model that says an individual is an always-taker with probability $\iota_A$, a complier with probability $\iota_C$ (i.e., the compliance rate), and a never-taker with probability $\iota_N$. We have $\iota_A + \iota_N + \iota_C = 1$ as there are no defiers under the monotonicity assumption. Denote the mean of $r_{Cij}$ for compliers to be $\mu_C$, for always-takers $\mu_A$, and for never-takers $\mu_N$. We assume that the compliance categories of different individuals are independent, and the distribution of $r_{Cij}$ for always-takers and that for never-takers are both location shifts from that of compliers. Let $I$ be the number of matched pairs. We now derive the asymptotic relative efficiency (ARE) of testing $H_0: \beta = \beta_0$ in model (\ref{model: prop model}) when there are two valid IVs of different strength.

\begin{Theorem}[Quantitative trade-off between IV strength and sample size] \label{thm: ARE in independent case} 
Consider a sequence of testing problems consisting of a null hypothesis $H_{0}: \beta=\beta_{0}$ versus $H_{1}=\beta_{n}$. Suppose that $\beta_{n}=\beta_{0}+\Delta/\sqrt{n} + o(1/\sqrt{n})$ and without loss of generality, suppose that $\Delta>0$. Let $\pi_{I}(\cdot)$ be the power function of the one-sided Wilcoxon signed rank test (or the sign test) when testing the proportional treatment effect model in a randomized encouragement design using $I$ matched pairs. Let $I_{n}$ be the minimum number of matched pairs needed such that $\pi_{I_{n}}(\beta_0)\leq \alpha$ and $\pi_{I_{n}}(\beta_n)\geq \gamma$ for some $\alpha \in (0,1)$ and $\gamma \in (\alpha,1)$. Let $I_{n,1}$ and $I_{n,2}$ be the number of matched pairs needed for two given sequences of tests using two different valid IVs: one with parameters $\varrho_{1}=(\iota_{C,1},\iota_{A,1},\iota_{N,1},\mu_{C,1}, \mu_{A,1}, \mu_{N,1})$ and the other with parameters $\varrho_{2}=(\iota_{C,2},\iota_{A,2},\iota_{N,2},\mu_{C,2}, \mu_{A,2}, \mu_{N,2})$. Suppose that the compliance structure $(d_{Tij},d_{Cij})$ is independent of $r_{Cij}$ for both IVs, implying that $\mu_{C,1} = \mu_{A,1} = \mu_{N,1}$ and $\mu_{C,2} = \mu_{A,2} = \mu_{N,2}$. Let $f(x)$ be the density function of the error term $\epsilon_{i}$. Suppose that $f$ is continuously differentiable, $\underset{x \to \infty}{\lim} f(x)=0$, and $f \in L^{2}$. For both the Wilcoxon signed rank test and the sign test, we have
\begin{equation*}
    \lim_{n \rightarrow \infty}\frac{I_{n,1}}{I_{n,2}}=\frac{\iota_{c,2}^{2}}{\iota_{c,1}^{2}}.
\end{equation*}
\end{Theorem}

Proofs of all theorems in this paper are in Supplementary Material B.

For a more general version of Theorem~\ref{thm: ARE in independent case} where the compliance structure $(d_{Tij},d_{Cij})$ can be correlated with $r_{Cij}$ for both IVs, see Theorem~\ref{thm: complete ARE in general case} in Supplementary Material B.1.

Theorem~\ref{thm: ARE in independent case} helps explain some empirical findings in the literature. For instance, \citet{small2008war} found the $95\%$ confidence interval obtained from inverting the Wilcoxon signed rank test based on $14,000$ pairs of veterans born in the $1926/1928$ cohort (estimated compliance rate $0.51$) is $[-1,445, -500]$, which is shorter than that based on $28,000$ pairs of veterans born in the $1924/1926/1928$ cohort (estimated compliance rate $0.27$) (95\% CI: [-1,524, -261]). This result is not surprising in light of Theorem~\ref{thm: ARE in independent case}: when we compare two IVs in matched studies, without any additional information, we would prefer the one with significantly larger
\begin{equation*}
    \text{``effective sample size"}= \text{sample size}\times (\text{compliance rate})^{2}. 
\end{equation*}
The effective sample size of birth time for the $1926/1928$ cohort to participate in the Vietnam War is $14,000\times 0.51^{2} \approx 3,641$, which is much larger than that for the $1924/1926/1928$ cohort ($28,000\times 0.27^{2}\approx 2,041$). Therefore, testing with the $1926/1928$ cohort should be more powerful than testing with the $1924/1926/1928$ cohort. Theorem~\ref{thm: ARE in independent case} also resonates with similar results in the two-stage least squares (2SLS) literature: the variance of the 2SLS estimator obtained from $N$ samples is at least as large as the variance from having a sample of $N \times \iota_C^2$ (effective sample size) known compliers, and the standard error, or the length of the confidence interval, at least scales proportional to $\sqrt{N}\times \iota_C$ (\citealp{Imbens_CACE1994}; \citealp{Baiocchi_ivtutorial2014}).


\subsection{Simulations}\label{subsec: Simulations for ARE}
We assess how accurately the asymptotic result in Theorem~\ref{thm: ARE in independent case} captures the trade-off between strength and sample size for a valid IV under fixed alternatives via simulations. For a fixed IV strength $\iota_C$, fixed proportions of always-takers and never-takers, and a fixed $\alpha$ level, we use Monte Carlo simulations to determine the sample size needed to achieve a specified power $\gamma$ when the effect size is $\beta - \beta_0 = 0.1$. We use binary search to search for the sample size and $20,000$ simulated datasets to estimate the power for each sample size. Specifically, we consider three different pairs of IVs: a pair of IVs that are slightly different in strength with $(\iota_{C, 1} = 0.5, \iota_{C, 2} = 0.6)$, a pair moderately different with $(\iota_{C, 1} = 0.4, \iota_{C, 2} = 0.7)$, and a pair vastly different with $(\iota_{C, 1} = 0.3, \iota_{C, 2} = 0.8)$. Table \ref{tbl: ARE simulation} summarizes the sample size ratio determined via simulation, and contrasts it to the ARE when the error is normal or Laplace, and the level is set to be $0.05$. We see simulation results agree very well with ARE, with relative error less than $1\%$, and the sample size ratio is independent of nuisance parameters $\iota_A$, $\iota_N$, and error distribution. Figure \ref{fig: power pair 1 normal}, \ref{fig: power pair 2 normal}, and \ref{fig: power pair 3 normal} plot power against sample size for three pairs of IVs when the error is normal, and \ref{fig: normal sample size ratio} plots the ratio of sample sizes needed for the stronger and weaker IVs in each pair to attain the same power. The ratio of needed sample size agrees very well with the theoretical value $\iota^2_{C, 2}/\iota^2_{C, 1}$ (three background lines in Figure~\ref{fig: normal sample size ratio}) for all cases. Similar plots when the error is Laplace can be found in Supplementary Material C.1.

\begin{table}[ht]
\caption{Comparing simulated sample size ratio to ARE using the (one-sided) Wilcoxon signed rank test when compliance $(d_{Tij}, d_{Cij})$ is independent of response under control $r_{Cij}$. The errors follow standard normal distribution, $\alpha = 0.05$, $\beta-\beta_{0}=0.1$, and $\iota_N = 1 - \iota_A - \iota_C$. We repeat the simulation $20,000$ times to approximate the power. `SIM' and `THEO' denote the simulated and theoretical ratios of the two required sample sizes respectively.}
\label{tbl: example table one}
\begin{tabular}{cccccccccccccc}
\hline 
&\multicolumn{2}{c}{Parameters} & & \multicolumn{3}{c}{Normal error} & &\multicolumn{3}{c}{Laplace error} \\ 
\cline{5-7} \cline{9-11} 
&Power  &$\iota_A$ && $\iota_C = 0.5$ & $\iota_C = 0.6$ & SIM & & $\iota_C = 0.5$ & $\iota_C = 0.6$ & SIM & THEO\\  \cline{2-3} 

&\multirow{2}{*}{0.8} &0   && 2590 & 1790 & 1.45 &  &1670  &1160  &1.44  & 1.44 \\ 
&&$(1 - \iota_C)/2$  && 2596 & 1803 & 1.44 &  &1675  &1177  &1.42  & 1.44\\ [0.3em] \cline{2-3} 

&\multirow{2}{*}{0.7} &0  &&1980  &1370   &1.45  & &1265  &890  &1.42  &1.44\\ 
&&$(1 - \iota_C)/2$ &&1972  & 1378 &1.43  & &1261 &868  &1.45 &1.44\\ [0.3em] \cline{2-3} 

&\multirow{2}{*}{0.6} &0  && 1510 & 1047 & 1.44  &  &978  &677 &1.44   & 1.44\\ 
&&$(1 - \iota_C)/2$ && 1517 & 1046  & 1.45  &  &974 &683  &1.43   & 1.44\\ 

\hline 
&Power  &$\iota_A$ && $\iota_C = 0.4$ & $\iota_C = 0.7$ & SIM &  & $\iota_C = 0.4$ & $\iota_C = 0.7$ & SIM & THEO\\  \cline{2-3} 

&\multirow{2}{*}{0.8} &0   &&4071  &1328  &3.07  &  &2595  & 838 &3.10  & 3.06\\ 
&&$(1 - \iota_C)/2$  &&4036  & 1310 &3.08  &  &2577  &844  &3.05  & 3.06\\ [0.3em] \cline{2-3} 

&\multirow{2}{*}{0.7} &0  &&3092  &1014  &3.05  & &1956  &647 &3.02 &3.06\\ 
&&$(1 - \iota_C)/2$ &&3096 &1019  &3.04 & &1977 &644 &3.07 &3.06\\ [0.3em] \cline{2-3} 

&\multirow{2}{*}{0.6} &0  &&2347  &769  &3.05 &  &1538 &502  &3.06  & 3.06\\ 
&&$(1 - \iota_C)/2$ &&2355  &769  &3.06 &  &1510  &491  &3.08   & 3.06\\  \hline 

&Power  &$\iota_A$ && $\iota_C = 0.3$ & $\iota_C = 0.8$ & SIM & & $\iota_C = 0.3$ & $\iota_C = 0.8$ & SIM & THEO\\  \cline{2-3} 

&\multirow{2}{*}{0.8} &0   &&7246  &1009  &7.17  &  &4594  &648  &7.09  & 7.11\\ 
&&$(1 - \iota_C)/2$  &&7210  &1015  &7.10  &  &4590  &645  &7.12  & 7.11\\ [0.3em] \cline{2-3} 

&\multirow{2}{*}{0.7} &0  &&5515  &775 &7.12  & &3500  &493  &7.10  &7.11\\ 
&&$(1 - \iota_C)/2$ &&5500 &775 &7.10  & &3542 &496 &7.14  &7.11\\ [0.3em] \cline{2-3} 

&\multirow{2}{*}{0.6} &0  &&4216  &590  &7.15 &  &2683  &379  &7.08  & 7.11\\ 
&&$(1 - \iota_C)/2$ && 4225 &595   &7.10   &  &2700  &383  &7.05  & 7.11\\  \hline

\end{tabular}
\label{tbl: ARE simulation}
\end{table}

\begin{figure}[ht]
\caption{ Panels (a), (b), and (c): power against sample size for three pairs of IVs with different strength: $\beta - \beta_0 = 0.1$, normal error, $\alpha = 0.05$. Panel (d): sample sizes needed to obtain a fixed power for the stronger and weaker IV in each pair. Lines with slopes equal to $\iota^2_{C, 2}/\iota^2_{C, 1}$ are imposed.}
\label{fig: two pics normal error}
\begin{subfigure}{.5\linewidth}
\centering
\includegraphics[width = 6.5 cm, height = 5 cm]{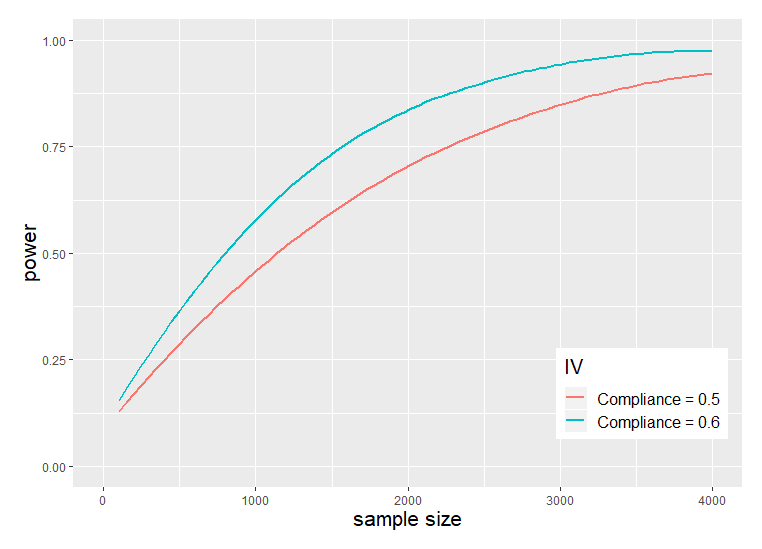}
\caption{Pair 1: $(\iota_{C, 1} = 0.5, ~\iota_{C, 2} = 0.6)$ }\label{fig: power pair 1 normal}
\end{subfigure}%
\begin{subfigure}{.5\linewidth}
\centering
\includegraphics[width = 6.5 cm, height = 5 cm]{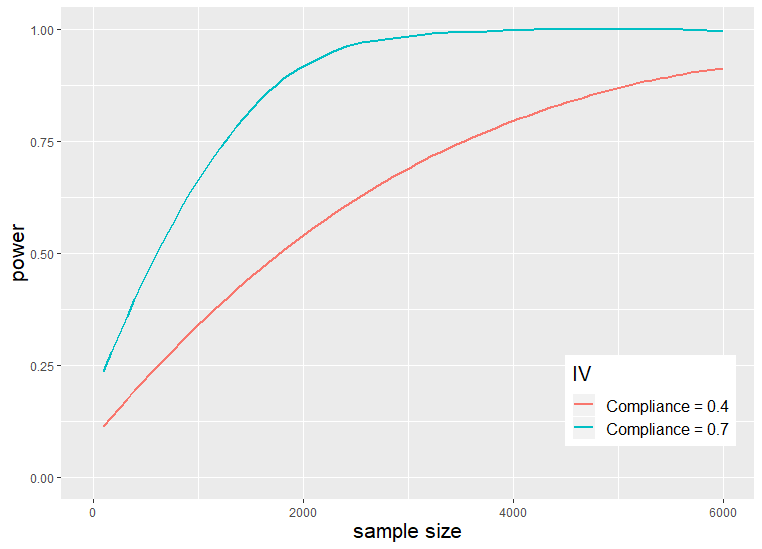}
\caption{Pair 2: $(\iota_{C, 1} = 0.4, ~\iota_{C, 2} = 0.7)$ }\label{fig: power pair 2 normal}
\end{subfigure}\\[1ex]
\begin{subfigure}{.5\linewidth}
\centering
\includegraphics[width = 6.5 cm, height = 5 cm]{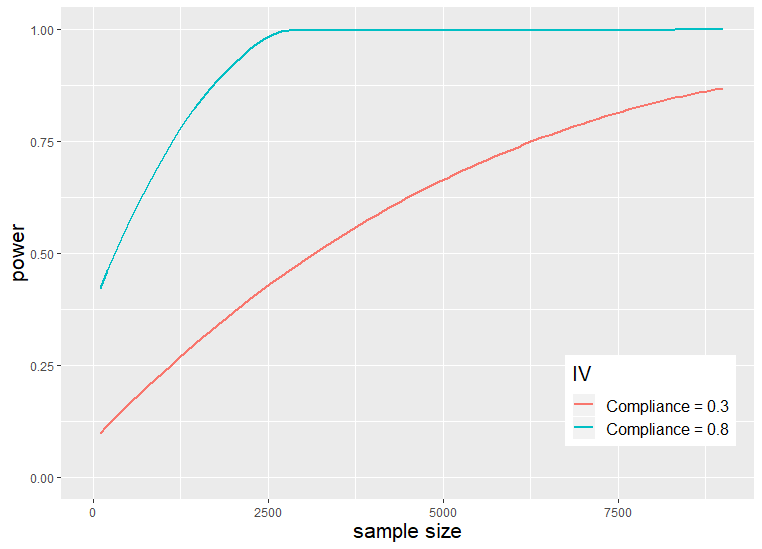}
\caption{Pair 3: $(\iota_{C, 1} = 0.3, ~\iota_{C, 2} = 0.8)$ }\label{fig: power pair 3 normal}
\end{subfigure}%
\begin{subfigure}{.5\linewidth}
\centering
\includegraphics[width = 7 cm, height = 5 cm]{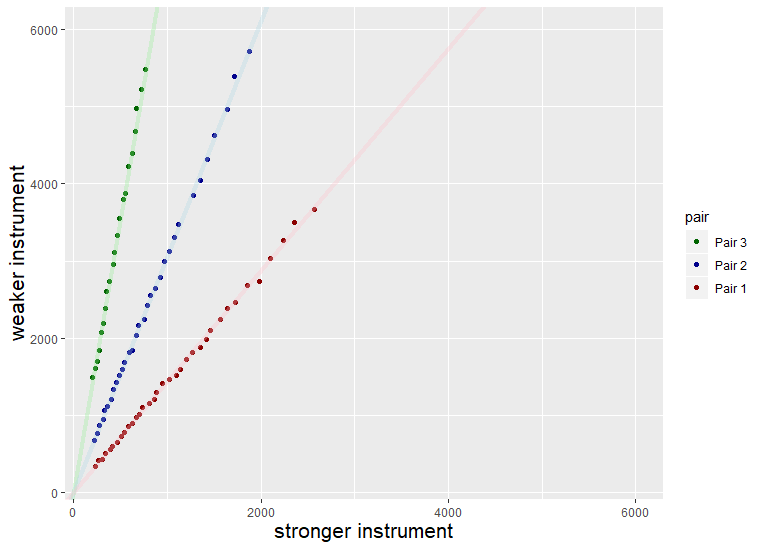}
\caption{Sample size ratio}\label{fig: normal sample size ratio}
\end{subfigure}
\end{figure}

\section{What constitutes a valid sensitivity analysis with continuous IVs?}
\label{sec: limitation of Gamma and AOB}
When an IV is effectively randomized, then when asking whether we should strengthen it, the key question is how the power compares before and after the IV is strengthened; this was the question we considered in Section \ref{sec: strengthen valid IV}. In observational studies, however, it is often unrealistic to assume the encouragement is randomized, even after a set of observed covariates are controlled for. A sensitivity analysis in matched observational studies asks how a departure from random assignment of encouragement would affect the causal conclusion drawn from a primary analysis that assumes the encouragement is effectively randomized. In this section, we review the $\Gamma$ sensitivity analysis model developed for a binary IV, and discuss what constitutes a valid model for describing the biased dose assignment mechanism when IVs are continuous. A rather surprising consequence of our result is that the $\Gamma$ sensitivity analysis model, despite being used extensively in the literature (e.g., \citealp{baiocchi2010building}; \citealp{zubizarreta2013stronger}; \citealp{keele2016strong}; \citealp{ertefaie2018quantitative}), is \emph{not} a valid model for biased dose assignment mechanism when IVs are continuous.

\subsection{Review: the $\Gamma$ sensitivity analysis model for a binary IV}
\label{subsec: randomization and Gamma}

In a randomized encouragement experiment with a binary encouragement, we have $P(\mathbf{Z}=\mathbf{z}\mid \mathcal{F}_{0}, \mathcal{Z})=1/|\mathcal{Z}|=1/2^{I}$ for all $\mathbf{z} \in \mathcal{Z}$. The model of \citet[Chapter~4]{rosenbaum2002observational} quantifies how failing to match on a potential unmeasured confounder $u_{ij}$ in each matched pair would bias the encouragement assignment probability and has been extensively used in the causal inference literature as a sensitivity analysis framework. Let $\mathcal{F}_{1}=\mathcal{F}_{0}\cup \{u_{ij}: i=1,\dots, I, j=1,2 \}=\{(\mathbf{x}_{ij}, u_{ij}, d_{Tij}, d_{Cij}, r_{Tij}, r_{Cij}): i=1,\dots, I, j=1,2 \}$ and $\pi_{ij}=P(Z_{ij}=1\mid \mathcal{F}_{1})$ denote the probability that individual $j$ in matched pair $i$ receives the encouragement in the possible presence of $u_{ij}$. The Rosenbaum-type sensitivity analysis considers the following logit model of $\pi_{ij}$ (\citealp{rosenbaum2002observational}):
\begin{equation} \label{eqn: rosenbaum sens model}
    \log \left(\frac{\pi_{ij}}{1-\pi_{ij}} \right)=\theta(\mathbf{x}_{ij})+\gamma u_{ij}, \quad \text{where $u_{ij}\in [0,1]$,}
\end{equation}
where $\theta(\mathbf{x}_{ij})$ is an arbitrary unknown function of $\mathbf{x}_{ij}$ and $\gamma\geq 0$ is a sensitivity parameter. Note that the constraint $u_{ij}\in [0,1]$ is no more restrictive than assuming a bounded support of $u_{ij}$ and is only imposed to make the sensitivity parameter $\gamma$ more interpretable. It is straightforward to show that under model (\ref{eqn: rosenbaum sens model}), for two individuals $i1$ and $i2$ in the same matched pair $i$ with $\mathbf{x}_{i1}=\mathbf{x}_{i2}$, the ratio of their odds of receiving the encouragement is bounded by $\Gamma=\exp(\gamma)\geq 1$:
\begin{equation*} 
	\frac{1}{\Gamma}\leq \frac{\pi_{i1}(1-\pi_{i2})}{\pi_{i2}(1-\pi_{i1})}\leq \Gamma, \quad \text{for all $i$} ~\text{with $\mathbf{x}_{i1}=\mathbf{x}_{i2}$}.
\end{equation*}
When the encouragement $\mathbf{Z}$ is randomized in each pair, we have $\Gamma = 1$. The more $\Gamma$ deviates from $1$, the more encouragement assignment potentially deviates from randomization. Assuming that the treatment assignment is independent across matched pairs, model (\ref{eqn: rosenbaum sens model}) implies the following probability distribution of the vector of encouragement assignment indicators (assuming $\mathbf{x}_{i1} = \mathbf{x}_{i2}$) with $\gamma=\log(\Gamma)$: 
\begin{equation}
	P(\mathbf{Z}=\mathbf{z}\mid \mathcal{F}_{1}, \mathcal{Z})=\prod_{i=1}^{I}\frac{\exp(\gamma \sum_{j=1}^{2}z_{ij}u_{ij})}{\sum_{j=1}^{2}\exp(\gamma u_{ij})}, \quad \mathbf{z} \in \mathcal{Z}, \ 0 \leq u_{ij} \leq 1.
	\label{eqn: sens model gamma equiv}
\end{equation}
Model (\ref{eqn: rosenbaum sens model}) along with its implication (\ref{eqn: sens model gamma equiv}) are known as the $\Gamma$ sensitivity analysis model. Under the null hypothesis $H_0$ and (\ref{eqn: sens model gamma equiv}), the permutation distribution of the test statistic becomes
\begin{equation*}
   P(T \geq t \mid \mathcal{F}_{1}, \mathcal{Z})=\sum_{\mathbf{z}\in \mathcal{Z}}\mathbbm{1}(T(\mathbf{z}, \mathbf{R}) \geq t) \cdot \prod_{i=1}^{I}\frac{\exp(\gamma \sum_{j=1}^{2}z_{ij}u_{ij})}{\sum_{j=1}^{2}\exp(\gamma u_{ij})}.
\end{equation*}
In a sensitivity analysis for a one-sided test with $\Gamma=\exp(\gamma)$, researchers are interested in the ``worst-case'' p-value reported by a test statistic $T$ given its observed value $t$:
\begin{equation*}
\max_{0 \leq u_{ij}\leq 1}P(T \geq t \mid \mathcal{F}_{1}, \mathcal{Z})=\max_{0 \leq u_{ij}\leq 1}\sum_{\mathbf{z}\in \mathcal{Z}}\mathbbm{1}(T(\mathbf{z}, \mathbf{R}) \geq t) \cdot \prod_{i=1}^{I}\frac{\exp(\gamma \sum_{j=1}^{2}z_{ij}u_{ij})}{\sum_{j=1}^{2}\exp(\gamma u_{ij})}.
\end{equation*}
In practice, an empirical researcher gradually increases the sensitivity parameter $\Gamma$, computes the worst-case p-value for each $\Gamma$, and reports the largest $\Gamma$ such that this worst-case p-value exceeds some prespecified level $\alpha$. Such a changepoint $\Gamma$ is known as the \emph{sensitivity value} (\citealp{zhao2019sensitivityvalue}) and informs the audience of the magnitude of potential hidden bias needed to alter the causal conclusion. \citet{rosenbaum2002observational} showed that for a large class of nonparametric test statistics, the worst-case p-value in a pair-matched study is obtained when the subject with higher response in each matched pair has an unmeasured confounder $u = 1$ and the other has $u = 0$. 
\subsection{A necessary condition for a valid sensitivity analysis model}
\label{subsec: limitations of Gamma}


 Using the notation introduced in Section~\ref{subsec: review of continuous IV}, the dose assignment mechanism entails the conditional probability $P(\widetilde{Z}_{i1} = \widetilde{Z}_{i1} \wedge \widetilde{Z}_{i2}, \widetilde{Z}_{i2} = \widetilde{Z}_{i1}\vee \widetilde{Z}_{i2}\mid \mathcal{F}_{0}, \widetilde{\mathbf{Z}}_{\vee}, \widetilde{\mathbf{Z}}_{\wedge})$, which reduces to a constant $1/2$ when the IV random assignment assumption holds and otherwise does not. When there is unmeasured confounding $U$ such that the IV random assignment assumption does not hold, some parsimonious model is often assumed to describe this biased dose assignment mechanism $P(\widetilde{Z}_{i1} = \widetilde{Z}_{i1} \wedge \widetilde{Z}_{i2}, \widetilde{Z}_{i2} = \widetilde{Z}_{i1}\vee \widetilde{Z}_{i2}\mid \mathcal{F}_{1}, \widetilde{\mathbf{Z}}_{\vee}, \widetilde{\mathbf{Z}}_{\wedge})$. For example, \citet{baiocchi2010building} proposed to model it using the $\Gamma$ sensitivity analysis model for binary IVs reviewed in Section~\ref{subsec: randomization and Gamma} by embedding the continuous IV $\widetilde{Z}_{ij}$ into the binary encouragement indicator $Z_{ij}=\mathbbm{1}(\widetilde{Z}_{ij}<\widetilde{Z}_{ij^{\prime}})$ after matching and then directly applying model (\ref{eqn: sens model gamma equiv}) to $Z_{ij}$:
 \begin{align*}
     &\quad P(\widetilde{Z}_{i1} = \widetilde{Z}_{i1} \wedge \widetilde{Z}_{i2}, \widetilde{Z}_{i2} = \widetilde{Z}_{i1}\vee \widetilde{Z}_{i2}\mid \mathcal{F}_{1}, \widetilde{\mathbf{Z}}_{\vee}, \widetilde{\mathbf{Z}}_{\wedge})\\
     &= P(Z_{i1} =\mathbbm{1}(\widetilde{Z}_{i1} < \widetilde{Z}_{i2})=1, \widetilde{Z}_{i2} = \mathbbm{1}(\widetilde{Z}_{i1}>\widetilde{Z}_{i2})=0\mid \mathcal{F}_{1}, \widetilde{\mathbf{Z}}_{\vee}, \widetilde{\mathbf{Z}}_{\wedge})\\
     &=\frac{\exp(\gamma u_{i1})}{\exp(\gamma u_{i1})+\exp(\gamma u_{i2})}.
 \end{align*}
 This approach was adopted in much subsequent work.

A valid sensitivity analysis model for describing the biased dose assignment mechanism in each matched pair should correspond to a valid data generating process in the population \emph{prior to matching}. Theorem~\ref{thm: impossible} states that such a valid sensitivity analysis model \textit{must} incorporate the information about the minimum and maximum continuous doses in each pair. A surprising consequence of Theorem \ref{thm: impossible} is that the $\Gamma$ sensitivity analysis model, which does \emph{not} incorporate continuous doses, cannot be directly applied to the continuous IV settings to model the biased IV (encouragement) dose assignment mechanism.

\begin{Theorem}[An exclusion principle]\label{thm: impossible}
Let $\widetilde{Z}$ be a continuous IV, $\mathbf{X}$ a vector of observed covariates, and $U$ a hypothesized unmeasured confounder that can bias the IV assignment probability. Denote the conditional density function $f(\widetilde{Z}=\widetilde{z}\mid \mathbf{X}=\mathbf{x}, U=u)$ as $\xi(\widetilde{z}, \mathbf{x}, u)$, with $(\widetilde{z}, \mathbf{x}, u)\in \mathbb{R}^{p+2}$ where $p$ is the dimension of $\mathbf{X}$. Let $\mathcal{G}$ denote the collection of all nonnegative functions $g(\widetilde{z}, \mathbf{x}, u)$ such that $g(\widetilde{z}, \mathbf{x}, u)=\eta(\mathbf{x}, u) \zeta (\widetilde{z}, \mathbf{x}) \vartheta(\widetilde{z},u)$ for some nonnegative functions $\eta$, $\zeta$, and $\vartheta$, where $\vartheta$ is a function that has continuous second partial derivatives over the support of $g$. If $\xi(\widetilde{z}, \mathbf{x}, u) \in \mathcal{G}$, then the following two statements (S1) and (S2) cannot hold true simultaneously: (S1) The dose assignment probability in each matched pair $i$, i.e., $P(\widetilde{Z}_{i1}=\widetilde{Z}_{i1}\wedge \widetilde{Z}_{i2}, \widetilde{Z}_{i2}=\widetilde{Z}_{i1}\vee \widetilde{Z}_{i2} \mid \mathcal{F}_{1}, \widetilde{\mathbf{Z}}_{\vee}, \widetilde{\mathbf{Z}}_{\wedge})$, is biased by the hypothesized unmeasured confounder $U$ (i.e., $\widetilde{Z}\nbigCI U \mid \mathbf{X}$); (S2) The dose assignment probability $P(\widetilde{Z}_{i1}=\widetilde{Z}_{i1}\wedge \widetilde{Z}_{i2}, \widetilde{Z}_{i2}=\widetilde{Z}_{i1}\vee \widetilde{Z}_{i2} \mid \mathcal{F}_{1}, \widetilde{\mathbf{Z}}_{\vee}, \widetilde{\mathbf{Z}}_{\wedge})$ does not depend on $(\widetilde{Z}_{i1}\wedge \widetilde{Z}_{i2}, \widetilde{Z}_{i1}\vee \widetilde{Z}_{i2})$ in each matched pair.
\end{Theorem}

\begin{Remark}\rm \label{remark: an exclusion principle}
Note that $\mathcal{G}$ contains models that describe dose assignment mechanisms in the population prior to matching, not in each pair after matching. $\mathcal{G}$ covers a large class of biased and unbiased IV dose assignment models in the literature. For example, the class of all dose assignment models for a valid continuous IV, i.e., $f(\widetilde{Z}=\widetilde{z}\mid \mathbf{X}=\mathbf{x})$, is a subclass of $\mathcal{G}$ and corresponds to taking $\eta(\mathbf{x}, u)=1$, $\zeta (\widetilde{z}, \mathbf{x})=f(\widetilde{Z}=\widetilde{z}\mid \mathbf{X}=\mathbf{x})$, and $\vartheta(\widetilde{z},u)=1$. The partially linear model $\widetilde{Z} = h(\mathbf{X}) + \gamma \cdot U + \epsilon$ where $U$ has an additive effect on $\widetilde{Z}$ and $\epsilon$ is normally distributed belongs to $\mathcal{G}$. The semiparametric model considered in \citet{rosenbaum1989sensitivity} also belongs to $\mathcal{G}$ with $\vartheta(\widetilde{z}, u)=\exp(\gamma\cdot \widetilde{z}u)$ for some constant $\gamma$. See Supplementary Material A.2 for more details.
\end{Remark}

Theorem~\ref{thm: impossible} is proved in Supplementary Material B.3 and we only give some intuition here. Suppose that (S2) is true and $f(\widetilde{Z}=\widetilde{z}\mid \mathbf{X}=\mathbf{x})$ factors into $\eta(\mathbf{x}, u) \zeta (\widetilde{z}, \mathbf{x}) \vartheta(\widetilde{z},u)$ for some functions $\eta$, $\zeta$, and $\vartheta$, where $\vartheta$ is smooth. It can be shown that $\partial^{2} \ln \vartheta(\widetilde{z}, u)/\partial \widetilde{z} \partial u \equiv 0$, which implies that $\vartheta(\widetilde{z}, u)$ is separable, i.e., there exist $\vartheta_{1}$ and $\vartheta_{2}$ such that $\vartheta(\widetilde{z}, u)=\vartheta_{1}(\widetilde{z})\vartheta_{2}(u)$. Therefore, $f(\widetilde{Z}=\widetilde{z}\mid \mathbf{X}=\mathbf{x}, U=u)=\eta(\mathbf{x}, u) \zeta (\widetilde{z}, \mathbf{x})\vartheta_{1}(\widetilde{z})\vartheta_{2}(u)\propto \zeta (\widetilde{z}, \mathbf{x})\vartheta_{1}(\widetilde{z})$, which implies that $\widetilde{Z}\indep U \mid \mathbf{X}$ and thus (S1) cannot possibly hold. In other words, to conduct a sensitivity analysis with continuous IVs through modelling the biased encouragement dose assignment probability $P(\widetilde{Z}_{i1}=\widetilde{Z}_{i1}\wedge \widetilde{Z}_{i2}, \widetilde{Z}_{i2}=\widetilde{Z}_{i1}\vee \widetilde{Z}_{i2} \mid \mathcal{F}_{1}, \widetilde{\mathbf{Z}}_{\vee}, \widetilde{\mathbf{Z}}_{\wedge})$ in each matched pair after matching, we must incorporate the information of doses $(\widetilde{Z}_{i1}\wedge \widetilde{Z}_{i2}, \widetilde{Z}_{i1}\vee \widetilde{Z}_{i2})$ into that biased encouragement dose assignment model to avoid mathematical inconsistencies. Note that the $\Gamma$ sensitivity analysis model does not incorporate the information of doses $(\widetilde{Z}_{i1}\wedge \widetilde{Z}_{i2}, \widetilde{Z}_{i1}\vee \widetilde{Z}_{i2})$, and therefore does not constitute a valid sensitivity analysis model when the IV is continuous according to Theorem~\ref{thm: impossible}. We further illustrate this point using the NICU data in Supplementary Material A.3.

Previous theoretical development for a strengthened IV in the presence of biased encouragement dose assignment has assumed the $\Gamma$ sensitivity analysis model. Conclusions drawn under this model, e.g., strengthening an IV increases the power of a sensitivity analysis in large samples, may no longer apply when IVs are continuous in light of previous discussion. In the next section, we quantify the bias in inference from a biased continuous IV and use it to develop a valid sensitivity analysis framework in the following section, and then use this framework to re-evaluate previous conclusions about strengthening IVs.

\section{Assessing an IV estimator via asymptotic oracle bias (AOB)}
\label{sec: AOB}
When the IV is valid and the monotonicity assumption holds, \citet{AIR1996} showed that the Wald estimator, which coincides with two stage least squares when the treatment and IV are binary, nonparametrically estimates the average treatment effect among compliers (CATE): $\mathbb{E}\{r_{T}-r_{C}~|~(d_{T}, d_{C})=(1,0)\}$. Recall the Wald estimator in a matched study is:
\begin{equation}
    \widehat{\beta}_{\text{IV}} = \frac{\widehat{\mathbb{E}}[\ R_{ij}\mid Z_{ij}=1\ ]-\widehat{\mathbb{E}}[\ R_{ij}\mid Z_{ij}=0\ ]}{\widehat{\mathbb{E}}[\ D_{ij}\mid Z_{ij}=1\ ]-\widehat{\mathbb{E}}[\ D_{ij}\mid Z_{ij}=0 \ ]}= \frac{\sum_{i=1}^{I}(Z_{i1}-Z_{i2})(R_{i1}-R_{i2})}{\sum_{i=1}^{I}(Z_{i1}-Z_{i2})(D_{i1}-D_{i2})},
    \label{eqn: wald estimator in matched study}
\end{equation}
where $Z_{i1}=\mathbbm{1}(\widetilde{Z}_{i1}<\widetilde{Z}_{i2})$ and $Z_{i2}=\mathbbm{1}(\widetilde{Z}_{i1}>\widetilde{Z}_{i2})$, assuming no ties of $\widetilde{Z}$ in each matched pair. We let $\widehat{\iota}_{C}=I^{-1}\sum_{i=1}^{I}(Z_{i1}-Z_{i2})(D_{i1}-D_{i2})$ denote the average encouraged-minus-control difference in the treatment indicator $D$. Note that $\widehat{\iota}_C$ consistently estimates the compliance rate under the IV random assignment and monotonicity assumptions (\citealp{AIR1996}). We will refer to $\widehat{\iota}_{C}$ as the estimated compliance rate (\citealp{ertefaie2018quantitative}).

Our new framework evaluates the bias of the Wald estimator and serves as the basis for a new sensitivity analysis framework to be discussed in Section~\ref{sec: bias due to U and SA}. Our framework is distinct from previous frameworks in that we incorporate the information contained in the original continuous IV $\widetilde{Z}$. We first clarify some concepts, notation, and assumptions. Suppose that there are $N$ unmatched samples. Let $\widetilde{Z}_{n}$, $\mathbf{X}_{n}$, $D_{n}$, and $R_{n}$ denote the continuous IV, the covariates, the treatment indicator, and the outcome of individual $n$, $n=1,\dots,N$. Suppose that $(\widetilde{Z}_{n}, \mathbf{X}_{n}, D_{n}, R_{n})$ are i.i.d. random vectors, $n=1,\dots,N$. Consider the following partially linear model for the outcome $R$:
\begin{equation}\label{eqn: Y model original}
    R_{n}= \beta\cdot {D}_{n}+ f(\mathbf{X}_{n})+\widetilde{\epsilon}_{n},\quad n=1,\dots, N,
\end{equation}
where $f(\mathbf{X}_{n})$ is an arbitrary unknown function of $\mathbf{X}_{n}$ and $\beta$ denotes a constant additive treatment effect. In practice, the treatment effect could be heterogeneous (\citealp{AIR1996}), and the Wald estimate is nonparametric and does not rely on the homogeneous treatment effect assumption. We adopt this constant treatment effect model to make our sensitivity analysis framework more transparent and easier to implement. Without loss of generality, we assume that $\widetilde{\epsilon} \ \indep \ \mathbf{X}$. When the IV $\widetilde{Z}$ is valid, we have $\widetilde{Z}\ \indep \ \widetilde{\epsilon} \mid  \mathbf{X}$. We below consider the possibility that $\widetilde{Z} \nbigCI \widetilde{\epsilon} \mid \mathbf{X}$ in which case the IV is not valid. This can occur when, for example, there is an unmeasured IV-outcome confounder (\citealp{garabedian2014potential}). In this case, an IV analysis assuming $\widetilde{Z}\ \indep \ \widetilde{\epsilon} \mid  \mathbf{X}$ (i.e., no IV-outcome unmeasured confounding) may not lead to valid statistical inference (\citealp{Bound1995}; \citealp{small2008war}). We follow \citet{imai2010identification} and model residuals $\widetilde{\epsilon}_{n}$ as $\widetilde{\epsilon}_{n}=\delta \cdot U_{\text{tot},n}+\epsilon_{n}$, $n=1,\dots,N$, such that $\epsilon_{n}$ are i.i.d. random disturbances with $\mathbb{E}[\epsilon_{n}] = 0$, $\mbox{Var}[\epsilon_{n}] = \sigma^{2}$, and $\{\epsilon_{n}: n=1,\dots, N\} \ \indep \ \{(\widetilde{Z}_{n}, \mathbf{X}_{n}, U_{\text{tot},n}): n=1,\dots, N\}$. That is, each $U_{\text{tot},n}$ includes the total impact of all potential IV-outcome and treatment-outcome unmeasured confounders of individual $n$. Under this formulation, we allow $\widetilde{Z}\nbigCI   U_{\text{tot}} \mid \mathbf{X}$. Therefore, we can rewrite (\ref{eqn: Y model original}) as 
\begin{equation}\label{eqn: Y model}
    R_{n}= \beta\cdot {D}_{n}+ f(\mathbf{X}_{n})+\delta\cdot U_{\text{tot},n} +\epsilon_{n},\quad n=1,\dots, N.
\end{equation}

Model (\ref{eqn: Y model}) is imposed on samples before matching. To derive the asymptotic results for matched samples, we need to furthermore clarify some concepts and assumptions concerning a matching algorithm. A matching algorithm $\mathcal{M}$ acting on $N$ unmatched samples can be viewed as a mapping $\mathcal{M}_{N}$ from the information of $N$ individuals $\widetilde{\mathcal{F}}=\{q_{n}=(\widetilde{Z}_{n}, \mathbf{X}_{n}, U_{\text{tot},n}): n=1,\dots, N\} $ to a partition $\Pi_I$ that divides $2I$ matched individuals into $I$ matched pairs $\{ \{q_{i1}, q_{i2}\}: i=1,\dots,I\}$. A valid matching algorithm involves only the information of covariates $\mathbf{X}_{n}$ and continuous scores $\widetilde{Z}_{n}$. Definition~\ref{definition: matching alg is valid} formalizes this statement.

\begin{Definition}[A valid matching algorithm $\mathcal{M}$]
\label{definition: matching alg is valid}
A matching algorithm $\mathcal{M}$ is said to be \emph{valid} if the output of $\mathcal{M}$ (i.e., matched sets) is determined only by the information on the observed covariates and the IV.
\end{Definition}

Conditional on $\{\mathbf{X}_n, n = 1,...,N\}$ and $\{\widetilde{Z}_n, n = 1,...,N\}$, a valid matching algorithm $\mathcal{M}$ is independent of $\mathbf{U}_{\text{tot}}=(U_{\text{tot},1}, \dots, U_{\text{tot},N})$, i.e.,
\begin{equation*}
    \mathcal{M} \ \indep \ \mathbf{U}_{\text{tot}} ~|~ \{(\mathbf{X}_{n}, \widetilde{Z}_{n}): n=1,\dots, N\}.
\end{equation*}
For examples of matching algorithms applied to continuous IVs that satisfy Definition \ref{definition: matching alg is valid}, see \citet{baiocchi2010building}, \citet{zubizarreta2013stronger}, \citet{yang2014dissonant}, \citet{keele2016strong, keele2018stronger}, \citet{ertefaie2018quantitative}, and \citet{fogarty2019biased}.

To illuminate the main issues involved in strengthening IVs, we will focus on the bias of the Wald estimator under model (\ref{eqn: Y model}). We call this the \textit{oracle bias} where oracle refers to evaluating the bias under the oracle knowledge of model (\ref{eqn: Y model}) holding. Theorem~\ref{thm: asymp. Wald} gives the asymptotic normality of $\widehat{\beta}_{\text{IV}}$, which will be leveraged to develop a valid sensitivity analysis framework in Section~\ref{sec: bias due to U and SA}.
\begin{Theorem}[Asymptotic normality]
Suppose the outcome satisfies model (\ref{eqn: Y model}), and a \emph{valid} matching algorithm $\mathcal{M}$ is applied to $N$ samples and yields $I$ matched pairs. As $N \rightarrow \infty$ and consequently $I \rightarrow \infty$, the Wald estimator $\widehat{\beta}_{\text{IV}}$ satisfies
\begin{equation*}
\frac{\widehat{\beta}_{\text{IV}}-\beta-\mbox{Bias}(\widehat{\beta}_{\text{IV}})}{\mbox{sd}(\widehat{\beta}_{\text{IV}})}  \overset{\mathcal{L}}{\longrightarrow} \mathcal{N}(0,1), 
\end{equation*}
where 
\[\mbox{sd}(\widehat{\beta}_{\text{IV}})=\sqrt{2} \sigma\cdot [\sqrt{I}\cdot \widehat{\iota}_{C}]^{-1},
\]and 
\[\mbox{Bias}(\widehat{\beta}_{\text{IV}})=\frac{\sum_{i=1}^{I}(Z_{i1}-Z_{i2})[f(\mathbf{X}_{i1})-f(\mathbf{X}_{i2})]}{\sum_{i=1}^{I}(Z_{i1}-Z_{i2})(D_{i1}-D_{i2})}+\delta \cdot \frac{\sum_{i=1}^{I}(Z_{i1}-Z_{i2})(U_{\text{tot},i1}-U_{\text{tot},i2})}{\sum_{i=1}^{I}(Z_{i1}-Z_{i2})(D_{i1}-D_{i2})}.
\]
\label{thm: asymp. Wald}
\end{Theorem}
The oracle bias of $\widehat{\beta}_{\text{IV}}$ is a finite-sample quantity. If we further assume that the data are realizations from some superpopulation model, we can derive the \textit{asymptotic oracle bias} (AOB) of $\widehat{\beta}_{\text{IV}}$ under regularity conditions. The AOB formula yields more insights into when strengthening an IV reduces the bias and when it does not. 

Suppose that each $q_{n}=(\widetilde{Z}_{n}, \mathbf{X}_{n}, U_{\text{tot},n})$, $n=1,\dots, N$, is an i.i.d. realization from a distribution $G$. For the encouraged individual in matched pair $i$, i.e., the one with the smaller $\widetilde{Z}$ in pair $i$, the associated information vector $q_{Ti}=(\widetilde{Z}_{Ti}, \mathbf{X}_{Ti}, U_{\text{tot},Ti})$ follows the marginal distribution $G_{\mathcal{M}_{N}, T}$ induced by $G$ and the matching algorithm $\mathcal{M}$ applied to $\{q_{n}, n=1, \dots, N\}$. Analogously, the information vector $q_{Ci}=(\widetilde{Z}_{Ci}, \mathbf{X}_{Ci}, U_{\text{tot},Ci})$ associated with each control in pair $i$, i.e., the individual with the larger $\widetilde{Z}$ in pair $i$, follows a marginal distribution $G_{\mathcal{M}_{N}, C}$. To stress, $G_{\mathcal{M}_{N}, T}$ and $G_{\mathcal{M}_{N}, C}$ depend on $G$, the matching algorithm $\mathcal{M}$, and sample size $N$. Assumption~\ref{assumption: conv of matched dist} says $G_{\mathcal{M}_{N}, T}$ and $G_{\mathcal{M}_{N}, C}$ have well-defined limits as $N \rightarrow \infty$, and Assumption~\ref{assumption: conv of sample quantities} requires convergence in probability of some sample quantities. 
\begin{Assumption}[Convergence of the distribution of matched samples]
\label{assumption: conv of matched dist}
As $N \rightarrow \infty$ and consequently $I\rightarrow \infty$, we have $G_{\mathcal{M}_{N}, T} \xrightarrow{\mathcal{L}} G_{\mathcal{M}, T}$ and $G_{\mathcal{M}_{N}, C} \xrightarrow{\mathcal{L}} G_{\mathcal{M}, C}$ for some distributions $G_{\mathcal{M}, T}$ and $G_{\mathcal{M}, C}$.
\end{Assumption}
\begin{Assumption}[Convergence of sample means of matched data]
\label{assumption: conv of sample quantities}
Let $\mathbb{E}_{\mathcal{M},T}[\cdot ]$ and $\mathbb{E}_{\mathcal{M},C}[\cdot ]$ denote the expectations under distributions $G_{\mathcal{M}, T}$ and $G_{\mathcal{M}, C}$. As $N\rightarrow \infty$ and correspondingly $I\rightarrow \infty$, we have the following regularity conditions hold: $\frac{1}{I}\sum_{i=1}^{I}D_{Ti} \xrightarrow{p} \mathbb{E}_{\mathcal{M}, T}[D]$; $\frac{1}{I}\sum_{i=1}^{I}D_{Ci} \xrightarrow{p} \mathbb{E}_{\mathcal{M}, C}[D]$; $\frac{1}{I}\sum_{i=1}^{I}f(\mathbf{X}_{Ti})\xrightarrow{p} \mathbb{E}_{\mathcal{M}, T}[f(\mathbf{X})]$; $\frac{1}{I}\sum_{i=1}^{I}f(\mathbf{X}_{Ci})\xrightarrow{p} \mathbb{E}_{\mathcal{M}, C}[f(\mathbf{X})]$; $\frac{1}{I}\sum_{i=1}^{I}U_{\text{tot},Ti} \xrightarrow{p} \mathbb{E}_{\mathcal{M}, T}[U_{\text{tot}}]$ and $\frac{1}{I}\sum_{i=1}^{I}U_{\text{tot},Ci} \xrightarrow{p} \mathbb{E}_{\mathcal{M}, C}[U_{\text{tot}}]$. 
\end{Assumption}

Theorem~\ref{thm: asymp bias} derives the AOB of $\widehat{\beta}_{\text{IV}}$.

\begin{Theorem}[The asymptotic oracle bias of the Wald estimator]
\label{thm: asymp bias}
Suppose the outcome satisfies model (\ref{eqn: Y model}). A valid matching algorithm $\mathcal{M}$ is applied to $N$ samples and yields $I$ matched pairs. Suppose Assumptions \ref{assumption: conv of matched dist} and \ref{assumption: conv of sample quantities} hold. As $N \rightarrow \infty$ and consequently $I \rightarrow \infty$, the Wald estimator $\widehat{\beta}_{\text{IV}}$ satisfies:
\begin{equation}
   \widehat{\beta}_{\text{IV}}-\beta \overset{p}{\longrightarrow} \frac{\mathbb{E}_{\mathcal{M},T}[f(\mathbf{X})]-\mathbb{E}_{\mathcal{M},C}[f(\mathbf{X})]}{\mathbb{E}_{\mathcal{M}, T}[D]-\mathbb{E}_{\mathcal{M}, C}[D]}+\delta \cdot \frac{\mathbb{E}_{\mathcal{M}, T}[U_{\text{tot}}]-\mathbb{E}_{\mathcal{M},C}[U_{\text{tot}}]}{\mathbb{E}_{\mathcal{M},T}[D]-\mathbb{E}_{\mathcal{M},C}[D]},
   \label{eqn: asymp bias}
\end{equation}
where $\mathbb{E}_{\mathcal{M},T}[\cdot]$ and $\mathbb{E}_{\mathcal{M},C}[\cdot ]$ denote the expectations under distributions $G_{\mathcal{M}, T}$ and $G_{\mathcal{M}, C}$.
\end{Theorem}

\section{Comparing IV designs: two sources of bias}
\label{sec: two sources of bias}
Consider expression (\ref{eqn: asymp bias}) which gives the asymptotic bias of the Wald estimator in the presence of IV-outcome unmeasured confounding under model (\ref{eqn: Y model}). This bias is decomposed into two components: a bias due to the residual imbalance of $\mathbf{X}$ after matching, which is captured by the term $(\mathbb{E}_{\mathcal{M},T}[f(\mathbf{X})]-\mathbb{E}_{\mathcal{M},C}[f(\mathbf{X})])/(\mathbb{E}_{\mathcal{M}, T}[D]-\mathbb{E}_{\mathcal{M}, C}[D])$, and a bias due to failing to adjust for $U_{\text{tot}}$, which is captured by the term $(\mathbb{E}_{\mathcal{M}, T}[U_{\text{tot}}]-\mathbb{E}_{\mathcal{M},C}[U_{\text{tot}}])/(\mathbb{E}_{\mathcal{M},T}[D]-\mathbb{E}_{\mathcal{M},C}[D])$. We now leverage expression (\ref{eqn: asymp bias}) to compare two designs.
\subsection{Bias due to residual imbalance after matching}
\label{subsec: Bias due to residual imbalance after matching}
In practice, exact matching on the observed covariates may not be attainable and there will be residual imbalance due to not perfectly matching/balancing the observed covariates. Suppose we are now in a favorable situation where there is no unmeasured confounder, i.e., $\delta = 0$. The AOB of the Wald estimator then reduces to
\begin{equation}
\mbox{AOB}(\widehat{\beta}_{\text{IV}}) = \frac{\mathbb{E}_{\mathcal{M},T}[f(\mathbf{X})]-\mathbb{E}_{\mathcal{M},C}[f(\mathbf{X})]}{\mathbb{E}_{\mathcal{M},T}[D]-\mathbb{E}_{\mathcal{M},C}[D]}. \label{eqn: AOB due to X}  
\end{equation}

Let $\mathcal{M}_0$ be a matching algorithm that constructs an IV, $\mathcal{M}_1$ a strengthening-IV algorithm that constructs a stronger IV,
and $\widehat{\beta}_{\text{IV}, \mathcal{M}_0}$ and $\widehat{\beta}_{\text{IV}, \mathcal{M}_1}$ be the Wald estimators computed from $\mathcal{M}_0$ and $\mathcal{M}_1$, respectively. Consider a linear model $f(\mathbf{X})=\mathbf{X}^{T}\mathbf{\gamma}$ for some $\mathbf{\gamma}=(\gamma_{1}, \dots, \gamma_{p})^{T}$. Suppose there are $p$ observed covariates, $\{X_1, X_2, ..., X_p\}$, all of which, except possibly $X_p$, are well-balanced for both matching algorithms $\mathcal{M}_{0}$ and $\mathcal{M}_1$. By 
``$X_i$ is well-balanced'', we mean the marginal means of $X_i$ in the ``near group'' and ``far group'' are approximately equal. This information can be easily read off from the balance table (e.g., Table~\ref{tbl: balance table}). 

Magnitudes of biases due to not well balancing the observed covariate $X_p$ for designs $\mathcal{M}_0$ and $\mathcal{M}_1$, under this linear $f(\cdot)$, are
\begin{equation}
    |\mbox{AOB}(\widehat{\beta}_{\text{IV}, f(\cdot), \mathcal{M}_0})| = \left|\gamma_{p} \cdot \frac{\mathbb{E}_{\mathcal{M}_{0},T}[X_{p}]-\mathbb{E}_{\mathcal{M}_{0},C}[X_{p}]}{\mathbb{E}_{\mathcal{M}_{0},T}[D]-\mathbb{E}_{\mathcal{M}_{0},C}[D]}\right|
    \label{eqn: bias expression f linear 1}
\end{equation}
and 
\begin{equation}
     |\mbox{AOB}(\widehat{\beta}_{\text{IV}, f(\cdot), \mathcal{M}_1})| = \left|\gamma_{p} \cdot \frac{\mathbb{E}_{\mathcal{M}_{1},T}[X_{p}]-\mathbb{E}_{\mathcal{M}_{1},C}[X_{p}]}{\mathbb{E}_{\mathcal{M}_{1},T}[D]-\mathbb{E}_{\mathcal{M}_{1},C}[D]}\right|,
     \label{eqn: bias expression f linear 2}
\end{equation}
respectively. Expressions (\ref{eqn: bias expression f linear 1}) and (\ref{eqn: bias expression f linear 2}) suggest that in this simple scenario, if the marginal mean of $X_p$ is well-balanced in the ``near group'' and ``far group'' for both designs, the strengthening-IV algorithm $\mathcal{M}_1$ would largely reduce the bias due to imperfect matching through increasing the compliance (the denominator). On the other hand, if the balance of $X_p$ significantly deteriorates after the IV is strengthened, both the numerator $|\mathbb{E}_{\mathcal{M},T }[X_{p}]-\mathbb{E}_{\mathcal{M},C}[X_{p}]|$ and the denominator $|\mathbb{E}_{\mathcal{M},T }[D]-\mathbb{E}_{\mathcal{M},C }[D]|$ would grow larger in $\mathcal{M}_1$ compared to $\mathcal{M}_0$, and it is unclear whether the magnitude of bias due to not perfectly balancing $X_p$ becomes larger or smaller.

In practice, $f(\cdot)$ is not linear and many covariates are not exactly matched. Without knowing the actual functional form $f(\cdot)$, little definite can be said. The biases due to not perfectly matching on $X_{1}, \dots, X_{p}$ may add up and be very large or they may happen to cancel each other. However, practitioners should keep in mind that the strengthening-IV algorithm may introduce bias into the Wald estimator as a by-product of inferior balance on observed covariates. One useful strategy is to make sure that covariate balance after strengthening is no worse than that in the original design. In Supplementary Material A.7, we discuss how to incorporate this objective into existing matching algorithms and illustrate the approach using a subset of the NICU study data. 

\subsection{Bias due to the unmeasured confounding}\label{subsec: Bias due to the unmeasured confounder}
When exact or near-exact matching on the observed covariates $\mathbf{X}$ is possible, the conditional distributions of $\mathbf{X}$ in the encouragement group (``near'' group) and the control group (``far'' group) are the same and the magnitude of the AOBs for the two designs reduces to 
\[
|\mbox{AOB}(\widehat{\beta}_{\text{IV}, \mathcal{M}_0})| = \left|\delta \cdot \frac{\mathbb{E}_{\mathcal{M}_0,T}[U_{\text{tot}}]-\mathbb{E}_{\mathcal{M}_0,C}[U_{\text{tot}}]}{\mathbb{E}_{\mathcal{M}_0,T}[D] - \mathbb{E}_{\mathcal{M}_0,C}[D]}\right|
\] and
\[|\mbox{AOB}(\widehat{\beta}_{\text{IV}, \mathcal{M}_1})| = \left|\delta \cdot \frac{\mathbb{E}_{\mathcal{M}_1,T}[U_{\text{tot}}] - \mathbb{E}_{\mathcal{M}_1,C}[U_{\text{tot}}]}{\mathbb{E}_{\mathcal{M}_1,T}[D] - \mathbb{E}_{\mathcal{M}_1,C}[D]}\right|.
\] The \emph{ratio of the magnitude of bias due to unmeasured confounding} is 
\begin{equation}\label{equ: ratio of the bias via AOB}
    \begin{split}
        \Delta &= \frac{|\mbox{AOB}(\widehat{\beta}_{\text{IV}, \mathcal{M}_1})|}{ |\mbox{AOB}(\widehat{\beta}_{\text{IV}, \mathcal{M}_0})|} = \left|\frac{\mathbb{E}_{\mathcal{M}_{0},T}[D]-\mathbb{E}_{\mathcal{M}_{0},C}[D]}{\mathbb{E}_{\mathcal{M}_{1},T}[D]-\mathbb{E}_{\mathcal{M}_{1},C}[D]}\right| \times \left|\frac{\mathbb{E}_{\mathcal{M}_{1},T}[U_{\text{tot}}]-\mathbb{E}_{\mathcal{M}_{1},C}[U_{\text{tot}}]}{\mathbb{E}_{\mathcal{M}_{0},T}[U_{\text{tot}}]-\mathbb{E}_{\mathcal{M}_{0},C}[U_{\text{tot}}]}\right|.
    \end{split}
\end{equation}
The first factor in expression (\ref{equ: ratio of the bias via AOB}) is the ratio of encouraged-minus-control difference in the treatment indicator $D$ after matching using algorithms $\mathcal{M}_0$ and $\mathcal{M}_1$, while the second factor captures the ratio of average encouraged-minus-control difference in $U_{\text{tot}}$ after matching. Whether or not the magnitude of bias contributed by $U_{\text{tot}}$ increases when strengthening an IV depends on \textit{both} ratios. Strengthening an IV typically makes the first ratio in (\ref{equ: ratio of the bias via AOB}) less than 1 but may render the second ratio in (\ref{equ: ratio of the bias via AOB}) greater than 1. Example \ref{example: nicu U} highlights this point using the NICU data, and a simulation is provided in Supplementary Material C.2 for further illustration.

\begin{Example}[Bias of the Wald estimator in NICU data]\label{example: nicu U}
\rm
Consider a hypothetical NICU dataset where, contrary to the fact, the single birth indicator is not collected. The analyst would not know this and the best he or she could do is to match on all observed covariates. We therefore match on all covariates except the single birth indicator and examine the bias due to this unmeasured confounder in the near/far matching algorithm ($\mathcal{M}_0$) without strengthening the IV and the strengthening-IV algorithm ($\mathcal{M}_1$) that only forms half as many matched pairs as the data allows. The estimated compliance rate is $0.21$ in $\mathcal{M}_0$ and $0.42$ in $\mathcal{M}_1$. The average encouraged-minus-control differences in the unmatched single birth indicator are $1.89 \times 10^{-3}$ for design $\mathcal{M}_0$ and $5.01 \times 10^{-3}$ for design $\mathcal{M}_1$. Put together, the magnitude of bias evaluates to $\delta \times 9.16 \times 10^{-3}$ for design $\mathcal{M}_0$ and $\delta \times 1.19 \times 10^{-2}$ for $\mathcal{M}_1$, and hence the ratio of bias is $\Delta = 1.30 > 1$. In this hypothetical scenario, strengthening the IV ($\mathcal{M}_1$) amplifies the bias of the Wald estimator under model (\ref{eqn: Y model}) due to the unmeasured confounder (i.e., the single birth indicator). 
\end{Example}

\section{A valid AOB-based model-assisted sensitivity analysis framework}
\label{sec: bias due to U and SA}
Suppose that we have controlled for all observed covariates available in a study via matching. What can be said about an IV design's sensitivity to unmeasured confounding? We answer this question by developing a valid model-assisted sensitivity analysis. 

\subsection{A general sensitivity analysis framework}\label{subsec: AOB sensitivity analysis}
A sensitivity analysis outputs a range of plausible values of the parameter of interest when an assumption is allowed to be violated up to a certain magnitude. Consider the data generating process (\ref{eqn: Y model}) described in Section~\ref{sec: AOB}. Our goal is to output a range of plausible values of $\beta$ when unmeasured confounding is present (i.e., $\delta \neq 0$). Below, we derive and describe a unified sensitivity analysis framework based on the AOB formula. There are three key ingredients in our framework.

\textbf{1) A parametric model that relates $U_{\text{tot}}$ to $\widetilde{Z}$ and $\mathbf{X}$:} To still make inference on $\beta$ in the presence of unmeasured confounding, we need to quantify the bias due to it. Recall the AOB of $\widehat{\beta}_{\text{IV}}$ (conditional on $\mathbf{X}_{i1} = \mathbf{X}_{i2}$) is
\[
\mbox{AOB}(\widehat{\beta}_{\text{IV}}) = \delta \cdot \frac{\mathbb{E}_{\mathcal{M},T}[U_{\text{tot}}] - \mathbb{E}_{\mathcal{M},C}[U_{\text{tot}}]}{\mathbb{E}_{\mathcal{M},T}[D] - \mathbb{E}_{\mathcal{M},C}[D]},
\]
where the denominator can be directly estimated from data and the numerator involves expectations of $U_{\text{tot}}$ in the encouragement group and the control group. The following key observation helps estimate these two expectations:
\begin{equation}
    \begin{split}
    &\mathbb{E}_{\mathcal{M},T}[U_{\text{tot}}]-\mathbb{E}_{\mathcal{M},C}[U_{\text{tot}}] \\
    = ~&\mathbb{E}_{\mathcal{M},T}\{\mathbb{E}_{\mathcal{M},T}[U_{\text{tot}} \mid \widetilde{Z}, \mathbf{X}]\} - \mathbb{E}_{\mathcal{M},C}\{\mathbb{E}_{\mathcal{M},C}[U_{\text{tot}} \mid \widetilde{Z}, \mathbf{X}]\} \\
   = ~&\mathbb{E}_{\mathcal{M},T}\{\mathbb{E}[U_{\text{tot}} \mid \widetilde{Z}, \mathbf{X}]\} - \mathbb{E}_{\mathcal{M},C}\{\mathbb{E}[U_{\text{tot}} \mid \widetilde{Z}, \mathbf{X}] \}~\text{by Definition 1},
    \end{split}
    \label{eqn: key obs in SA}
\end{equation}
where $\mathbb{E}[U_{\text{tot}} \mid \widetilde{Z}, \mathbf{X}$] refers to the expectation of $U_{\text{tot}}$ given $\widetilde{Z}$ and $\mathbf{X}$ prior to matching in the population. Equation (\ref{eqn: key obs in SA}) suggests that in order to estimate the bias contributed by $U_{\text{tot}}$ after matching, it suffices to posit a model for $U_{\text{tot}} \mid \widetilde{Z}, \mathbf{X}$ before matching, provided that the matching algorithm $\mathcal{M}$ is valid (Definition~\ref{definition: matching alg is valid}).

\textbf{2) A sensitivity zone:} Consider a parametric or possibly semiparametric model of $U_{\text{tot}} \mid \widetilde{Z}, \mathbf{X}$ parametrized by some finite-dimensional parameter $\boldsymbol{\theta}$. Given modeling assumptions on $U_{\text{tot}} \mid \widetilde{Z}, \mathbf{X}$, we may identify $\boldsymbol{\theta}$ from the observed data. However, the identification is typically very weak and these parameters would better be regarded as sensitivity parameters (\citealp{Copas1997}; \citealp{Scharfstein1999}; \citealp{Imbens2003}). A sensitivity zone refers to a set of plausible values of sensitivity parameters and is denoted as $\boldsymbol{I}$. The simplest form of a sensitivity zone is the Cartesian product of intervals or general point sets of each coordinate of $\boldsymbol{\theta}$ and the scalar $\delta$. In many situations, however, it helps to reparametrize the sensitivity parameters from $\boldsymbol{\theta}$ to parameters that are easier to understand and communicate (\citealp{Imbens2003}; \citealp{rosenbaum2009amplification}; \citealp{cinelli2018making}). 

\textbf{3) A sensitivity interval:} The output of our sensitivity analysis is a $100(1-\alpha)\%$ finite-sample \emph{sensitivity interval} (SI) of $\beta$ given a sensitivity zone $\boldsymbol{I}$, which contains the true $\beta$ with probability at least $1-\alpha$ for any unknown true sensitivity parameter $(\delta, \boldsymbol{\theta}) \in \boldsymbol{I}$. In Supplementary Material A.5, we described in detail how to construct a $100(1-\alpha)\%$ confidence interval of $\beta$ given the true $(\delta, \boldsymbol{\theta})$ by treating $\{(U_{\text{tot},i1}, U_{\text{tot},i2}), i = 1,...,I\}$ as missing covariates and adopting the multiple imputation (MI) paradigm (\citealp{rubin1987multiple}; \citealp{ichino2008temporary}). For each fixed $(\delta, \boldsymbol{\theta}) \in \boldsymbol{I}$, a $100(1-\alpha)\%$ confidence interval of $\beta$ is
\begin{equation*}
  \text{CI}(\delta, \boldsymbol{\theta})=\Big [K^{-1}\sum_{k = 1}^K \widehat{\beta}_{\text{IV}, k}-t_{\nu}^{-1}(1-\alpha/2)\cdot \widehat{\sigma}_{\text{total}}, \ K^{-1}\sum_{k = 1}^K \widehat{\beta}_{\text{IV}, k}+t_{\nu}^{-1}(1-\alpha/2)\cdot \widehat{\sigma}_{\text{total}} \Big],
\end{equation*}
where there are $K$ imputed datasets, 
\begin{equation*}
   \widehat{\beta}_{\text{IV}, k} =\widehat{\beta}_{\text{IV}} - \delta \cdot \frac{\sum_{i=1}^{I}(Z_{i1}-Z_{i2})(U_{\text{tot},i1}^{(k)}-U_{\text{tot},i2}^{(k)})}{\sum_{i=1}^{I}(Z_{i1}-Z_{i2})(D_{i1}-D_{i2})}, \quad k=1,\dots, K,
\end{equation*}
with $\{(U_{\text{tot},i1}^{(k)}, U_{\text{tot},i2}^{(k)}), i = 1,...,I\}$ being the $k$-th imputed set of $U_{\text{tot}}$ according to the specified model for $U_{\text{tot}} \mid \widetilde{Z}, \mathbf{X}$ given the sensitivity parameter $(\delta, \boldsymbol{\theta})$, and $K^{-1}\sum_{k = 1}^K \widehat{\beta}_{\text{IV}, k}$ is the pooled estimate for $\beta$ under $(\delta, \boldsymbol{\theta})$. The estimated total variance of the pooled estimate $K^{-1}\sum_{k = 1}^K \widehat{\beta}_{\text{IV}, k}$ is
\begin{equation}
    \widehat{\sigma}_{\text{total}}^{2}= \frac{1 + K^{-1}}{K - 1}\sum_{k = 1}^K \Big(\widehat{\beta}_{\text{IV}, k} - K^{-1}\sum_{k = 1}^K \widehat{\beta}_{\text{IV}, k}\Big)^2 + K^{-1}\sum_{k=1}^{K} \frac{2\widehat{\sigma}^{2}_{k}}{I \cdot \widehat{\iota}^2_C},
\label{eqn: total variance}
\end{equation}
and $t_{\nu}^{-1}$ is the quantile function of the Student's t-distribution with degrees of freedom 
\begin{equation*}
    \nu = (K - 1)\Bigg [1 + \frac{K^{-1}\sum_{k=1}^{K}2\widehat{\sigma}^{2}_{k}\cdot [I\cdot \widehat{\iota}^2_C]^{-1}}{\frac{1 + K^{-1}}{K - 1}\sum_{k = 1}^K \Big(\widehat{\beta}_{\text{IV}, k} - K^{-1}\sum_{k = 1}^K \widehat{\beta}_{\text{IV}, k}\Big)^2}\Bigg ]^2.
\end{equation*}
A $100(1-\alpha)\%$ sensitivity interval (SI) of $\beta$ given the sensitivity zone $\boldsymbol{I}$ is then $\text{SI}=\bigcup_{(\delta, \boldsymbol{\theta}) \in \boldsymbol{I}}\text{CI}(\delta, \boldsymbol{\theta})$.

In expression (\ref{eqn: total variance}), $\widehat{\sigma}_k^2$ is an estimate of $\sigma^2$ based on the $k$-th imputed dataset. One estimation strategy is to posit a model for $f(\mathbf{X}_{n})$ in (\ref{eqn: Y model}); in Supplementary Material A.6, we show that $\sigma^2$ can also be estimated from matched pair data without positing a model for $f(\mathbf{X}_{n})$, but positing some other additional models, including a model relating $D$, ${\bf{X}}$, $U_{\text{tot}}$ and $\widetilde{Z}$. The additional model(s) make the sensitivity analysis model-assisted. The primary analysis, the Wald estimator after matching (i.e., $\widehat{\beta}_{IV}$ defined in (\ref{eqn: wald estimator in matched study})), is nonparametric. Model-assisted sensitivity analyses that complement nonparametric primary analyses have been proposed in other contexts (\citealp{rosenbaum2009amplification}; \citealp{nattino2018model}).

\subsection{An example of AOB-based model-assisted sensitivity analysis}
\label{subsec: example of sens model and application to NICU}
We illustrate the proposed AOB-based model-assisted sensitivity analysis framework by describing a concrete model for $U_{\text{tot}} \mid \widetilde{Z}, \mathbf{X}$ and associated parametrization of sensitivity parameters. Consider the following simple and parsimonious model: 
\begin{equation}
    \begin{split}
        &\widetilde{Z} \sim H(\cdot), \text{ such that } \mathbb{E}_{\widetilde{Z} \sim H(\cdot)}[\widetilde{Z}] = 0 \text{ and } \text{Var}_{\widetilde{Z} \sim H(\cdot)}[\widetilde{Z}] = 1,\\
        &U_{\text{tot}} \sim \text{Bern} \bigg(\frac{\exp(\lambda_0 + \lambda_1 \widetilde{Z})}{1 + \exp(\lambda_0 + \lambda_1 \widetilde{Z})}\bigg).
    \end{split}
    \label{model: sens model}
\end{equation}
According to model (\ref{model: sens model}), there exists a binary unmeasured confounder in the population, and $\widetilde{Z}$ has mean zero and unit variance in the population, marginally. To incorporate observed covariates $\mathbf{X}$, we specify that model (\ref{model: sens model}) holds when $\widetilde{Z}$ is replaced with $\overset{\approx}{Z} = (\widetilde{Z} - \mathbb{E}[\widetilde{Z} \mid \mathbf{X}])/\sigma(\widetilde{Z} \mid \mathbf{X})$, the standardized residual of $\widetilde{Z}$ after the effect of $\mathbf{X}$ is taken out. In practice, we may fit a linear regression of $\widetilde{Z}$ on $\mathbf{X}$, take the residual, and scale it to have mean $0$ and variance $1$. It facilitates interpretation to reparametrize the sensitivity parameters $(\lambda_0, \lambda_1)$ into $(\tau, \lambda_1)$, where $\tau = P(U_{\text{tot}} =1 \mid \widetilde{Z} > \mbox{Median}(\widetilde{Z})) - P(U_{\text{tot}} = 1 \mid \widetilde{Z} < \mbox{Median}(\widetilde{Z}))$. In this way, $\tau$ quantifies how important $U$ is to $\widetilde{Z}$ conditional on $\widetilde{Z}$ being either some constant $c$ above or below the median of $\widetilde{Z}$, and $\lambda_{1}$ quantifies how fast $P(U_{\text{tot}}=1 \mid \widetilde{Z}=\mbox{Median}(\widetilde{Z})+c)$ changes as $c$ changes for $c>0$. One nuance in this reparametrization is that the map from $(\lambda_0, \lambda_1)$ to $(\tau, \lambda_1)$ is not surjective: there are combinations of $(\tau, \lambda_1)$ that do not correspond to any choice of $(\lambda_0, \lambda_1)$. After this reparametrization, there are a total of three sensitivity parameters in this sensitivity analysis model: $(\delta, \tau, \lambda_1)$. Let $\boldsymbol{\Delta}$ be a collection of plausible values for $\delta$, and similarly $\boldsymbol{T}$ for $\tau$ and $\boldsymbol{\Lambda}$ for $\lambda_1$. For the sensitivity zone, we consider $\text{I}(\boldsymbol{\Delta}, \boldsymbol{T}, \boldsymbol{\Lambda}) = \boldsymbol{\Delta} \times \boldsymbol{T} \times \boldsymbol{\Lambda}$.

\subsection{Comparing IV designs via the power of sensitivity analyses}
When the treatment indeed has an effect on the outcome of interest, we would not be able to recognize this from observational data and if a large enough amount of unmeasured confounding is considered, a sensitivity analysis will say the observed treatment effect might be purely due to this unmeasured confounding rather than a genuine treatment effect. The situation where there is a genuine treatment effect and no unmeasured confounding is called a \emph{favorable situation} (\citealp{heller2009split}; \citealp{rosenbaum2010design}). In such a favorable situation, we would like a sensitivity analysis to conclude that the treatment effect is insensitive to a small or moderately large amount of unmeasured confounding. \emph{The power of a sensitivity analysis} is the probability that we are able to make such a statement when unmeasured confounding is limited up to some extent (\citealp{rosenbaum2004design, rosenbaum2010design}).

We examine power in our model-assisted sensitivity analysis framework in this section. We consider the sensitivity analysis discussed in Section~\ref{subsec: example of sens model and application to NICU}. Let $(R_{n}, D_{n}, \widetilde{Z}_{n}, X_{1n}, X_{2n}$, $X_{3n})$, $n=1,\dots, 1000$, be i.i.d. data generated from the following model:
\begin{align}
\label{model: simulated power via AOB}
&(X_{1n}, X_{2n}, X_{3n}) \sim \mathcal{N}(0, \mathbf{I}_3), \ \epsilon_{n}\overset{\text{i.i.d.}}{\sim} \mathcal{N}(0,1), \nonumber\\
    & \widetilde{Z}_{n} \overset{\text{i.i.d.}}{\sim} \mathcal{N}(0 , 1),\ D_{n}\sim \text{Bern$\bigg(\frac{\exp(\xi \widetilde{Z}_{n})}{1+\exp(\xi \widetilde{Z}_{n})}\bigg)$},\\
    &R_{n}=\beta \cdot D_{n}+ 0.2\cdot X_{1n}+ 0.5\cdot \log (|X_{2n}|)+0.3\cdot \sin(X_{3n}) + \epsilon_{n} \nonumber,
\end{align}
where $\widetilde{Z}$ is a continuous IV, $D$ a binary treatment, $(X_{1}, X_{2}, X_{3})$ three observed covariates, and $R$ an observed outcome. We conducted a sensitivity analysis under model (\ref{model: sens model}) with reparametrized sensitivity parameters $(\tau, \lambda_1, \delta)$.

We calculate the finite sample power of this sensitivity analysis under the \textit{favorable situation} (\ref{model: simulated power via AOB}) for two designs: an ``unstrengthened'' IV design $\mathcal{M}_0$ that uses all data and forms $500$ matched pairs with similar $(X_1, X_2, X_{3})$ but distinct $\widetilde{Z}$, and a ``strengthened'' IV design $\mathcal{M}_1$ that only forms half as many matched pairs as the data allows ($250$ matched pairs) where the pairs are still similar in $(X_1, X_2, X_{3})$ but are even more distinct in $\widetilde{Z}$; see Supplementary Material A.1 for details on matching. For each simulated dataset and a fixed sensitivity zone $\boldsymbol{I} = \{\tau\} \times \{\lambda_1\} \times [-\delta_{\sup}, \delta_{\sup}]$, we calculate the $95\%$ sensitivity interval $\text{SI}=\bigcup_{(\delta, \boldsymbol{\theta}) \in \boldsymbol{I}}\text{CI}(\delta, \boldsymbol{\theta})$ as described in Section~\ref{subsec: AOB sensitivity analysis}, and reject the null hypothesis of no treatment effect (i.e., $\beta=0$ in the hypothesized outcome model (\ref{eqn: Y model})) if the SI does not contain the origin (to focus on main issues in illuminating the effect of strengthening the IV, we assume the true $\sigma = 1$ is known). Table~\ref{tab: simulated power via AOB} summarizes the results for various sensitivity zones (indexed by $(\delta_\text{sup}, \lambda_1)$ with $\tau = 0.01$) and different data generating processes of $D$ (indexed by $\xi$). The estimated compliance rate $\widehat{\iota}_{C}$, average bias, and average standard deviation for both designs are also reported in the table.

\begin{table}[ht]

\caption{ The simulated power of unstrengthened-IV design $\mathcal{M}_{0}$ and strengthened-IV design $\mathcal{M}_{1}$ for various $\beta$, $\delta_{\sup}$, $\xi$, and $\lambda_{1}$. $\mathcal{M}_{0}$ uses all the 1000 samples and forms 500 matched pairs. $\mathcal{M}_{1}$ only forms half as many matched pairs as the data allows ($250$ matched pairs). We set $\tau=0.01$. Estimated average bias and average sampling standard deviation (SD) $\widehat{\sigma}_{\text{total}}$ are reported in the parenthesis as (Bias/SD). We also report the estimated compliance rate $\widehat{\iota}_{C}$. Not shown in the table, the simulated size (i.e., setting $\beta = \delta_{\sup}=0$) is 0.053 for $\mathcal{M}_{0}$ and 0.054 for $\mathcal{M}_{1}$. 2000 simulations were done for these settings and each setting in the table.}
\begin{tabular}{c c c c c c c c c c} 
 \hline
 \multirow{2}{*}{Compliance rate $\widehat{\iota}_{C}$}& \multicolumn{2}{c}{$\xi=1.0$}& \multicolumn{2}{c}{$\xi=1.2$}& \multicolumn{2}{c}{$\xi=4.0$}& \multicolumn{2}{c}{$\xi=5.0$}\\
\cmidrule(r){2-3} \cmidrule(r){4-5} \cmidrule(r){6-7} \cmidrule(r){8-9} 
   & $\mathcal{M}_{0}$ & $\mathcal{M}_{1}$ & $\mathcal{M}_{0}$ & $\mathcal{M}_{1}$ & $\mathcal{M}_{0}$ & $\mathcal{M}_{1}$ & $\mathcal{M}_{0}$ & $\mathcal{M}_{1}$ \\ [0.5ex]
 \hline
  & 0.24  & 0.41 & 0.27 &  0.46 & 0.45  & 0.75 &  0.46 &  0.78    \\
 \hline
 \smallskip
 \smallskip
\end{tabular}
 \begin{tabular}{c c c c c c} 
 \hline
 \multirow{2}{*}{Power (Bias/SD)} & \multicolumn{2}{c}{$\xi=1.0$}& \multicolumn{2}{c}{$\xi=1.2$} \\
\cmidrule(r){2-3} \cmidrule(r){4-5}
 $\beta=0.8, \delta_{\sup}=0.5$  & $\mathcal{M}_{0}$ & $\mathcal{M}_{1}$ & $\mathcal{M}_{0}$ & $\mathcal{M}_{1}$  \\ [0.5ex] 
 \hline
 $\lambda_{1}=1.0$  & 0.70 (0.14/0.27)  & 0.83 (0.14/0.22)  & 0.80 (0.12/0.24)   & 0.93 (0.12/0.20)  \\
  \hline
 $\lambda_{1}=1.5$  & 0.66 (0.16/0.27)  & 0.81 (0.16/0.22)  & 0.78 (0.14/0.24)   & 0.91 (0.14/0.20)  \\
  \hline
 $\lambda_{1}=2.0$  & 0.64 (0.18/0.27)  & 0.80 (0.19/0.22) & 0.76 (0.16/0.24)  & 0.90 (0.16/0.20)  \\
  \hline
 $\lambda_{1}=2.5$ & 0.60 (0.21/0.27)  & 0.75 (0.21/0.22) & 0.73 (0.19/0.24)  & 0.87 (0.19/0.20) \\ 
  \hline
 $\lambda_{1}=3.0$ & 0.57 (0.23/0.27) & 0.71 (0.24/0.22) & 0.70 (0.20/0.24) & 0.84 (0.21/0.20) \\
  \hline
 \multirow{2}{*}{Power (Bias/SD)} & \multicolumn{2}{c}{$\xi=4.0$}& \multicolumn{2}{c}{$\xi=5.0$} \\
\cmidrule(r){2-3} \cmidrule(r){4-5}
 $\beta=4, \delta_{\sup}=10$  & $\mathcal{M}_{0}$ & $\mathcal{M}_{1}$ & $\mathcal{M}_{0}$ & $\mathcal{M}_{1}$  \\ [0.5ex] 
 \hline
 $\lambda_{1}=6.0$  & 0.67 (3.21/0.27)  & 0.52 (3.37/0.26)  & 0.78 (3.08/0.26)  & 0.62 (3.27/0.25)  \\
  \hline
 $\lambda_{1}=9.0$  & 0.54 (3.42/0.24)  & 0.39 (3.61/0.23)  & 0.68 (3.29/0.23) & 0.51 (3.50/0.22)  \\
  \hline
 $\lambda_{1}=12.0$  & 0.52 (3.51/0.22)  & 0.36 (3.70/0.21) & 0.66 (3.38/0.21) & 0.50  (3.57/0.20)  \\
  \hline
 $\lambda_{1}=15.0$ & 0.51 (3.54/0.21)  & 0.36 (3.74/0.19) & 0.65 (3.42/0.20) & 0.45 (3.62/0.19) \\ 
  \hline
 $\lambda_{1}=18.0$ & 0.50 (3.57/0.20) & 0.36 (3.76/0.19) & 0.63 (3.45/0.19) & 0.48 (3.65/0.18) \\
 \hline
\end{tabular}
\label{tab: simulated power via AOB}
\end{table}

Table~\ref{tab: simulated power via AOB} suggests that under our framework, the strengthening-IV design sometimes but does \emph{not} always increase the power of a sensitivity analysis. Whether or not strengthening-IV designs would increase or decrease the power of a sensitivity analysis depends on how strengthening affects both bias and variance. When the sampling variability is small compared to the bias (e.g., when the sample size $I$ is very large or the association between the unmeasured confounder and the outcome $\delta$ is very large), the power is primarily determined by the bias of the Wald estimator. If a strengthening-IV design $\mathcal{M}_1$ amplifies the bias compared to an unstrengthening-IV design $\mathcal{M}_0$, we would expect $\mathcal{M}_1$ to have a smaller power of sensitivity analysis; if $\mathcal{M}_1$ mitigates the bias compared to $\mathcal{M}_0$, we would expect $\mathcal{M}_{1}$ to have a larger power. See Section \ref{subsec: Bias due to the unmeasured confounder} for a detailed discussion on when a strengthening-IV design would amplify or mitigate the bias. Table~\ref{tab: simulated power via AOB} exhibits cases (with $\beta = 4$ and $\delta_{\sup}=10$) where the standard deviation of the estimate is small compared to the bias (e.g., SD = $0.27$ and Bias = $3.21$ for $\mathcal{M}_{0}$ when $\xi=4.0$ and $\lambda_1 = 6.0$), and the power is driven by the bias. For instance, when $\xi = 4.0$ and $5.0$, we see that the strengthening-IV design $\mathcal{M}_1$ amplifies the bias compared to $\mathcal{M}_0$, and as a result it decreases the power of a sensitivity analysis. Our observation and conclusion here is distinct from previous results derived from the $\Gamma$ sensitivity analysis model: \citet{ertefaie2018quantitative} showed that, in a favorable situation, testing using a strong IV (e.g., a ``strengthened'' IV) always exhibits a larger insensitivity to unmeasured confounding (i.e., larger power of sensitivity analysis) asymptotically (known as the \emph{design sensitivity}), compared to doing the same test but with a weaker IV. 

When the biases of both designs are comparable, the variance of the estimator contributes significantly to the power. Table~\ref{tab: simulated power via AOB} exhibits cases (with $\beta = 0.8$ and $\delta_{\sup} = 0.5$) where the biases of both designs are about the same but the standard deviations are quite different (e.g., when $\xi = 1.0$ and $\lambda_1 = 1.0$, the estimated biases under both designs are $0.14$; the standard deviation is $0.27$ for $\mathcal{M}_{0}$ and $0.22$ for $\mathcal{M}_{1}$). In such cases, a design with smaller sampling variability often exhibits a larger power (e.g., $\mathcal{M}_1$ is more powerful than $\mathcal{M}_0$ when $\xi = 1.0$ and $\lambda_1 = 1.0$).

\subsection{Sensitivity analysis for the NICU study}
\label{subsec: Sensitivity analysis for the NICU study}
We apply the proposed sensitivity analysis model to the NICU data described in Section~\ref{sec: introduction}. The outcome of interest is infants' length of stay at hospital. For infants who died, we imputed the $99\%$ upper quantile of survivors length of stay; this is a burden of illness approach to analyzing the length of stay in which death is considered to be a very bad outcome (\citealp{chang1994reduction}; \citealp{lin2017placement}). The $95\%$ confidence interval of the naive Wald estimator assuming no unmeasured confounding is $[0.78, 1.90]$ for a near/far matching design $\mathcal{M}_0$ that does not strengthen the IV, and $[1.04, 1.80]$ for the strengthening-IV design $\mathcal{M}_1$. Thus, if the IV of excess travel time is valid, we have evidence that high-level NICUs increase infants' length of stay at a hospital.  

We next consider doing a sensitivity analysis under model (\ref{model: sens model}). For each $(\tau, \lambda_1, \delta)$ combination, we estimate $\sigma$ assuming $f(\mathbf{X}_n)$ in (\ref{eqn: Y model}) is linear in $\mathbf{X}_n$, and construct the sensitivity interval according to the MI paradigm discussed in Section \ref{subsec: AOB sensitivity analysis}. For each $(\tau, \lambda_1)$ combination, Figure \ref{subfig: nicu sens plot not strengthen} reports the largest $\Delta$ (denoted as $\Delta_{\text{sup}}$), where $\Delta$ is defined as $\delta$ divided by the standard deviation of the outcome, such that the $95\%$ sensitivity interval does not contain $0$ for design $\mathcal{M}_0$, and Figure \ref{subfig: nicu sens plot strengthen} plots the same information for design $\mathcal{M}_1$. In both figures, darker shades correspond to larger $\Delta_\text{sup}$, i.e., a larger association between the outcome and the hypothesized unmeasured confounder needed in order to nullify the conclusion that high-level NICUs increase length of stay. 
 
\begin{figure}[!htb]
  \centering
  \caption{ Sensitivity plots of the NICU study for the ``unstrengthened'' IV design $\mathcal{M}_0$ and the ``strengthened'' IV design $\mathcal{M}_1$. x-axis is $\lambda_1$ and y-axis is $\tau$. Darker shades correspond to larger values of $\Delta_\text{sup}$. White lines are isopleths.}
  \begin{subfigure}{.5\textwidth}
    \centering
    \includegraphics[width = 6.0 cm, height = 4.5 cm]{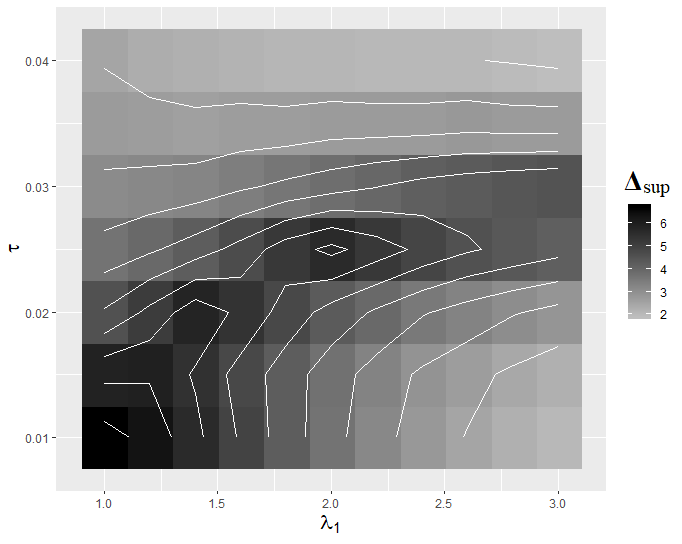}
    \caption{ }
    \label{subfig: nicu sens plot not strengthen}
  \end{subfigure}%
  \hspace{0 cm}
  \begin{subfigure}{.5\textwidth}
    \centering
    \includegraphics[width = 6.0 cm, height = 4.5 cm]{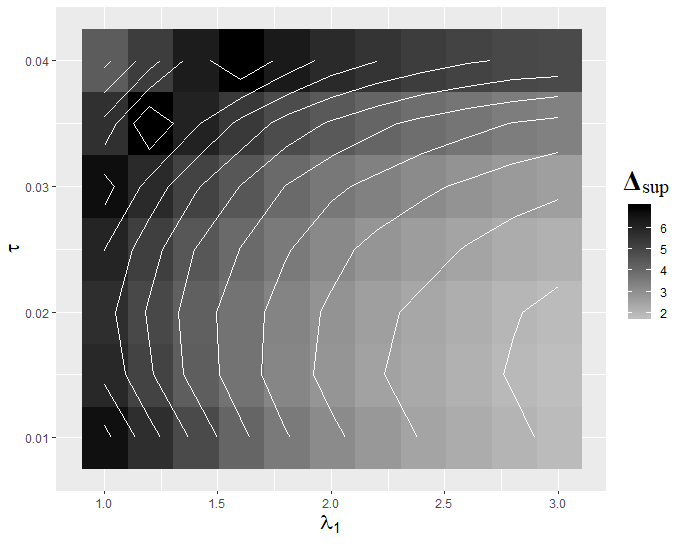}
    \caption{}
    \label{subfig: nicu sens plot strengthen}
  \end{subfigure}
  \label{fig: sens plot NICU}
\end{figure}
\thispagestyle{empty}

For either design, we consider a sensitivity zone of $\tau$ from $0.01$ to $0.04$ and $\lambda_1$ from $1$ to $3$. A $\lambda_1$ as large as $3$ means that for one standard deviation increase in $\overset{\approx}{Z}$, the standardized residual of $\widetilde{Z}$, we expect to see an $\exp (3)\approx 20$-fold increase in the odds of $U_{\text{tot}}= 1$ versus $U_{\text{tot}} = 0$. On the other hand, $\tau$ encodes the difference in $P(U_{\text{tot}} = 1)$ conditional on $\overset{\approx}{Z}$ being above or below its median. To get a sense of how large $\tau$ could be, we may look at the value of $\tau$ if $U_{\text{tot}}$ were an observed covariate. This value equals 0.01 for the covariate ``single birth'' and $0.06$ for the covariate ``white''. The combination $\tau \times \lambda_1 = [0.01, 0.04] \times [1, 3]$ thus encodes a wide range of possible relationships between $U_{\text{tot}}$ and $\widetilde{Z}$, and Figure~\ref{fig: sens plot NICU} suggests that it would require a moderate to enormous association between the unmeasured confounder and the outcome in order to nullify the causal conclusion for both designs.

As we compare Figure~\ref{subfig: nicu sens plot not strengthen} to \ref{subfig: nicu sens plot strengthen}, we observe that the strengthened IV does not always exhibit larger insensitivity: there are $(\tau \times \lambda_1)$ combinations such that $\mathcal{M}_0$ exhibits a larger insensitivity (darker shade) and vice versa. Recall that small $\lambda_1$ means that $P(U_{\text{tot}} = 1 \mid \widetilde{Z} = \text{Median}(\widetilde{Z}) + c)$ does not change rapidly as $c$ changes, which corresponds to the strength of unmeasured confounding not being much modified by the continuous IV $\widetilde{Z}$. Therefore, not surprisingly, when $\lambda_1$ is small, we observe the strengthened IV seems to be more robust to unmeasured confounding compared to the unstrengthened IV.

\section{Discussion: the implications and some practical advice}
\label{sec: conclusion}
In this article, we study the implications of building a stronger IV from an existing, but possibly weak, IV in the design stage of a matched analysis. We spell out the potential outcomes framework in a matching-based study design approach to continuous IV analysis, and  clarify what constitutes a valid sensitivity analysis model for randomization-based inference in this setting. In particular, we demonstrate, via theory and real data examples, that the extensively used $\Gamma$ sensitivity analysis model is in general \textit{not} valid in the continuous IV setting. This observation implies that theoretical analyses based on the $\Gamma$ sensitivity analysis model that are valid in a binary IV setting, for instance conclusions concerning the power of a sensitivity analysis, do not necessarily extend to the continuous IV setting. We developed a new framework to assess IV designs' sensitivity to unmeasured confounding based on evaluating the bias of the Wald estimator. Indeed, under our framework, strengthening an IV may not increase the power of sensitivity analyses. 

We now provide some practical guidance for empirical researchers based on our results. When an IV is believed to be valid conditional on the set of observed covariates, efficiency of testing is an important criterion. Our first set of results quantifies the relative efficiency of performing hypothesis testing using IVs of different strength: in order for a weaker IV (with compliance rate $\iota_{c, 1}$) to attain the same power as a stronger IV (with compliance rate $\iota_{c, 2}$ such that $\iota_{c, 2} > \iota_{c, 1}$), the weaker IV needs to have a sample size that is at least $(\iota_{c, 2}^2/\iota_{c, 1}^2)$ times the sample size of the stronger IV. If an empirical researcher has multiple binary IVs, he or she may build a stronger IV by defining one aggregate binary IV out of them. For instance, one may define encouragement to be ``doubly encouraged'' or ``multiply encouraged'', meaning a person is encouraged if two or more IVs simultaneously prod her into accepting the IV. A valid and potentially stronger IV with smaller sample size will be forged, and our theory provides some insight into when this is a good idea. As for a continuous IV, one can always forge a stronger IV using the strategy illustrated in \citet{baiocchi2010building}, i.e. optimal non-bipartite matching with sinks. Our result suggests that practitioners may throw away as much as three-fourths of the sample if they believe doing so may double the compliance rate, or half of the data if doing so may increase the compliance rate by roughly $40\%$. 

When there is concern that the putative IV may be invalid because of IV-outcome unmeasured confounding, we found that strengthening an IV in the design stage may amplify or mitigate the bias of the Wald estimator due to not controlling for the unmeasured confounding. A strengthening-IV design is favorable if it helps reduce the bias due to unmeasured confounding. In practice, we do not get to observe the IV-outcome unmeasured confounders. One sensible strategy is to sequentially leave out each observed covariate and calculate how omitting the covariate affects the bias as in Example~\ref{example: nicu U}. One might pay particular attention to the results from leaving out observed covariates that are thought to be related to unmeasured confounders of greatest concern. There is also bias due to inexact matching in reality. Strengthening an IV, valid or not, may potentially amplify this bias. We recommend that empirical researchers make sure that covariate balance is no worse in the strengthening-IV design compared to the original design, which can be achieved in some circumstances via a two-step matching algorithm described in our Supplementary Material A.7.

\section*{Acknowledgements}
The authors would like to thank Hyunseung Kang and Jos\'e Zubizarreta for helpful discussions, Michael Baiocchi for sharing the code for the NICU study, and participants in the causal inference reading group of the University of Pennsylvania for helpful comments. 

\bibliographystyle{apalike}
\bibliography{strengthenIV}

\begin{thebibliography}{}

\bibitem[Angrist and Krueger, 1994]{angrist1994world}
Angrist, J. and Krueger, A.~B. (1994).
\newblock Why do {World War II} veterans earn more than nonveterans?
\newblock {\em Journal of Labor Economics}, 12(1):74--97.

\bibitem[Angrist et~al., 1996]{AIR1996}
Angrist, J.~D., Imbens, G.~W., and Rubin, D.~B. (1996).
\newblock Identification of causal effects using instrumental variables.
\newblock {\em Journal of the American Statistical Association},
  91(434):444--455.

\bibitem[Baiocchi et~al., 2014]{Baiocchi_ivtutorial2014}
Baiocchi, M., Cheng, J., and Small, D.~S. (2014).
\newblock Instrumental variable methods for causal inference.
\newblock {\em Statistics in Medicine}, 33(13):2297--2340.

\bibitem[Baiocchi et~al., 2010]{baiocchi2010building}
Baiocchi, M., Small, D.~S., Lorch, S., and Rosenbaum, P.~R. (2010).
\newblock Building a stronger instrument in an observational study of perinatal
  care for premature infants.
\newblock {\em Journal of the American Statistical Association},
  105(492):1285--1296.

\bibitem[Baiocchi et~al., 2012]{baiocchi2012near}
Baiocchi, M., Small, D.~S., Yang, L., Polsky, D., and Groeneveld, P.~W. (2012).
\newblock Near/far matching: a study design approach to instrumental variables.
\newblock {\em Health Services and Outcomes Research Methodology},
  12(4):237--253.

\bibitem[Balke and Pearl, 1997]{balke1997bounds}
Balke, A. and Pearl, J. (1997).
\newblock Bounds on treatment effects from studies with imperfect compliance.
\newblock {\em Journal of the American Statistical Association},
  92(439):1171--1176.

\bibitem[Berkowitz et~al., 2017]{berkowitz2017supplemental}
Berkowitz, S.~A., Seligman, H.~K., Rigdon, J., Meigs, J.~B., and Basu, S.
  (2017).
\newblock Supplemental nutrition assistance program (snap) participation and
  health care expenditures among low-income adults.
\newblock {\em JAMA Internal Medicine}, 177(11):1642--1649.

\bibitem[Berkowitz et~al., 2019]{berkowitz2019association}
Berkowitz, S.~A., Terranova, J., Randall, L., Cranston, K., Waters, D.~B., and
  Hsu, J. (2019).
\newblock Association between receipt of a medically tailored meal program and
  health care use.
\newblock {\em JAMA Internal Medicine}, 179(6):786--793.

\bibitem[Bound et~al., 1995]{Bound1995}
Bound, J., Jaeger, D.~A., and Baker, R.~M. (1995).
\newblock Problems with instrumental variables estimation when the correlation
  between the instruments and the endogeneous explanatory variable is weak.
\newblock {\em Journal of the American Statistical Association},
  90(430):443--450.

\bibitem[Bronars and Grogger, 1994]{bronars1994economic}
Bronars, S.~G. and Grogger, J. (1994).
\newblock The economic consequences of unwed motherhood: Using twin births as a
  natural experiment.
\newblock {\em The American Economic Review}, pages 1141--1156.

\bibitem[Chang et~al., 1994]{chang1994reduction}
Chang, M., Guess, H., and Heyse, J. (1994).
\newblock Reduction in burden of illness: a new efficacy measure for prevention
  trials.
\newblock {\em Statistics in Medicine}, 13(18):1807--1814.

\bibitem[Cinelli and Hazlett, 2019]{cinelli2018making}
Cinelli, C. and Hazlett, C. (2019).
\newblock Making sense of sensitivity: Extending omitted variable bias.
\newblock {\em Journal of the Royal Statistical Society. Series B (Statistical
  Methodology)}, (forthcoming).

\bibitem[Copas and Li, 1997]{Copas1997}
Copas, J.~B. and Li, H.~G. (1997).
\newblock Inference for non-random samples.
\newblock {\em Journal of the Royal Statistical Society. Series B (Statistical
  Methodology)}, 59(1):55--95.

\bibitem[Ertefaie et~al., 2018]{ertefaie2018quantitative}
Ertefaie, A., Small, D.~S., and Rosenbaum, P.~R. (2018).
\newblock Quantitative evaluation of the trade-off of strengthened instruments
  and sample size in observational studies.
\newblock {\em Journal of the American Statistical Association},
  113(523):1122--1134.

\bibitem[Fisher, 1935]{fisher1937design}
Fisher, R.~A. (1935).
\newblock {\em The Design of Experiments}.
\newblock Edinburgh and London: Oliver and Boyd.

\bibitem[Fogarty et~al., 2019]{fogarty2019biased}
Fogarty, C.~B., Lee, K., Kelz, R.~R., and Keele, L.~J. (2019).
\newblock Biased encouragements and heterogeneous effects in an instrumental
  variable study of emergency general surgical outcomes.
\newblock {\em arXiv preprint arXiv:1909.09533}.

\bibitem[Garabedian et~al., 2014]{garabedian2014potential}
Garabedian, L.~F., Chu, P., Toh, S., Zaslavsky, A.~M., and Soumerai, S.~B.
  (2014).
\newblock Potential bias of instrumental variable analyses for observational
  comparative effectiveness research.
\newblock {\em Annals of Internal Medicine}, 161(2):131--138.

\bibitem[Goyal et~al., 2013]{goyal2013length}
Goyal, N., Zubizarreta, J.~R., Small, D.~S., and Lorch, S.~A. (2013).
\newblock Length of stay and readmission among late preterm infants: an
  instrumental variable approach.
\newblock {\em Hospital Pediatrics}, 3(1):7--15.

\bibitem[Grieve et~al., 2019]{grieve2019analysis}
Grieve, R., O’Neill, S., Basu, A., Keele, L., Rowan, K.~M., and Harris, S.
  (2019).
\newblock Analysis of benefit of intensive care unit transfer for deteriorating
  ward patients: a patient-centered approach to clinical evaluation.
\newblock {\em JAMA Network Open}, 2(2):e187704--e187704.

\bibitem[Hansen, 2004]{hansen2004full}
Hansen, B.~B. (2004).
\newblock Full matching in an observational study of coaching for the sat.
\newblock {\em Journal of the American Statistical Association},
  99(467):609--618.

\bibitem[Heller et~al., 2009]{heller2009split}
Heller, R., Rosenbaum, P.~R., and Small, D.~S. (2009).
\newblock Split samples and design sensitivity in observational studies.
\newblock {\em Journal of the American Statistical Association},
  104(487):1090--1101.

\bibitem[Hern{\'a}n and Robins, 2006]{hernan2006instruments}
Hern{\'a}n, M.~A. and Robins, J.~M. (2006).
\newblock Instruments for causal inference: an epidemiologist's dream?
\newblock {\em Epidemiology}, pages 360--372.

\bibitem[Holland, 1988]{holland1988causal}
Holland, P.~W. (1988).
\newblock Causal inference, path analysis and recursive structural equations
  models.
\newblock {\em ETS Research Report Series}, 1988(1):i--50.

\bibitem[Ichino et~al., 2008]{ichino2008temporary}
Ichino, A., Mealli, F., and Nannicini, T. (2008).
\newblock From temporary help jobs to permanent employment: what can we learn
  from matching estimators and their sensitivity?
\newblock {\em Journal of Applied Econometrics}, 23(3):305--327.

\bibitem[Imai et~al., 2010]{imai2010identification}
Imai, K., Keele, L., and Yamamoto, T. (2010).
\newblock Identification, inference and sensitivity analysis for causal
  mediation effects.
\newblock {\em Statistical Science}, pages 51--71.

\bibitem[Imbens, 2003]{Imbens2003}
Imbens, G.~W. (2003).
\newblock {Sensitivity to exogeneity assumptions in program evaluation}.
\newblock {\em The American Economic Review}, 93:126--132.

\bibitem[Imbens and Angrist, 1994]{Imbens_CACE1994}
Imbens, G.~W. and Angrist, J.~D. (1994).
\newblock Identification and estimation of local average treatment effects.
\newblock {\em Econometrica}, 62(2):467--475.

\bibitem[Imbens and Rosenbaum, 2005]{Imbens_Rosenbaum2005}
Imbens, G.~W. and Rosenbaum, P.~R. (2005).
\newblock Robust, accurate confidence intervals with a weak instrument: Quarter
  of birth and education.
\newblock {\em Journal of the Royal Statistical Society. Series A (Statistics
  in Society)}, 168(1):109--126.

\bibitem[Jackson and Swanson, 2015]{jackson2015toward}
Jackson, J.~W. and Swanson, S.~A. (2015).
\newblock Toward a clearer portrayal of confounding bias in instrumental
  variable applications.
\newblock {\em Epidemiology}, 26(4):498.

\bibitem[Kang et~al., 2016]{kang2016full}
Kang, H., Kreuels, B., May, J., and Small, D.~S. (2016).
\newblock Full matching approach to instrumental variables estimation with
  application to the effect of malaria on stunting.
\newblock {\em The Annals of Applied Statistics}, 10(1):335--364.

\bibitem[Keele et~al., 2018]{keele2018stronger}
Keele, L., Harris, S., Pimentel, S.~D., and Grieve, R. (2018).
\newblock Stronger instruments and refined covariate balance in an
  observational study of the effectiveness of prompt admission to intensive
  care units.
\newblock {\em Journal of the Royal Statistical Society: Series A (Statistics
  in Society)}.

\bibitem[Keele et~al., 2016]{keele2016strong}
Keele, L., Morgan, J.~W., et~al. (2016).
\newblock How strong is strong enough? strengthening instruments through
  matching and weak instrument tests.
\newblock {\em The Annals of Applied Statistics}, 10(2):1086--1106.

\bibitem[Lehmann et~al., 2017]{lehmann2017strengthening}
Lehmann, D., Li, Y., Saran, R., and Li, Y. (2017).
\newblock Strengthening instrumental variables through weighting.
\newblock {\em Statistics in Biosciences}, 9(2):320--338.

\bibitem[Lehmann, 2004]{lehmann2004elements}
Lehmann, E.~L. (2004).
\newblock {\em Elements of Large-Sample Theory}.
\newblock Springer Science \& Business Media.

\bibitem[Lin et~al., 2017]{lin2017placement}
Lin, W., Halpern, S.~D., Prasad~Kerlin, M., and Small, D.~S. (2017).
\newblock A “placement of death” approach for studies of treatment effects
  on icu length of stay.
\newblock {\em Statistical Methods in Medical Research}, 26(1):292--311.

\bibitem[Lorch et~al., 2012]{lorch2012differential}
Lorch, S.~A., Baiocchi, M., Ahlberg, C.~E., and Small, D.~S. (2012).
\newblock The differential impact of delivery hospital on the outcomes of
  premature infants.
\newblock {\em Pediatrics}, 130(2):270--278.

\bibitem[Lu et~al., 2011]{lu2011optimal}
Lu, B., Greevy, R., Xu, X., and Beck, C. (2011).
\newblock Optimal nonbipartite matching and its statistical applications.
\newblock {\em The American Statistician}, 65(1):21--30.

\bibitem[Lu et~al., 2001]{lu2001matching}
Lu, B., Zanutto, E., Hornik, R., and Rosenbaum, P.~R. (2001).
\newblock Matching with doses in an observational study of a media campaign
  against drug abuse.
\newblock {\em Journal of the American Statistical Association},
  96(456):1245--1253.

\bibitem[Lum et~al., 2017]{lum2017causal}
Lum, K., Ma, E., and Baiocchi, M. (2017).
\newblock The causal impact of bail on case outcomes for indigent defendants in
  new york city.
\newblock {\em Observational Studies}, 3:39--64.

\bibitem[McClellan et~al., 1994]{mcclellan1994does}
McClellan, M., McNeil, B.~J., and Newhouse, J.~P. (1994).
\newblock Does more intensive treatment of acute myocardial infarction in the
  elderly reduce mortality?: analysis using instrumental variables.
\newblock {\em JAMA}, 272(11):859--866.

\bibitem[Nattino and Lu, 2018]{nattino2018model}
Nattino, G. and Lu, B. (2018).
\newblock Model assisted sensitivity analyses for hidden bias with binary
  outcomes.
\newblock {\em Biometrics}, 74(4):1141--1149.

\bibitem[Neuman et~al., 2014]{neuman2014anesthesia}
Neuman, M.~D., Rosenbaum, P.~R., Ludwig, J.~M., Zubizarreta, J.~R., and Silber,
  J.~H. (2014).
\newblock Anesthesia technique, mortality, and length of stay after hip
  fracture surgery.
\newblock {\em JAMA}, 311(24):2508--2517.

\bibitem[Neyman, 1923]{neyman1923application}
Neyman, J.~S. (1923).
\newblock On the application of probability theory to agricultural experiments.
  essay on principles. section 9.(tlanslated and edited by dm dabrowska and tp
  speed, statistical science (1990), 5, 465-480).
\newblock {\em Annals of Agricultural Sciences}, 10:1--51.

\bibitem[Pimentel et~al., 2015]{pimentel2015large}
Pimentel, S.~D., Kelz, R.~R., Silber, J.~H., and Rosenbaum, P.~R. (2015).
\newblock Large, sparse optimal matching with refined covariate balance in an
  observational study of the health outcomes produced by new surgeons.
\newblock {\em Journal of the American Statistical Association},
  110(510):515--527.

\bibitem[Robins, 1994]{robins1994correcting}
Robins, J.~M. (1994).
\newblock Correcting for non-compliance in randomized trials using structural
  nested mean models.
\newblock {\em Communications in Statistics-Theory and methods},
  23(8):2379--2412.

\bibitem[Rosenbaum, 1989]{rosenbaum1989sensitivity}
Rosenbaum, P.~R. (1989).
\newblock Sensitivity analysis for matched observational studies with many
  ordered treatments.
\newblock {\em Scandinavian Journal of Statistics}, pages 227--236.

\bibitem[Rosenbaum, 1999]{rosenbaum1999using}
Rosenbaum, P.~R. (1999).
\newblock Using quantile averages in matched observational studies.
\newblock {\em Journal of the Royal Statistical Society: Series C (Applied
  Statistics)}, 48(1):63--78.

\bibitem[Rosenbaum, 2002]{rosenbaum2002observational}
Rosenbaum, P.~R. (2002).
\newblock {\em Observational Studies}.
\newblock Springer.

\bibitem[Rosenbaum, 2004]{rosenbaum2004design}
Rosenbaum, P.~R. (2004).
\newblock Design sensitivity in observational studies.
\newblock {\em Biometrika}, 91(1):153--164.

\bibitem[Rosenbaum, 2005]{Rosenbuam_heter_causality2005}
Rosenbaum, P.~R. (2005).
\newblock Heterogeneity and causality.
\newblock {\em The American Statistician}, 59(2):147--152.

\bibitem[Rosenbaum, 2010]{rosenbaum2010design}
Rosenbaum, P.~R. (2010).
\newblock {\em Design of Observational Studies}.
\newblock Springer.

\bibitem[Rosenbaum and Silber, 2009]{rosenbaum2009amplification}
Rosenbaum, P.~R. and Silber, J.~H. (2009).
\newblock Amplification of sensitivity analysis in matched observational
  studies.
\newblock {\em Journal of the American Statistical Association},
  104(488):1398--1405.

\bibitem[Rubin, 1973]{rubin1973matching}
Rubin, D.~B. (1973).
\newblock Matching to remove bias in observational studies.
\newblock {\em Biometrics}, pages 159--183.

\bibitem[Rubin, 1974]{rubin1974estimating}
Rubin, D.~B. (1974).
\newblock Estimating causal effects of treatments in randomized and
  nonrandomized studies.
\newblock {\em Journal of Educational Psychology}, 66(5):688.

\bibitem[Rubin, 1980]{rubin1980discussion}
Rubin, D.~B. (1980).
\newblock Discussion of "randomization analysis of experimental data in the
  fisher randomization test" by {D. Basu}.
\newblock {\em Journal of the American Statistical Association}, 75:591--593.

\bibitem[Rubin, 1987]{rubin1987multiple}
Rubin, D.~B. (1987).
\newblock {\em Multiple Imputation for Nonresponse in Surveys}.
\newblock John Wiley \& Sons.

\bibitem[Rubin, 1996]{rubin1996multiple}
Rubin, D.~B. (1996).
\newblock Multiple imputation after 18+ years.
\newblock {\em Journal of the American statistical Association},
  91(434):473--489.

\bibitem[Rubin, 2008]{rubin2008objective}
Rubin, D.~B. (2008).
\newblock For objective causal inference, design trumps analysis.
\newblock {\em The Annals of Applied Statistics}, 2(3):808--840.

\bibitem[Santana-Davila et~al., 2015]{santana2015cisplatin}
Santana-Davila, R., Devisetty, K., Szabo, A., Sparapani, R., Arce-Lara, C.,
  Gore, E.~M., Moran, A., Williams, C.~D., Kelley, M.~J., and Whittle, J.
  (2015).
\newblock Cisplatin and etoposide versus carboplatin and paclitaxel with
  concurrent radiotherapy for stage iii non--small-cell lung cancer: An
  analysis of veterans health administration data.
\newblock {\em Journal of Clinical Oncology}, 33(6):567.

\bibitem[Scharfstein et~al., 1999]{Scharfstein1999}
Scharfstein, D.~O., Rotnitzky, A., and Robins, J.~M. (1999).
\newblock Adjusting for nonignorable drop-out using semiparametric nonresponse
  models.
\newblock {\em Journal of the American Statistical Association},
  94(448):1096--1120.

\bibitem[Small and Rosenbaum, 2008]{small2008war}
Small, D.~S. and Rosenbaum, P.~R. (2008).
\newblock War and wages: the strength of instrumental variables and their
  sensitivity to unobserved biases.
\newblock {\em Journal of the American Statistical Association},
  103(483):924--933.

\bibitem[Staiger and Stock, 1994]{staiger1994instrumental}
Staiger, D.~O. and Stock, J.~H. (1994).
\newblock Instrumental variables regression with weak instruments.
\newblock {\em National Bureau of Economic Research Cambridge, Mass., USA}.

\bibitem[Stock et~al., 2002]{stock2002survey}
Stock, J.~H., Wright, J.~H., and Yogo, M. (2002).
\newblock A survey of weak instruments and weak identification in generalized
  method of moments.
\newblock {\em Journal of Business \& Economic Statistics}, 20(4):518--529.

\bibitem[Strauss, 2007]{strauss2007partial}
Strauss, W.~A. (2007).
\newblock {\em Partial Differential Equations: An Introduction}.
\newblock John Wiley \& Sons.

\bibitem[Stuart, 2010]{stuart2010matching}
Stuart, E.~A. (2010).
\newblock Matching methods for causal inference: A review and a look forward.
\newblock {\em Statistical Science}, 25(1):1.

\bibitem[Swanson et~al., 2018]{swanson2018partial}
Swanson, S.~A., Hern{\'a}n, M.~A., Miller, M., Robins, J.~M., and Richardson,
  T.~S. (2018).
\newblock Partial identification of the average treatment effect using
  instrumental variables: review of methods for binary instruments, treatments,
  and outcomes.
\newblock {\em Journal of the American Statistical Association},
  113(522):933--947.

\bibitem[Wald, 1940]{wald1940fitting}
Wald, A. (1940).
\newblock The fitting of straight lines if both variables are subject to error.
\newblock {\em The Annals of Mathematical Statistics}, 11(3):284--300.

\bibitem[Yang et~al., 2014]{yang2014dissonant}
Yang, F., Zubizarreta, J.~R., Small, D.~S., Lorch, S., and Rosenbaum, P.~R.
  (2014).
\newblock Dissonant conclusions when testing the validity of an instrumental
  variable.
\newblock {\em The American Statistician}, 68(4):253--263.

\bibitem[Zhao, 2019]{zhao2019sensitivityvalue}
Zhao, Q. (2019).
\newblock On sensitivity value of pair-matched observational studies.
\newblock {\em Journal of the American Statistical Association},
  114(526):713--722.

\bibitem[Zhao and Small, 2018]{zhao2018graphical}
Zhao, Q. and Small, D.~S. (2018).
\newblock Graphical diagnosis of confounding bias in instrumental variable
  analysis.
\newblock {\em Epidemiology}, 29(4):e29--e31.

\bibitem[Zubizarreta, 2012]{zubizarreta2012using}
Zubizarreta, J.~R. (2012).
\newblock Using mixed integer programming for matching in an observational
  study of kidney failure after surgery.
\newblock {\em Journal of the American Statistical Association},
  107(500):1360--1371.

\bibitem[Zubizarreta et~al., 2013]{zubizarreta2013stronger}
Zubizarreta, J.~R., Small, D.~S., Goyal, N.~K., Lorch, S., Rosenbaum, P.~R.,
  et~al. (2013).
\newblock Stronger instruments via integer programming in an observational
  study of late preterm birth outcomes.
\newblock {\em The Annals of Applied Statistics}, 7(1):25--50.

\end{thebibliography}

\clearpage

\bigskip
  \bigskip
  \bigskip
  \begin{center}
    {\Large \bf Web-based Supplementary Material for ``Re-Evaluating Strengthened-IV Designs: Asymptotic Efficiency, Bias Formula, and the Validity and Power of Sensitivity Analyses''}
 \end{center}
  \medskip

\begin{abstract}
    Supplementary Material A contains detailed discussion on near/far matching, Remark~\ref{remark: an exclusion principle}, Theorem~\ref{thm: impossible}, Example~\ref{example: nicu U}, the MI paradigm, estimation of $\sigma$ in model $(\ref{eqn: Y model})$, and the proposed two-step debiased matching algorithm. Supplementary Material B contains the statement and proof of a more general version of Theorem~\ref{thm: ARE in independent case} (i.e., Theorem~\ref{thm: complete ARE in general case}), proofs of Theorems \ref{thm: ARE in independent case}, \ref{thm: impossible}, \ref{thm: asymp. Wald}, and $\ref{thm: asymp bias}$. Supplementary Material C contains additional simulations on Theorem~\ref{thm: ARE in independent case} and bias due to the unmeasured confounder. 
\end{abstract}

\begin{center}
{\large\bf Supplementary Material A: Discussions and Details on Examples and Matching}
\end{center}

\subsection*{A.1: More details on statistical matching}
We discuss how to apply \citet{lu2011optimal}'s optimal non-bipartite matching to pair babies with similar observed covariates but different excess travel times. Suppose there are $2I$ babies prior to matching. An $2I \times 2I$ distance matrix was defined between every pair of babies. An optimal non-bipartite matching then divides these $2I$ babies into $I$ non-overlapping pairs of two babies in an optimal way, meaning the sum of distances within the $I$ pairs is minimized. The distance between two babies consists of two parts: a distance that quantifies the distance between the observed covariates of each pair of babies, and a substantial penalty is added to the distance between any pair of babies whose $\widetilde{Z}$ differ (in absolute value) by at most $\Lambda$. For the vanilla matching algorithm $\mathcal{M}_{0}$, we set $\Lambda=0$. To strengthen the IV, we eliminate some babies in the matching by adding ``sinks'' (\citealp{lu2001matching, lu2011optimal}; \citealp{baiocchi2010building}). In order to eliminate $e$ babies, $e$ sinks are added to the dataset before matching, yielding a $(2I + e) \times (2I + e)$ distance matrix. The distance between every sink and every baby is $0$ and that between every pair of sinks is set to infinity. An optimal match pairs $e$ babies to the $e$ sinks in a way that minimizes the remaining total sum of distances of $I - e/2$ pairs of babies, i.e., an optimal size-$e$ set of babies is removed. Throughout our real data analysis, we used a robust rank-based Mahalanobis distance between observed covariates, and a penalty caliper $\Lambda = 25$ was applied. When strengthening an IV, we added $e = I$ sinks and eliminated half of the $2I$ babies.

\subsection*{A.2: \citet{rosenbaum1989sensitivity}'s semiparametric dose assignment model}
\label{app: a second Rosenbaum model}
\citet{rosenbaum1989sensitivity} considered a different dose assignment model where encouragement probability is allowed to depend on magnitude of continuous treatments or IVs. According to \citet{rosenbaum1989sensitivity}, the dose assignment probability for a multi-valued encouragement $\mathbf{\widetilde{Z}}$ follows the semiparametric model below:
\begin{equation}\label{model: semipara}
    \begin{split}
       &P(\widetilde{Z}_{11}=\widetilde{z}_{11}, \dots, \widetilde{Z}_{I2}=\widetilde{z}_{I2}\mid \mathcal{F}_{1})=\prod_{i=1}^{I}\prod_{j=1}^{2} \xi(\widetilde{z}_{ij}, \mathbf{x}_{ij}, u_{ij}),\\
    &\xi (\widetilde{z}_{ij}, \mathbf{x}_{ij}, u_{ij})=\alpha(\mathbf{x}_{ij}, u_{ij})\cdot \exp\{ \chi(\widetilde{z}_{ij}, \mathbf{x}_{ij})+\gamma \cdot \widetilde{z}_{ij}\cdot u_{ij}\}, 
    \end{split}
\end{equation}
where $\chi(\widetilde{z}_{ij}, \mathbf{x}_{ij})$ an unknown function, $\gamma$ an unknown scalar parameter, and $\alpha(\cdot, \cdot)$ the normalizing constant. When $\gamma=0$, two individuals with the same $\mathbf{X} = \mathbf{x}$ have the same distribution of dose $\widetilde{Z}$. When $\gamma > 0$, after adjusting for $\mathbf{X}$, individuals with higher values of $U$ tend to receive higher dose $\widetilde{Z}$: for two individuals $m$ and $n$ such that $\mathbf{x}_{m}=\mathbf{x}_{n}$, $P(\widetilde{Z}_{m}=\widetilde{z}\mid \mathcal{F}_{1})/P(\widetilde{Z}_{n}=\widetilde{z}\mid \mathcal{F}_{1})\propto \exp\{ \gamma \cdot (u_{m}-u_{n}) \cdot 
\widetilde{z}\}$ increases in $\widetilde{z}$ when $\gamma>0$ and $u_{m}>u_{n}$. Under model (\ref{model: semipara}), \cite{rosenbaum1989sensitivity} showed that:
\begin{equation}
    \begin{split}
    P(Z_{i1}=1, Z_{i2}=0\mid \mathcal{F}_{1}, \widetilde{\mathbf{Z}}_{\vee}, \widetilde{\mathbf{Z}}_{\wedge})&=P(\widetilde{Z}_{i1}>\widetilde{Z}_{i2}\mid \mathcal{F}_{1}, \widetilde{\mathbf{Z}}_{\vee}, \widetilde{\mathbf{Z}}_{\wedge})\\
    &=P(\widetilde{Z}_{i1}>\widetilde{Z}_{i2}\mid \mathcal{F}_{1}, \widetilde{Z}_{i1}\wedge \widetilde{Z}_{i2}, \widetilde{Z}_{i1}\vee \widetilde{Z}_{i2})\\
    &=\frac{\exp\{\gamma (\widetilde{Z}_{i1}-\widetilde{Z}_{i2}) (u_{i1}-u_{i2})\}}{1+\exp\{\gamma (\widetilde{Z}_{i1}-\widetilde{Z}_{i2}) (u_{i1}-u_{i2})\}},
    \end{split}
    \label{eqn: semi model}
\end{equation}
when $\widetilde{Z}_{i1}\wedge \widetilde{Z}_{i2} \neq \widetilde{Z}_{i1}\vee \widetilde{Z}_{i2}$. Equation (\ref{eqn: semi model}) suggests that dose assignment probability is in general correlated with the information of continuous doses through $\widetilde{Z}_{i1}\wedge \widetilde{Z}_{i2}$ and $\widetilde{Z}_{i1}\vee \widetilde{Z}_{i2}$.

\subsection*{A.3: Illustrating Theorem~\ref{thm: impossible} using the NICU study}

We consider a hypothetical version of the dataset introduced in Section~\ref{sec: introduction}, where, contrary to the fact, the single birth indicator is not collected. We further assume that there is no other unmeasured confounder other than the unmatched single birth indicator. Therefore, the unmeasured confounder $U = 1$ if that person's single birth indicator is $1$ and $0$ otherwise. Consider a subset of people $\mathcal{A}$ with excess travel time $\widetilde{Z} \in [5 ~\text{min}, 8~\text{min}]$ and $U = 1$. We perform a near/far matching as described in Supplementary Material A.1 and keep track of whom $\mathcal{A}$ are paired to after matching. Consider those among $\mathcal{A}$ who are paired to someone with $U = 0$ and calculate the probability that these people are assigned encouragement (i.e., paired to someone with a larger $\widetilde{Z}$). According to the $\Gamma$ sensitivity analysis model (\ref{eqn: sens model gamma equiv}), this probability depends only on the difference in $U$ and not on the continuous IVs. Since $U_{i1} - U_{i2}$ is held fixed at $1$, this probability is a constant and does not vary with the absolute difference in the continuous IV: $|\widetilde{Z}_{i1} - \widetilde{Z}_{i2}|=\widetilde{Z}_{i1}\vee \widetilde{Z}_{i2}-\widetilde{Z}_{i1}\wedge \widetilde{Z}_{i2}$. Figure~\ref{fig: gamma is problematic} plots this encouragement probability as a function of $|\widetilde{Z}_{i1} - \widetilde{Z}_{i2}|$. We test the null hypothesis that this assignment probability is not correlated to  $|\widetilde{Z}_{i1} - \widetilde{Z}_{i2}|$ using Spearman's correlation test, and the p-value is $2.07\times 10^{-6}$. 

\begin{figure}[ht]
\caption{\small Probability that a subject with excess travel time $Z \in [5, 8]$ and $U = 1$ is assigned to encouragement indicator $Z=1$ after being matched to a subject with $U = 0$ and varying excess travel time. Spearman's rank correlation test yields a p-value of $2.07\times 10^{-6}$.}
\label{fig: gamma is problematic}
\centering
\includegraphics[width = 7 cm, height = 5 cm]{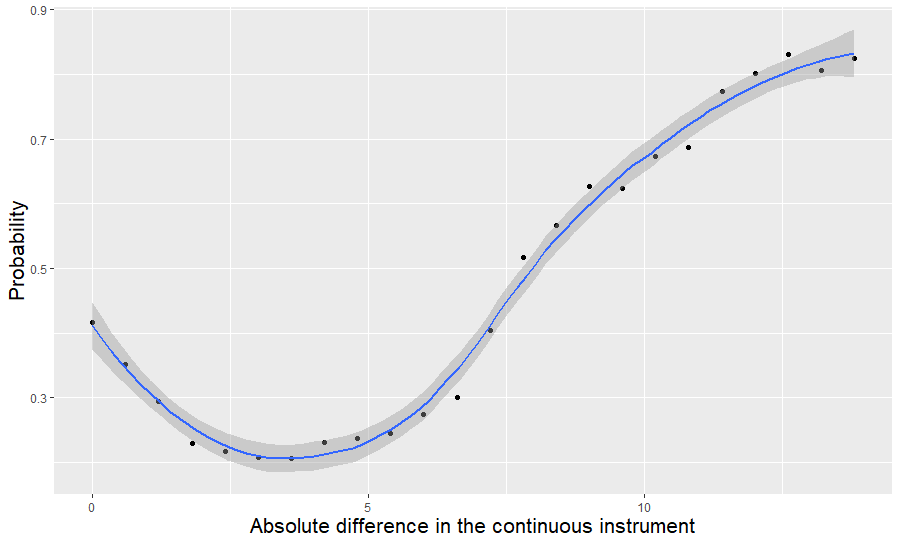}
\end{figure}

\subsection*{A.4: More details on Example~\ref{example: nicu U}}
\label{app: details on ex 2}
\begin{table}[H]
    \caption{Two hypothetical scenario where the single birth indicator and the race indicator (white or not) are not observed or matched on.}
    \label{tbl: AOB niuc example}
    \begin{tabular}{ c c c c c } 
  \hline
  \multirow{2}{*}{} &  \multicolumn{2}{c}{$U_{\text{tot}}=$ single birth indicator}  &  \multicolumn{2}{c}{$U_{\text{tot}}=$ race indicator (white or not)} \\ 
   & Original ($\mathcal{M}_0$) & Strengthened ($\mathcal{M}_1$)& Original ($\mathcal{M}_0$)& Strengthened ($\mathcal{M}_1$) \\
\hline
Strength & 0.21 & 0.42 & 0.22  & 0.45 \\
Difference in $U_{\text{tot}}$ & $1.89 \times 10^{-3}$ & $5.01 \times 10^{-3}$ & $7.96 \times 10^{-2}$ & $2.08 \times 10^{-1}$ \\
Bias & $\delta \times 9.16 \times 10^{-3}$ &  $\delta \times 1.19 \times 10^{-2}$  & $\delta \times 0.37$ & $\delta \times 0.46$ \\
Ratio of bias $\Delta$ &\multicolumn{2}{c}{$1.30> 1$} &\multicolumn{2}{c}{$1.26> 1$} \\
 \hline 
\end{tabular}
\end{table}

\subsection*{A.5: Details on constructing a finite-sample sensitivity interval using the MI paradigm}
We regard $\{U_{\text{tot},i1}, U_{\text{tot},i2}, i = 1,...,I\}$ as missing covariates and describe how to construct a finite-sample sensitivity interval using the multiple imputation (MI) paradigm. Note that the only difference between the standard MI paradigm and our procedure here is that the data-generating process of $U_{\text{tot}}$ is the sensitivity analysis model under consideration and completely known (\citealp{rubin1987multiple}; \citealp{ichino2008temporary}). Fix $(\delta, \boldsymbol{\theta}) \in \boldsymbol{I}$ and impute $\{(U_{\text{tot},i1}, U_{\text{tot},i2}), i = 1,...,I\}$ using the parametric model relating $U_{\text{tot}}$ to $\mathbf{X}$ and $\widetilde{Z}$ discussed above. For each imputed dataset $k$, Theorem~\ref{thm: asymp. Wald} implies a bias-corrected point estimate of $\beta$ is 
\begin{equation*}
   \widehat{\beta}_{\text{IV}, k} = \widehat{\beta}_{\text{IV}} - \widehat{\text{Bias}}(\widehat{\beta}_{\text{IV}})=\widehat{\beta}_{\text{IV}} - \delta \cdot \frac{\sum_{i=1}^{I}(Z_{i1}-Z_{i2})(U_{\text{tot},i1}^{(k)}-U_{\text{tot},i2}^{(k)})}{\sum_{i=1}^{I}(Z_{i1}-Z_{i2})(D_{i1}-D_{i2})}, \quad k=1,\dots, K,
\end{equation*}
where $\{(U_{\text{tot},i1}^{(k)}, U_{\text{tot},i2}^{(k)}), i = 1,...,I\}$ is the $k$-th imputed set of $U_{\text{tot}}$, and the associated variance is $\text{Var}(\widehat{\beta}_{\text{IV}, k}) = 2\sigma^2\cdot [I\cdot \widehat{\iota}^2_C]^{-1}$. According to Rubin's rule, the overall estimate for $\beta$ is $K^{-1}\sum_{k = 1}^K \widehat{\beta}_{\text{IV}, k}$, and $\text{the variance of the estimate} = \text{within-imputation variance} + \text{across-imputation variance}$ (\citealp{rubin1996multiple}). So the estimated variance of $K^{-1}\sum_{k = 1}^K \widehat{\beta}_{\text{IV}, k}$ is 
\begin{equation}\label{eqn: total variance with known variance}
    \widehat{\sigma}_{\text{total}}^{2}= \frac{1 + K^{-1}}{K - 1}\sum_{k = 1}^K \Big(\widehat{\beta}_{\text{IV}, k} - K^{-1}\sum_{k = 1}^K \widehat{\beta}_{\text{IV}, k}\Big)^2 + \frac{2\sigma^2}{I \cdot \widehat{\iota}^2_C}.
\end{equation}
Therefore, for each fixed $(\delta, \boldsymbol{\theta}) \in \boldsymbol{I}$, a $100(1-\alpha)\%$ confidence interval of $\beta$ is 
\begin{equation*}
  \text{CI}(\delta, \boldsymbol{\theta})=\Big [K^{-1}\sum_{k = 1}^K \widehat{\beta}_{\text{IV}, k}-t_{\nu}^{-1}(1-\alpha/2)\cdot \widehat{\sigma}_{\text{total}}, \ K^{-1}\sum_{k = 1}^K \widehat{\beta}_{\text{IV}, k}+t_{\nu}^{-1}(1-\alpha/2)\cdot \widehat{\sigma}_{\text{total}} \Big],
\end{equation*}
where $t_{\nu}^{-1}$ is the quantile function of the Student's t-distribution with degrees of freedom 
\begin{equation*}
    \nu = (K - 1)\Bigg [1 + \frac{2\sigma^2\cdot [I\cdot \widehat{\iota}^2_C]^{-1}}{\frac{1 + K^{-1}}{K - 1}\sum_{k = 1}^K \Big(\widehat{\beta}_{\text{IV}, k} - K^{-1}\sum_{k = 1}^K \widehat{\beta}_{\text{IV}, k}\Big)^2}\Bigg ]^2.
\end{equation*}
Then a $100(1-\alpha)\%$ sensitivity interval (SI) of $\beta$ given the sensitivity zone $\boldsymbol{I}$ is $\text{SI}=\bigcup_{(\delta, \boldsymbol{\theta}) \in \boldsymbol{I}}\text{CI}(\delta, \boldsymbol{\theta})$.

We further incorporate the uncertainty of estimating $\sigma$ in the MI paradigm by replacing $\text{Var}(\widehat{\beta}_{\text{IV}, k}) = 2\sigma^2\cdot [I \cdot \widehat{\iota}^2_C]^{-1}$ in (\ref{eqn: total variance with known variance}) with $K^{-1}\sum_{k=1}^{K}2\widehat{\sigma}_k^2\cdot [I \cdot \widehat{\iota}^2_C]^{-1}$, resulting in (\ref{eqn: total variance}), where $\widehat{\sigma}_k$ estimates $\sigma$ in the $k$-th imputed dataset under some parsimonious model of $f(\mathbf{X}_{n})$ in (\ref{eqn: Y model}).

\subsection*{A.6: Estimation of $\sigma$ in the outcome generating model (\ref{eqn: Y model})}
Under the outcome generating model (\ref{eqn: Y model}), since $\{\epsilon_{n}: n=1,\dots, N\} \ \indep \ \{(\widetilde{Z}_{n}, \mathbf{X}_{n}, U_{\text{tot},n}): n=1,\dots, N\}$, by Definition \ref{definition: matching alg is valid}, we have $\{\epsilon_{n}: n=1,\dots, N\} \ \indep \ \mathcal{M}$. Thus, $\epsilon_{T1}, \dots, \epsilon_{TI}$, $\epsilon_{C1}, \dots, \epsilon_{CI}$ are i.i.d. with expectation zero and variance $\sigma^{2}$. Thus, by the law of large numbers, we have
\begin{align*}
  \frac{1}{I} \sum_{i=1}^{I}(\epsilon_{i1}-\epsilon_{i2})^{2}= \frac{1}{I} \sum_{i=1}^{I}(\epsilon_{Ti}-\epsilon_{Ci})^{2}\xrightarrow{a.s.} 2\sigma^{2}.
\end{align*}
Conditional on exact matching $\mathbf{X}_{i1}=\mathbf{X}_{i2}$, we have 
\begin{align*}
    (\epsilon_{i1}-\epsilon_{i2})^{2}&=(R_{i1}-R_{i2}-\beta D_{i1}+\beta D_{i2}-\delta U_{\text{tot},i1}+\delta U_{\text{tot},i2})^{2}\\
    &=(R_{i1}-R_{i2})^{2}+\beta^{2}\cdot (D_{i1}-D_{i2})^{2}+\delta^{2}\cdot (U_{\text{tot},i1}-U_{\text{tot},i2})^{2}- \beta \cdot (R_{i1}-R_{i2})(D_{i1}-D_{i2})\\
    &\quad \quad -\delta \cdot (R_{i1}-R_{i2})(U_{\text{tot},i1}-U_{\text{tot},i2})+\beta \cdot \delta \cdot (D_{i1}-D_{i2})(U_{\text{tot},i1}-U_{\text{tot},i2}).
\end{align*}
Assume that under the matching algorithm $\mathcal{M}$, we have $\frac{1}{I}\sum_{i=1}^{I}(R_{i1}-R_{i2})^{2}\xrightarrow{p}\mathbb{E}_{\mathcal{M}}[(R_{i1}-R_{i2})^{2}]$, $\frac{1}{I}\sum_{i=1}^{I}(D_{i1}-D_{i2})^{2}\xrightarrow{p}\mathbb{E}_{\mathcal{M}}[(D_{i1}-D_{i2})^{2}]$, $\frac{1}{I}\sum_{i=1}^{I}(U_{\text{tot},i1}-U_{\text{tot},i2})^{2}\xrightarrow{p}\mathbb{E}_{\mathcal{M}}[(U_{\text{tot},i1}-U_{\text{tot},i2})^{2}]$, $\frac{1}{I}\sum_{i=1}^{I}(R_{i1}-R_{i2})(D_{i1}-D_{i2})\xrightarrow{p}\mathbb{E}_{\mathcal{M}}[(R_{i1}-R_{i2})(D_{i1}-D_{i2})]$, $\frac{1}{I}\sum_{i=1}^{I}(R_{i1}-R_{i2})(U_{\text{tot},i1}-U_{\text{tot},i2})\xrightarrow{p}\mathbb{E}_{\mathcal{M}}[(R_{i1}-R_{i2})(U_{\text{tot},i1}-U_{\text{tot},i2})]$, and $\frac{1}{I}\sum_{i=1}^{I}(D_{i1}-D_{i2})(U_{\text{tot},i1}-U_{\text{tot},i2})\xrightarrow{p}\mathbb{E}_{\mathcal{M}}[(D_{i1}-D_{i2})(U_{\text{tot},i1}-U_{\text{tot},i2})]$. Thus, we can estimate $2\sigma^{2}$ by
\begin{align*}
    2 \widehat{\sigma}^{2}&=\frac{1}{I}\sum_{i=1}^{I}(R_{i1}-R_{i2})^{2}+\beta^{2}\cdot \frac{1}{I}\sum_{i=1}^{I} (D_{i1}-D_{i2})^{2}+\delta^{2}\cdot \frac{1}{I}\sum_{i=1}^{I}(U_{\text{tot},i1}-U_{\text{tot},i2})^{2}\\
    &\quad \quad - \beta \cdot \frac{1}{I}\sum_{i=1}^{I}(R_{i1}-R_{i2})(D_{i1}-D_{i2})-\delta \cdot \frac{1}{I}\sum_{i=1}^{I}(R_{i1}-R_{i2})(U_{\text{tot},i1}-U_{\text{tot},i2})\\
    &\quad \quad +\beta \cdot \delta \cdot \frac{1}{I}\sum_{i=1}^{I} (D_{i1}-D_{i2})(U_{\text{tot},i1}-U_{\text{tot},i2}).
\end{align*}
However, we cannot get $2 \widehat{\sigma}^{2}$ without positing some further models. For example, note that although the last term $\frac{1}{I}\sum_{i=1}^{I} (D_{i1}-D_{i2})(U_{\text{tot},i1}-U_{\text{tot},i2})$ can be approximated by $\frac{1}{I}\sum_{i=1}^{I} (D_{i1}-D_{i2})\mathbb{E}[(U_{\text{tot},i1}-U_{\text{tot},i2})\mid D_{i1}, D_{i2}, \mathbf{X}_{i1}, \mathbf{X}_{i2}, \widetilde{Z}_{i1}, \widetilde{Z}_{i2} ]$, it requires imposing an model on $\mathbb{E}[(U_{\text{tot},i1}-U_{\text{tot},i2})\mid D_{i1}, D_{i2}, \mathbf{X}_{i1}, \mathbf{X}_{i2}, \widetilde{Z}_{i1}, \widetilde{Z}_{i2}]$ in addition to $\mathbb{E}[U_{\text{tot}}\mid \mathbf{X}, \widetilde{Z}]$ because $U_{\text{tot}} \not\!\perp\!\!\!\perp D \mid \mathbf{X}, \widetilde{Z}$.

\subsection*{A.7: A two-step debiased matching algorithm}
Our proposed matching algorithm consists of two stages. In the first stage, we construct an IV using a matching algorithm (e.g., optimal non-bipartite matching) without strengthening the IV. Denote the difference in means of each observed covariate $X_i$ in this design as $\delta_{i}$. The second stage uses the mixed integer programming (MIP) framework introduced in \citet{zubizarreta2012using}. We impose a constraint to ensure that the difference in the continuous IV formed by the second-stage matching is on average at least $k$ times as large as that in the initial IV formed by the first-stage matching, where $k\in [1,+\infty)$. In practice, we can treat $k$ as a tuning parameter and try different values of $k$ to ensure both a increase in the average paired difference of $\widetilde{Z}$ (i.e., $k$ not too small), and no dramatic decrease in the sample size (i.e., $k$ not too large). We also add constraints to make sure that the difference in means of each observed covariate $X_i$ can be no larger than $\delta_{i}$. Algorithm~\ref{alg: matching} in Supplementary Material A.8 summarizes the proposed matching algorithm. For more details of integer programming and its application to statistical matching, see Supplementary Material A.8, \citet{zubizarreta2012using}, and \citet{zubizarreta2013stronger}.

\begin{Remark}\rm
By constraining the difference in means of each observed covariate in $\mathcal{M}_1$ to be no larger than that in $\mathcal{M}_0$, we make sure no bias is introduced as a by-product of the strengthening-IV algorithm when $f(\mathbf{X})$ is linear in $\mathbf{X}$ as in (\ref{eqn: AOB due to X}). One can put further constraints on the higher moments and even marginal distributions of $X_i$ and interactions $X_i X_j$. These objectives can all be incorporated in \citet{zubizarreta2012using}'s mixed-integer-programming (MIP) framework. 
\end{Remark}

We illustrate the proposed matching algorithm using a subset of NICU data (2005 NICU study). We consider the following six covariates as a simple example: the birth weight in kilos, gestational age in weeks, gestational diabetes indicator, single birth indicator, parity, and mother's age in years. Table~\ref{tbl: compare three matching algs} contrasts the standardized difference of each of the six covariates in three designs: a first design that uses the vanilla near/far matching algorithm ($\mathcal{M}_0$), a second design that strengthens the IV by throwing away half of the samples as in \citet{baiocchi2010building} ($\mathcal{M}_1$), and a third design that uses the proposed two-stage debiased matching algorithm ($\mathcal{M}_{\text{two-stage}}$). Equation~(\ref{eqn: AOB due to X}) suggests the bias due to inexact matching in $\mathbf{X}$ is amplified by the imperfect compliance (see also \citealp{jackson2015toward}; \citealp{zhao2018graphical}). To account for this, the standardized difference of each covariate in Table \ref{tbl: compare three matching algs} is further complemented by a version of it that is normalized by the estimated compliance rate $\widehat{\iota}_{C}$. From Table~\ref{tbl: compare three matching algs}, we see some covariates have worse balance in design $\mathcal{M}_1$ compared to $\mathcal{M}_0$, before the standardized difference is normalized by $\widehat{\iota}_{C}$. Even after we take into account that $\mathcal{M}_1$ has higher $\widehat{\iota}_{C}$, many covariates are still less balanced compared to $\mathcal{M}_0$. See numbers in bold in the parentheses. On the other hand, the standardized difference of every covariate in $\mathcal{M}_{\text{two-stage}}$ is, by the design of our proposed matching algorithm, no larger than that in $\mathcal{M}_0$ before being divided by $\widehat{\iota}_{C}$. In fact, since $\widehat{\iota}_{C}$ increases in $\mathcal{M}_{\text{two-stage}}$ compared to $\mathcal{M}_0$, the normalized standardized differences are guaranteed to be smaller than those in $\mathcal{M}_0$. Meanwhile, $\mathcal{M}_{\text{two-stage}}$ significantly strengthened the IV compared to $\mathcal{M}_{0}$, with the estimated compliance rate $0.51>0.26$.

\begin{table}[H]
\centering
\caption{Covariate balance in three designs: the vanilla near/far matching algorithm $\mathcal{M}_0$, the strengthening-IV algorithm $\mathcal{M}_1$ that throws away half of data, and the two-stage debiased matching algorithm $\mathcal{M}_{\text{two-stage}}$. We only focus on the 2005 NICU data here. The three absolute standardized differences of each covariate after matching with $\mathcal{M}_{0}$, $\mathcal{M}_{1}$ and $\mathcal{M}_{\text{two-stage}}$ are reported, followed by a version of it that is normalized by the estimated compliance rate in the parenthesis.}
\label{tbl: compare three matching algs}
\begin{tabular}{lccc}
\hline
& $\mathcal{M}_0$ & $\mathcal{M}_1$ & $\mathcal{M}_{\text{two-stage}}$ \\ \hline
Estimated compliance rate  &0.26   &0.45  &0.51 \\
Covariates &           &          &          \\ 
\hspace{0.5 cm}Birth weight, g & 0.00 (0.00)     & 0.00 (0.01)     & 0.00 (0.00)    \\
\hspace{0.5 cm}Gestational age, weeks & 0.00 (0.00)      & 0.00 (0.01)     & 0.00 (0.00)     \\
\hspace{0.5 cm}Gestational diabetes, 1/0 & 0.00 (0.01)     & 0.01 (0.02)    & 0.00 (0.00)    \\ 
\hspace{0.5 cm}Single birth, 1/0  & 0.00 (0.01)     & 0.00 (0.00)   & 0.00 (0.00)    \\ 
\hspace{0.5 cm}Parity   & 0.01 (0.05)      & 0.01 (0.03)    & 0.00  (0.01)\\
\hspace{0.5 cm}Mother's age, years    & 0.00 (0.01)       & 0.00 (0.01)      & 0.00 (0.00) \\
Number of matched pairs  & 5161   & 2579  &3138 \\
\hline
\end{tabular}
\end{table}

\subsection*{A.8: An overview on matching via mixed integer programming (MIP) and a two-step debiased matching algorithm in detail}
\label{app: alg matching}
We start by briefly introducing the integer programming and its connection to statistical matching. Readers who are interested in more details should refer to \citet{zubizarreta2012using, zubizarreta2013stronger}. 

An integer program takes the form 
 \[
 \mbox{minimize}_{\boldsymbol{a}} ~\boldsymbol{\eta}^T\boldsymbol{a} \quad \mbox{subject to} \quad \boldsymbol{B}\boldsymbol{a} \leq \boldsymbol{b}, \boldsymbol{a} \geq \boldsymbol{0} \qquad \mbox{with}~\boldsymbol{a}~\mbox{integer}.
 \]
 
 In the context of statistical matching, the decision variable $\boldsymbol{a}$ is further binary. Suppose there are $L$ subjects in the study and let $a_{lm}, 1 \leq l < m \leq L$ be binary variables such that $a_{lm} = 1$ if the subject $l$ and $m$ are paired. As introduced in Zubizzareta (2012, 2013), the following constraints need to be imposed:
 \begin{enumerate}
     \item Each subject $l$ appears in at most one matched pair: \[
     \sum_{m = 1}^{l - 1} a_{ml} + \sum_{m = l+1}^L a_{lm} \leq 1;
     \]
     \item Balance the means of observed covariate $X_i$:\[
     \left|\sum_{l = 1}^{L - 1} \sum_{m = l + 1}^{L} a_{lm}v_{i, l} - \sum_{l = 1}^{L - 1} \sum_{m = l + 1}^{L} a_{lm}v_{i, m} \right| \leq \delta_i \sum_{l = 1}^{L - 1} \sum_{m = l + 1}^{L} a_{lm},
     \] where $v_{i, l}$ is the covariate $X_i$ of the subject $l$ and $\delta_i$ is the difference in means of covariate $X_i$ in design $\mathcal{M}_0$;
     \item Force pairs to differ with respect to the mean of the IV:\[
     \sum_{l = 1}^{L - 1} \sum_{m = l + 1}^{L} a_{lm} v_l - \sum_{l = 1}^{L - 1} \sum_{m = l + 1}^{L} a_{lm}v_m \geq \phi \sum_{l = 1}^{L - 1} \sum_{m = l + 1}^{L} a_{lm},
     \] where $v_l$ here represents the IV of the subject $l$ and $\phi$ encodes the strength of the strengthening-IV design.
 \end{enumerate}
 
 Finally, the objective function encodes our desire to have as many matched pairs as possible:
 \[
 \mbox{maximize} ~\sum_{l = 1}^{L - 1} \sum_{m = l + 1}^{L} a_{lm}.
 \]
 
 Other constraints, such as the fine-balance and near-fine balance constraints, can also be formulated as linear constraints and added to the above mathematical program. For detailed discussion, see \citet{zubizarreta2012using,zubizarreta2013stronger}. 

Algorithm \ref{alg: matching} summarizes the two-step matching algorithm proposed in Section~\ref{subsec: Bias due to residual imbalance after matching}.

\begin{algorithm}[ht]
\caption{A two-step debiased matching algorithm}
\label{alg: matching}
\begin{algorithmic}[1]
\STATE Apply optimal non-bipartite matching without strengthening the IV. 
\STATE Output $\delta_i$, the difference in means of observed covariates $X_1, X_2, ..., X_{p}$.
\STATE For a fixed $\phi$, solve the following integer program:
\[
\mbox{maximize} ~\sum_{l = 1}^{L - 1} \sum_{m = l + 1}^{L} a_{lm}
\]
subject to 
\begin{align*}
    \begin{split}
        &\sum_{m = 1}^{l - 1} a_{ml} + \sum_{m = l+1}^L a_{lm} \leq 1, \\
        &\left|\sum_{l = 1}^{L - 1} \sum_{m = l + 1}^{L} a_{lm}v_{i, l} - \sum_{l = 1}^{L - 1} \sum_{m = l + 1}^{L} a_{lm}v_{i, m} \right| \leq \delta_i \sum_{l = 1}^{L - 1} \sum_{m = l + 1}^{L} a_{lm}, \quad i=1,\dots,p, \\
         &\sum_{l = 1}^{L - 1} \sum_{m = l + 1}^{L} a_{lm} v_l - \sum_{l = 1}^{L - 1} \sum_{m = l + 1}^{L} a_{lm}v_m \geq \phi \sum_{l = 1}^{L - 1} \sum_{m = l + 1}^{L} a_{lm}.
    \end{split}
\end{align*}
\STATE Tune the parameter $\phi$ to balance the trade-off between the average paired difference of the IV and the sample size.
\end{algorithmic}
\end{algorithm}

\clearpage
\begin{center}
{\large\bf Supplementary Material B: Proofs}
\end{center}

\subsection*{B.1: Proof of Theorem~\ref{thm: ARE in independent case}}

\begin{Theorem}[The complete and general form of Theorem~\ref{thm: ARE in independent case}]
\label{thm: complete ARE in general case} 
Consider a sequence of testing problems consisting of a null hypothesis $H_{0}: \beta=\beta_{0}$ versus $H_{1}=\beta_{n}$. Suppose that $\beta_{n}=\beta_{0}+\Delta/\sqrt{n} + o(1/\sqrt{n})$ and without loss of generality, suppose that $\Delta>0$. Let $\pi_{I}(\cdot)$ be the power function of the one-sided Wilcoxon signed rank test (or the sign test) when testing the proportional treatment effect model in a randomized encouragement design using $I$ matched pairs. Let $I_{n}$ be the minimal number of matched pairs needed such that $\pi_{I_{n}}(\beta_0)\leq \alpha$ and $\pi_{I_{n}}(\beta_n)\geq \gamma$ for some $\alpha \in (0,1)$ and $\gamma \in (\alpha,1)$. Let $I_{n,1}$ and $I_{n,2}$ be the number of matched pairs needed for two given sequences of tests using two different valid IVs: IV I with parameters $\varrho_{1}=(\iota_{C,1},\iota_{A,1},\iota_{N,1},\mu_{C,1}, \mu_{A,1}, \mu_{N,1})$ and IV II with parameters $\varrho_{2}=(\iota_{C,2},\iota_{A,2},\iota_{N,2},\mu_{C,2}, \mu_{A,2}, \mu_{N,2})$. Let $f(x)$ be the density function of the error term $\epsilon_{i}$, and let $f^{\prime}(x)$ be the derivative function of $f(x)$. Suppose that $f$ is continuously differentiable, $\underset{x \to \infty}{\lim} f(x)=0$, and $f \in L^{2}$. We have under the Wilcoxon signed rank test, 
\begin{equation*}
    \lim_{n \rightarrow \infty}\frac{I_{n,1}}{I_{n,2}}=\frac{\psi_{\text{wilc}}^{2}(\iota_{C,2},\iota_{A,2},\iota_{N,2},\mu_{C,2}, \mu_{A,2}, \mu_{N,2})}{\psi_{\text{wilc}}^{2}(\iota_{C,1},\iota_{A,1},\iota_{N,1},\mu_{C,1}, \mu_{A,1}, \mu_{N,1})},
\end{equation*}
where 
\begingroup
\allowdisplaybreaks
\begin{align*}
	\psi_{\text{wilc}}(\iota_{C},\iota_{A},\iota_{N},\mu_{C}, \mu_{A}, \mu_{N})&=A\cdot \iota_{C}^{4}+B_{1}\cdot \iota_{C}^{3}\ \iota_{A}+B_{2} \cdot \iota_{C}^{2}\ \iota_{A}^{2}+B_{3}\cdot \iota_{C}\ \iota_{A}^{3}\\
	&\quad +C_{1}\cdot \iota_{C}^{3}\ \iota_{N}+C_{2}\cdot \iota_{C}^{2}\ \iota_{N}^{2}+C_{3}\cdot \iota_{C}\ \iota_{N}^{3}\\
	&\quad +D_{1}\cdot \iota_{C}^{2}\ \iota_{A}\ \iota_{N}+D_{2}\cdot \iota_{C}\ \iota_{A}^{2} \ \iota_{N}+D_{3}\cdot \iota_{C}\ \iota_{A}\ \iota_{N}^{2},
\end{align*}
with 
\begin{align*}
	&A=2\cdot \{f*f\}(0), ~B_{1}=6 \cdot \{f*f\}(\mu_{A}-\mu_{C}), \\
	&B_{2}=4 \cdot \{f*f\}(0)+2 \cdot \{f*f\}(2\mu_{A}-2\mu_{C}),
	~B_{3}=2\cdot \{f*f\}(\mu_{A}-\mu_{C}),  \\
	&C_{1}=6\cdot \{f*f\}(\mu_{N}-\mu_{C}), 
	~C_{2}=4\cdot \{f*f\}(0)+2\cdot \{f*f\}(2\mu_{N}-2\mu_{C}), \\
	&C_{3}=2\cdot \{f*f\}(\mu_{N}-\mu_{C}),
	~D_{1}=8 \cdot \{f*f\}(\mu_{A}-\mu_{N})+4\cdot \{f*f\}(\mu_{A}+\mu_{N}-2\mu_{C}),\\
	&D_{2}=4 \cdot \{f*f\}(\mu_{N}-\mu_{C})+2 \cdot \{f*f\}(2\mu_{A}-\mu_{N}-\mu_{C}),\\
	&D_{3}=4 \cdot \{f*f\}(\mu_{A}-\mu_{C})+2 \cdot \{f*f\}(\mu_{A}-2\mu_{N}+\mu_{C}),
\end{align*}
\endgroup
where $\{f*g \}$ denotes the convolution of two functions $f$ and $g$, that is, $\{f*g\}(x)=\int_{-\infty}^{+\infty}f(x-y)g(y)dy$. Under the same set-up but testing the null hypothesis using the sign test, we have
\begin{equation*}
    \lim_{n \rightarrow \infty}\frac{I_{n,1}}{I_{n,2}}=\frac{\psi_{\text{sign}}^{2}(\iota_{C,2},\iota_{A,2},\iota_{N,2},\mu_{C,2}, \mu_{A,2}, \mu_{N,2})}{\psi_{\text{sign}}^{2}(\iota_{C,1},\iota_{A,1},\iota_{N,1},\mu_{C,1}, \mu_{A,1}, \mu_{N,1})},
\end{equation*}
where
\begin{equation*}
    \psi_{\text{sign}} = \iota_{C}^{2} \ f(0)+\iota_{C}\ \iota_{N} \ f(\mu_{N}-\mu_{C})+\iota_{C}\ \iota_{A}\ f(-\mu_{A}+\mu_{C}).
\end{equation*}
\end{Theorem}

\begin{Lemma}\label{lemma: convolution}
Given two functions $f$ and $g$, let $\{f*g \}$ denote the convolution of $f$ and $g$, that is, $\{f*g\}(x)=\int_{-\infty}^{+\infty}f(x-y)g(y)dy$. Let $F(x)$ be a distribution function, $f(x)$ be the corresponding density function, and $f^{\prime}(x)$ be the derivative function of $f(x)$. Suppose that $f^{\prime}$ is continuous and $f \in L^{2}$. If $f(x) \rightarrow 0$ as $x \rightarrow \infty$, we have, 
\begin{equation*}
	F*f^{\prime}=f*f.
\end{equation*} 

If we additionally assume that $f(x)=f(-x)$ for any $x \in \mathbb{R}$, we have for any $x \in 
\mathbb{R}$,
\begin{equation*}
	\{F*f^{\prime}\}(x)=\{ F*f^{\prime} \}(-x), \quad \{f*f\}(x)=\{f*f\}(-x).
\end{equation*}
\end{Lemma}

\begin{proof}
For any $x \in \mathbb{R}$, since $\lim \limits_{x \rightarrow -\infty} F(x)=0$, $\lim \limits_{x \rightarrow +\infty} F(x)=1$, $\lim \limits_{x \rightarrow -\infty}f(x)=\lim \limits_{x \rightarrow +\infty}f(x)=0$, we have
   \begin{align*}
   \{f*f\}(x)&=\int_{-\infty}^{+\infty}f(x-y)f(y)dy=-F(x-y)f(y)|_{-\infty}^{+\infty}+\int_{-\infty}^{+\infty}F(x-y)f^{\prime}(y)dy\\
   &=\int_{-\infty}^{+\infty}F(x-y)f^{\prime}(y)dy=\{F*f^{\prime}\}(x).
   \end{align*}

If for any $x\in \mathbb{R}$, we have $f(x)=f(-x)$, then we have
\begin{align*}
	\{f*f\}(x)&=\int_{-\infty}^{+\infty}f(x-y)f(y)dy=\int_{-\infty}^{+\infty}f(-x+y)f(y)dy\\
	&=\int_{-\infty}^{+\infty}f(y)f(x+y)dy=\int_{-\infty}^{+\infty}f(y)f(-x-y)dy\\
	&=\{f*f\}(-x).
\end{align*}

Thus, we have
\begin{equation*}
	\{F*f^{\prime}\}(x)=\{f*f\}(x)=\{f*f\}(-x)=\{F*f^{\prime}\}(-x).
\end{equation*}

\end{proof}

\begin{proof}[Proof of Theorem~\ref{thm: complete ARE in general case}] 
\label{app subsec: proof of power wilcoxon}
Under the setting as in Section~\ref{subsec: Quantitative evaluation of the trade-off between IV strength and sample size}, let $G$ be the distribution function of $Y_{i}-\beta_{0}D_{i}$. As described in \cite{ertefaie2018quantitative}, we have $G(x)=\sum_{i=1}^{8}G_{i}(x)$ where
	\begin{align*}
		&G_{1}(x)=\iota_{C}^{2} \ F(x-\beta+\beta_{0}),\\
		&G_{2}(x)=\iota_{C}\ \iota_{N} \ F(x+\mu_{N}-\mu_{C}-\beta+\beta_{0}),\\
		&G_{3}(x)=\iota_{C}\ \iota_{N}\ F(x-\mu_{N}+\mu_{C}),\\
		&G_{4}(x)=\iota_{C}\ \iota_{A}\ F(x-\mu_{A}+\mu_{C}-\beta+\beta_{0}),\\
		&G_{5}(x)=\iota_{C}\ \iota_{A}\ F(x+\mu_{A}-\mu_{C}),\\		
		&G_{6}(x)=(\iota_{A}^{2}+\iota_{N}^{2}) \ F(x),\\
		&G_{7}(x)=\iota_{A}\ \iota_{N} \ F(x-\mu_{A}+\mu_{N}-\beta+\beta_{0}),\\
		&G_{8}(x)=\iota_{A}\ \iota_{N} \ F(x-\mu_{N}+\mu_{A}+\beta-\beta_{0}). 
	\end{align*}
Let $g(x)=G^{\prime}(x)=\sum_{i=1}^{8}g_{i}(x)$, where $g_{i}(x)=G_{i}^{\prime}(x)$ for all $i$. Following Example 3.3.6 in \cite{lehmann2004elements}, we just need to figure out $\varphi^{\prime}(\beta_{0})$, the derivative of $\varphi(\beta)=P_{\beta}(X+Y>0)$ at $\beta=\beta_{0}$, where $X \sim G$, $Y \sim G$, $X$ and $Y$ are independent. Then the ARE under the Wilcoxon signed rank test is the ratio of $\{\varphi^{\prime}(\beta_{0})\}^{2}$ under $\varrho_{1}=(\iota_{C,1},\iota_{A,1},\iota_{N,1},\mu_{C,1}, \mu_{A,1}, \mu_{N,1})$ and $\varrho_{2}=(\iota_{C,2},\iota_{A,2},\iota_{N,2},\mu_{C,2}, \mu_{A,2}, \mu_{N,2})$. We have
    \begin{align*}
    	1-\varphi(\beta)&=P_{\beta}(X+Y\leq 0) = \iint_{x+y\leq 0}g(x,y)dydx\\
    	 &=\int_{-\infty}^{+\infty}\Big[ \int_{-\infty}^{-x}g(y)dy \Big] g(x)dx =\int_{-\infty}^{+\infty}G(-x)g(x)dx\\
    	 &=\sum_{i=1}^{8}\sum_{j=1}^{8}\int_{-\infty}^{+\infty}G_{i}(-x)g_{j}(x)dx.
    \end{align*}
    Thus, we have,
    \begin{align*}
    	\varphi^{\prime}(\beta_{0})
    	&= -\sum_{i=1}^{8}\sum_{j=1}^{8}\frac{\partial}{\partial \beta} \int_{-\infty}^{+\infty}G_{i}(-x)g_{j}(x)dx \Big|_{\beta=\beta_{0}}   \\
    	&=\sum_{i=1}^{8}\sum_{j=1}^{8}\Big(- \int_{-\infty}^{+\infty}\frac{\partial}{\partial \beta} \Big[G_{i}(-x)g_{j}(x)\Big]dx \Big|_{\beta=\beta_{0}} \Big).
    \end{align*}
    
    \begingroup
\allowdisplaybreaks
As discussed in Section~\ref{subsec: Quantitative evaluation of the trade-off between IV strength and sample size}, if the IV is valid, then $F(\cdot)$ is symmetric about zero. We have
    \begin{align*}
    	&-\int_{-\infty}^{+\infty}\frac{\partial}{\partial \beta} \Big[G_{1}(-x)g_{1}(x)\Big]dx \Big|_{\beta=\beta_{0}}=\iota_{C}^{4}\ \Big(\int_{-\infty}^{+\infty}f^{2}(x)dx+\int_{-\infty}^{+\infty}F(-x)f^{\prime}(x)dx \Big),	\\
    	&-\int_{-\infty}^{+\infty}\frac{\partial}{\partial \beta} \Big[G_{1}(-x)g_{2}(x)\Big]dx \Big|_{\beta=\beta_{0}}=\iota_{C}^{3}\ \iota_{N}\  \Big(\int_{-\infty}^{+\infty}f(x)f(x+\mu_{N}-\mu_{C})dx\\
    	&\qquad \qquad \qquad \qquad \qquad \qquad \qquad \qquad \qquad \qquad \qquad +\int_{-\infty}^{+\infty}F(-x)f^{\prime}(x+\mu_{N}-\mu_{C})dx \Big), \\
    	&-\int_{-\infty}^{+\infty}\frac{\partial}{\partial \beta} \Big[G_{1}(-x)g_{3}(x)\Big]dx \Big|_{\beta=\beta_{0}}=\iota_{C}^{3}\ \iota_{N} \  \Big(\int_{-\infty}^{+\infty}f(x)f(x-\mu_{N}+\mu_{C})dx\Big), \\
    	&-\int_{-\infty}^{+\infty}\frac{\partial}{\partial \beta} \Big[G_{1}(-x)g_{4}(x)\Big]dx \Big|_{\beta=\beta_{0}}=\iota_{C}^{3}\ \iota_{A} \  \Big(\int_{-\infty}^{+\infty}f(x)f(x-\mu_{A}+\mu_{C})dx\\
    	&\qquad \qquad \qquad \qquad \qquad \qquad \qquad \qquad \qquad \qquad \qquad +\int_{-\infty}^{+\infty}F(-x)f^{\prime}(x-\mu_{A}+\mu_{C})dx \Big),\\
    	&-\int_{-\infty}^{+\infty}\frac{\partial}{\partial \beta} \Big[G_{1}(-x)g_{5}(x)\Big]dx \Big|_{\beta=\beta_{0}}=\iota_{C}^{3}\ \iota_{A} \ \Big(\int_{-\infty}^{+\infty}f(x)f(x+\mu_{A}-\mu_{C})dx \Big),\\
    	&-\int_{-\infty}^{+\infty}\frac{\partial}{\partial \beta} \Big[G_{1}(-x)g_{6}(x)\Big]dx \Big|_{\beta=\beta_{0}}=\iota_{C}^{2}\ (\iota_{A}^{2}+\iota_{N}^{2}) \ \Big(\int_{-\infty}^{+\infty}f^{2}(x)dx\Big), \\
      	&-\int_{-\infty}^{+\infty}\frac{\partial}{\partial \beta} \Big[G_{1}(-x)g_{7}(x)\Big]dx \Big|_{\beta=\beta_{0}}=\iota_{C}^{2}\ \iota_{A} \ \iota_{N}  \Big(\int_{-\infty}^{+\infty}f(x)f(x-\mu_{A}+\mu_{N})dx\\
    	&\qquad \qquad \qquad \qquad \qquad \qquad \qquad \qquad \qquad \qquad \qquad
    	+\int_{-\infty}^{+\infty}F(-x)f^{\prime}(x-\mu_{A}+\mu_{N})dx \Big),\\
    	&-\int_{-\infty}^{+\infty}\frac{\partial}{\partial \beta} \Big[G_{1}(-x)g_{8}(x)\Big]dx \Big|_{\beta=\beta_{0}}=\iota_{C}^{2}\ \iota_{A}\ \iota_{N} \ \Big(\int_{-\infty}^{+\infty}f(x)f(x-\mu_{N}+\mu_{A})dx\\
    	&\qquad \qquad \qquad \qquad \qquad \qquad \qquad \qquad \qquad \qquad \qquad
    	-\int_{-\infty}^{+\infty}F(-x)f^{\prime}(x-\mu_{N}+\mu_{A})dx \Big),\\
    	&-\int_{-\infty}^{+\infty}\frac{\partial}{\partial \beta} \Big[G_{2}(-x)g_{1}(x)\Big]dx \Big|_{\beta=\beta_{0}}=\iota_{C}^{3}\ \iota_{N} \ \Big(\int_{-\infty}^{+\infty}f(x-\mu_{N}+\mu_{C})f(x)dx\\
    	&\qquad \qquad \qquad \qquad \qquad \qquad \qquad \qquad \qquad \qquad \qquad
    	+\int_{-\infty}^{+\infty}F(-x+\mu_{N}-\mu_{C})f^{\prime}(x)dx \Big), \\
    	&-\int_{-\infty}^{+\infty}\frac{\partial}{\partial \beta} \Big[G_{2}(-x)g_{2}(x)\Big]dx \Big|_{\beta=\beta_{0}}=\iota_{C}^{2}\ \iota_{N}^{2} \  \Big(\int_{-\infty}^{+\infty}f(x-\mu_{N}+\mu_{C})f(x+\mu_{N}-\mu_{C})dx\\
    	&\qquad \qquad \qquad \qquad \qquad \qquad \qquad \qquad \qquad \qquad \qquad
    	+\int_{-\infty}^{+\infty}F(-x+\mu_{N}-\mu_{C})f^{\prime}(x+\mu_{N}-\mu_{C})dx \Big),\\
    	&-\int_{-\infty}^{+\infty}\frac{\partial}{\partial \beta} \Big[G_{2}(-x)g_{3}(x)\Big]dx \Big|_{\beta=\beta_{0}}=\iota_{C}^{2}\ \iota_{N}^{2} \  \Big(\int_{-\infty}^{+\infty}f^{2}(x-\mu_{N}+\mu_{C})dx \Big),\\
    	&-\int_{-\infty}^{+\infty}\frac{\partial}{\partial \beta} \Big[G_{2}(-x)g_{4}(x)\Big]dx \Big|_{\beta=\beta_{0}}=\iota_{C}^{2}\ \iota_{A} \ \iota_{N}\ \Big(\int_{-\infty}^{+\infty}f(x-\mu_{N}+\mu_{C})f(x-\mu_{A}+\mu_{C})dx\\
    	&\qquad \qquad \qquad \qquad \qquad \qquad \qquad \qquad \qquad \qquad \qquad
    	+\int_{-\infty}^{+\infty}F(-x+\mu_{N}-\mu_{C})f^{\prime}(x-\mu_{A}+\mu_{C})dx \Big),\\
    	&-\int_{-\infty}^{+\infty}\frac{\partial}{\partial \beta} \Big[G_{2}(-x)g_{5}(x)\Big]dx \Big|_{\beta=\beta_{0}}=\iota_{C}^{2}\ \iota_{A}\ \iota_{N} \ \Big(\int_{-\infty}^{+\infty}f(x-\mu_{N}+\mu_{C})f(x+\mu_{A}-\mu_{C})dx \Big),\\
       	&-\int_{-\infty}^{+\infty}\frac{\partial}{\partial \beta} \Big[G_{2}(-x)g_{6}(x)\Big]dx \Big|_{\beta=\beta_{0}}=(\iota_{C}\ \iota_{A}^{2}\ \iota_{N}+\iota_{C}\ \iota_{N}^{3})\  \Big(\int_{-\infty}^{+\infty}f(x-\mu_{N}+\mu_{C})f(x)dx \Big),\\
    	&-\int_{-\infty}^{+\infty}\frac{\partial}{\partial \beta} \Big[G_{2}(-x)g_{7}(x)\Big]dx \Big|_{\beta=\beta_{0}}=\iota_{C} \ \iota_{A} \ \iota_{N}^{2} \ \Big( \int_{-\infty}^{+\infty}f(x-\mu_{N}+\mu_{C})f(x-\mu_{A}+\mu_{N})dx\\
    	&\qquad \qquad \qquad \qquad \qquad \qquad \qquad \qquad \qquad \qquad \qquad
    	+\int_{-\infty}^{+\infty}F(-x+\mu_{N}-\mu_{C})f^{\prime}(x-\mu_{A}+\mu_{N})dx\Big), \\
    	&-\int_{-\infty}^{+\infty}\frac{\partial}{\partial \beta} \Big[G_{2}(-x)g_{8}(x)\Big]dx \Big|_{\beta=\beta_{0}}=\iota_{C} \ \iota_{A} \ \iota_{N}^{2} \ \Big( \int_{-\infty}^{+\infty}f(x-\mu_{N}+\mu_{C})f(x-\mu_{N}+\mu_{A})dx\\
    	&\qquad \qquad \qquad \qquad \qquad \qquad \qquad \qquad \qquad \qquad \qquad
    	-\int_{-\infty}^{+\infty}F(-x+\mu_{N}-\mu_{C})f^{\prime}(x-\mu_{N}+\mu_{A})dx\Big), \\
       	&-\int_{-\infty}^{+\infty}\frac{\partial}{\partial \beta} \Big[G_{3}(-x)g_{1}(x)\Big]dx \Big|_{\beta=\beta_{0}}=\iota_{C}^{3}\ \iota_{N} \ \Big(\int_{-\infty}^{+\infty}F(-x-\mu_{N}+\mu_{C})f^{\prime}(x)dx \Big),\\
    	&-\int_{-\infty}^{+\infty}\frac{\partial}{\partial \beta} \Big[G_{3}(-x)g_{2}(x)\Big]dx \Big|_{\beta=\beta_{0}}=\iota_{C}^{2} \ \iota_{N}^{2} \ \Big(\int_{-\infty}^{+\infty}F(-x-\mu_{N}+\mu_{C})f^{\prime}(x+\mu_{N}-\mu_{C})dx \Big),\\
    	&-\int_{-\infty}^{+\infty}\frac{\partial}{\partial \beta} \Big[G_{3}(-x)g_{3}(x)\Big]dx \Big|_{\beta=\beta_{0}}=0,\\
       	&-\int_{-\infty}^{+\infty}\frac{\partial}{\partial \beta} \Big[G_{3}(-x)g_{4}(x)\Big]dx \Big|_{\beta=\beta_{0}}=\iota_{C}^{2}\ \iota_{A}\ \iota_{N}\ \Big(\int_{-\infty}^{+\infty}F(-x-\mu_{N}+\mu_{C})f^{\prime}(x-\mu_{A}+\mu_{C})dx \Big),\\    
    	&-\int_{-\infty}^{+\infty}\frac{\partial}{\partial \beta} \Big[G_{3}(-x)g_{5}(x)\Big]dx \Big|_{\beta=\beta_{0}}=0, \\
    	&-\int_{-\infty}^{+\infty}\frac{\partial}{\partial \beta} \Big[G_{3}(-x)g_{6}(x)\Big]dx \Big|_{\beta=\beta_{0}}=0,\\
    	&-\int_{-\infty}^{+\infty}\frac{\partial}{\partial \beta} \Big[G_{3}(-x)g_{7}(x)\Big]dx \Big|_{\beta=\beta_{0}}=\iota_{C}\ \iota_{A}\ \iota_{N}^{2}\ \Big(\int_{-\infty}^{+\infty}F(-x-\mu_{N}+\mu_{C})f^{\prime}(x-\mu_{A}+\mu_{N})dx \Big),\\
    	&-\int_{-\infty}^{+\infty}\frac{\partial}{\partial \beta} \Big[G_{3}(-x)g_{8}(x)\Big]dx \Big|_{\beta=\beta_{0}}=\iota_{C}\ \iota_{A}\ \iota_{N}^{2}\ \Big(-\int_{-\infty}^{+\infty}F(-x-\mu_{N}+\mu_{C})f^{\prime}(x-\mu_{N}+\mu_{A})dx \Big),\\
    	&-\int_{-\infty}^{+\infty}\frac{\partial}{\partial \beta} \Big[G_{4}(-x)g_{1}(x)\Big]dx \Big|_{\beta=\beta_{0}}=\iota_{C}^{3}\ \iota_{A}\ \Big(\int_{-\infty}^{+\infty}f(x+\mu_{A}-\mu_{C})f(x)dx\\
    	&\qquad \qquad \qquad \qquad \qquad \qquad \qquad \qquad \qquad \qquad \qquad
    	+\int_{-\infty}^{+\infty}F(-x-\mu_{A}+\mu_{C})f^{\prime}(x)dx \Big),	\\
    	&-\int_{-\infty}^{+\infty}\frac{\partial}{\partial \beta} \Big[G_{4}(-x)g_{2}(x)\Big]dx \Big|_{\beta=\beta_{0}}=\iota_{C}^{2}\ \iota_{A}\ \iota_{N}\  \Big(\int_{-\infty}^{+\infty}f(x+\mu_{A}-\mu_{C})f(x+\mu_{N}-\mu_{C})dx\\
    	&\qquad \qquad \qquad \qquad \qquad \qquad \qquad \qquad \qquad \qquad \qquad
    	+\int_{-\infty}^{+\infty}F(-x-\mu_{A}+\mu_{C})f^{\prime}(x+\mu_{N}-\mu_{C})dx \Big), \\
    	&-\int_{-\infty}^{+\infty}\frac{\partial}{\partial \beta} \Big[G_{4}(-x)g_{3}(x)\Big]dx, \Big|_{\beta=\beta_{0}}=\iota_{C}^{2}\ \iota_{A}\ \iota_{N}\  \Big(\int_{-\infty}^{+\infty}f(x+\mu_{A}-\mu_{C})f(x-\mu_{N}+\mu_{C})dx\Big), \\
    	&-\int_{-\infty}^{+\infty}\frac{\partial}{\partial \beta} \Big[G_{4}(-x)g_{4}(x)\Big]dx \Big|_{\beta=\beta_{0}}=\iota_{C}^{2}\ \iota_{A}^{2} \  \Big(\int_{-\infty}^{+\infty}f(x+\mu_{A}-\mu_{C})f(x-\mu_{A}+\mu_{C})dx,\\
    	&\qquad \qquad \qquad \qquad \qquad \qquad \qquad \qquad \qquad \qquad \qquad
    	+\int_{-\infty}^{+\infty}F(-x-\mu_{A}+\mu_{C})f^{\prime}(x-\mu_{A}+\mu_{C})dx \Big),\\
    	&-\int_{-\infty}^{+\infty}\frac{\partial}{\partial \beta} \Big[G_{4}(-x)g_{5}(x)\Big]dx \Big|_{\beta=\beta_{0}}=\iota_{C}^{2}\ \iota_{A}^{2} \ \Big(\int_{-\infty}^{+\infty}f^{2}(x+\mu_{A}-\mu_{C})dx \Big),\\
    	&-\int_{-\infty}^{+\infty}\frac{\partial}{\partial \beta} \Big[G_{4}(-x)g_{6}(x)\Big]dx \Big|_{\beta=\beta_{0}}=(\iota_{C}\ \iota_{A}^{3}+\iota_{C}\ \iota_{A}\ \iota_{N}^{2} ) \ \Big(\int_{-\infty}^{+\infty}f(x+\mu_{A}-\mu_{C})f(x)dx\Big), \\
      	&-\int_{-\infty}^{+\infty}\frac{\partial}{\partial \beta} \Big[G_{4}(-x)g_{7}(x)\Big]dx \Big|_{\beta=\beta_{0}}=\iota_{C}\ \iota_{A}^{2} \ \iota_{N}  \Big(\int_{-\infty}^{+\infty}f(x+\mu_{A}-\mu_{C})f(x-\mu_{A}+\mu_{N})dx\\
    	&\qquad \qquad \qquad \qquad \qquad \qquad \qquad \qquad \qquad \qquad \qquad
    	+\int_{-\infty}^{+\infty}F(-x-\mu_{A}+\mu_{C})f^{\prime}(x-\mu_{A}+\mu_{N})dx \Big),\\
    	&-\int_{-\infty}^{+\infty}\frac{\partial}{\partial \beta} \Big[G_{4}(-x)g_{8}(x)\Big]dx \Big|_{\beta=\beta_{0}}=\iota_{C}\ \iota_{A}^{2} \ \iota_{N}  \Big(\int_{-\infty}^{+\infty}f(x+\mu_{A}-\mu_{C})f(x-\mu_{N}+\mu_{A})dx\\
    	&\qquad \qquad \qquad \qquad \qquad \qquad \qquad \qquad \qquad \qquad \qquad
    	-\int_{-\infty}^{+\infty}F(-x-\mu_{A}+\mu_{C})f^{\prime}(x-\mu_{N}+\mu_{A})dx \Big), \\
    	&-\int_{-\infty}^{+\infty}\frac{\partial}{\partial \beta} \Big[G_{5}(-x)g_{1}(x)\Big]dx \Big|_{\beta=\beta_{0}}=\iota_{C}^{3}\ \iota_{A} \ \Big(\int_{-\infty}^{+\infty}F(-x+\mu_{A}-\mu_{C})f^{\prime}(x)dx \Big), \\
    	&-\int_{-\infty}^{+\infty}\frac{\partial}{\partial \beta} \Big[G_{5}(-x)g_{2}(x)\Big]dx \Big|_{\beta=\beta_{0}}=\iota_{C}^{2}\ \iota_{A}\ \iota_{N} \  \Big(\int_{-\infty}^{+\infty}F(-x+\mu_{A}-\mu_{C})f^{\prime}(x+\mu_{N}-\mu_{C})dx \Big),\\
    	&-\int_{-\infty}^{+\infty}\frac{\partial}{\partial \beta} \Big[G_{5}(-x)g_{3}(x)\Big]dx \Big|_{\beta=\beta_{0}}=0,\\
    	&-\int_{-\infty}^{+\infty}\frac{\partial}{\partial \beta} \Big[G_{5}(-x)g_{4}(x)\Big]dx \Big|_{\beta=\beta_{0}}=\iota_{C}^{2}\ \iota_{A}^{2}\ \Big(\int_{-\infty}^{+\infty}F(-x+\mu_{A}-\mu_{C})f^{\prime}(x-\mu_{A}+\mu_{C})dx \Big),\\
    	&-\int_{-\infty}^{+\infty}\frac{\partial}{\partial \beta} \Big[G_{5}(-x)g_{5}(x)\Big]dx \Big|_{\beta=\beta_{0}}=0,\\
       	&-\int_{-\infty}^{+\infty}\frac{\partial}{\partial \beta} \Big[G_{5}(-x)g_{6}(x)\Big]dx \Big|_{\beta=\beta_{0}}=0,\\
    	&-\int_{-\infty}^{+\infty}\frac{\partial}{\partial \beta} \Big[G_{5}(-x)g_{7}(x)\Big]dx \Big|_{\beta=\beta_{0}}=\iota_{C} \ \iota_{A}^{2} \ \iota_{N} \ \Big( \int_{-\infty}^{+\infty}F(-x+\mu_{A}-\mu_{C})f^{\prime}(x-\mu_{A}+\mu_{N})dx\Big), \\
    	&-\int_{-\infty}^{+\infty}\frac{\partial}{\partial \beta} \Big[G_{5}(-x)g_{8}(x)\Big]dx \Big|_{\beta=\beta_{0}}=\iota_{C} \ \iota_{A}^{2} \ \iota_{N} \ \Big(-\int_{-\infty}^{+\infty}F(-x+\mu_{A}-\mu_{C})f^{\prime}(x-\mu_{N}+\mu_{A})dx\Big), \\
       	&-\int_{-\infty}^{+\infty}\frac{\partial}{\partial \beta} \Big[G_{6}(-x)g_{1}(x)\Big]dx \Big|_{\beta=\beta_{0}}=(\iota_{C}^{2}\ \iota_{A}^{2}+ \iota_{C}^{2} \ \iota_{N}^{2}) \Big(\int_{-\infty}^{+\infty}F(-x)f^{\prime}(x)dx \Big),\\
    	&-\int_{-\infty}^{+\infty}\frac{\partial}{\partial \beta} \Big[G_{6}(-x)g_{2}(x)\Big]dx \Big|_{\beta=\beta_{0}}=(\iota_{C}\ \iota_{A}^{2} \ \iota_{N}+ \iota_{C}\ \iota_{N}^{3}) \ \Big(\int_{-\infty}^{+\infty}F(-x)f^{\prime}(x+\mu_{N}-\mu_{C})dx \Big),\\
    	&-\int_{-\infty}^{+\infty}\frac{\partial}{\partial \beta} \Big[G_{6}(-x)g_{3}(x)\Big]dx \Big|_{\beta=\beta_{0}}=0,\\
       	&-\int_{-\infty}^{+\infty}\frac{\partial}{\partial \beta} \Big[G_{6}(-x)g_{4}(x)\Big]dx \Big|_{\beta=\beta_{0}}=(\iota_{C}\ \iota_{A}^{3}+ \iota_{C}\ \iota_{A}\ \iota_{N}^{2}) \Big(\int_{-\infty}^{+\infty}F(-x)f^{\prime}(x-\mu_{A}+\mu_{C})dx \Big),\\    
    	&-\int_{-\infty}^{+\infty}\frac{\partial}{\partial \beta} \Big[G_{6}(-x)g_{5}(x)\Big]dx \Big|_{\beta=\beta_{0}}=0, \\
    	&-\int_{-\infty}^{+\infty}\frac{\partial}{\partial \beta} \Big[G_{6}(-x)g_{6}(x)\Big]dx \Big|_{\beta=\beta_{0}}=0,\\
    	&-\int_{-\infty}^{+\infty}\frac{\partial}{\partial \beta} \Big[G_{6}(-x)g_{7}(x)\Big]dx \Big|_{\beta=\beta_{0}}=(\iota_{A}^{3} \ \iota_{N}+ \iota_{A}\ \iota_{N}^{3})\Big(\int_{-\infty}^{+\infty}F(-x)f^{\prime}(x-\mu_{A}+\mu_{N})dx \Big),\\
    	&-\int_{-\infty}^{+\infty}\frac{\partial}{\partial \beta} \Big[G_{6}(-x)g_{8}(x)\Big]dx \Big|_{\beta=\beta_{0}}=(\iota_{A}^{3} \ \iota_{N}+ \iota_{A}\ \iota_{N}^{3})\Big(-\int_{-\infty}^{+\infty}F(-x)f^{\prime}(x-\mu_{N}+\mu_{A})dx \Big), \\
        &-\int_{-\infty}^{+\infty}\frac{\partial}{\partial \beta} \Big[G_{7}(-x)g_{1}(x)\Big]dx \Big|_{\beta=\beta_{0}}=\iota_{C}^{2}\ \iota_{A}\ \iota_{N} \ \Big(\int_{-\infty}^{+\infty}f(x+\mu_{A}-\mu_{N})f(x)dx\\
    	&\qquad \qquad \qquad \qquad \qquad \qquad \qquad \qquad \qquad \qquad \qquad
    	+\int_{-\infty}^{+\infty}F(-x-\mu_{A}+\mu_{N})f^{\prime}(x)dx \Big),\\
    	&-\int_{-\infty}^{+\infty}\frac{\partial}{\partial \beta} \Big[G_{7}(-x)g_{2}(x)\Big]dx \Big|_{\beta=\beta_{0}}=\iota_{C}\ \iota_{A}\ \iota_{N}^{2} \ \Big(\int_{-\infty}^{+\infty}f(x+\mu_{A}-\mu_{N})f(x+\mu_{N}-\mu_{C})dx\\
    	&\qquad \qquad \qquad \qquad \qquad \qquad \qquad \qquad \qquad \qquad \qquad
    	+\int_{-\infty}^{+\infty}F(-x-\mu_{A}+\mu_{N})f^{\prime}(x+\mu_{N}-\mu_{C})dx \Big), \\
    	&-\int_{-\infty}^{+\infty}\frac{\partial}{\partial \beta} \Big[G_{7}(-x)g_{3}(x)\Big]dx \Big|_{\beta=\beta_{0}}=\iota_{C}\ \iota_{A} \ \iota_{N}^{2}\ \Big(\int_{-\infty}^{+\infty}f(x+\mu_{A}-\mu_{N})f(x-\mu_{N}+\mu_{C})dx\Big), \\
    	&-\int_{-\infty}^{+\infty}\frac{\partial}{\partial \beta} \Big[G_{7}(-x)g_{4}(x)\Big]dx \Big|_{\beta=\beta_{0}}=\iota_{C}\ \iota_{A}^{2}\ \iota_{N}\ \Big(\int_{-\infty}^{+\infty}f(x+\mu_{A}-\mu_{N})f(x-\mu_{A}+\mu_{C})dx\\
    	&\qquad \qquad \qquad \qquad \qquad \qquad \qquad \qquad \qquad \qquad \qquad
    	+\int_{-\infty}^{+\infty}F(-x-\mu_{A}+\mu_{N})f^{\prime}(x-\mu_{A}+\mu_{C})dx \Big),\\
    	&-\int_{-\infty}^{+\infty}\frac{\partial}{\partial \beta} \Big[G_{7}(-x)g_{5}(x)\Big]dx \Big|_{\beta=\beta_{0}}=\iota_{C}\ \iota_{A}^{2} \ \iota_{N} \ \Big(\int_{-\infty}^{+\infty}f(x+\mu_{A}-\mu_{N})f(x+\mu_{A}-\mu_{C})dx\Big),\\
     	&-\int_{-\infty}^{+\infty}\frac{\partial}{\partial \beta} \Big[G_{7}(-x)g_{6}(x)\Big]dx \Big|_{\beta=\beta_{0}}=(\iota_{A}^{3}\ \iota_{N}+\iota_{A}\ \iota_{N}^{3}) \ \Big(\int_{-\infty}^{+\infty}f(x+\mu_{A}-\mu_{N})f(x)dx\Big), \\
        &-\int_{-\infty}^{+\infty}\frac{\partial}{\partial \beta} \Big[G_{7}(-x)g_{7}(x)\Big]dx \Big|_{\beta=\beta_{0}}=\iota_{A}^{2}\ \iota_{N}^{2} \ \Big(\int_{-\infty}^{+\infty}f(x+\mu_{A}-\mu_{N})f(x-\mu_{A}+\mu_{N})dx\\
    	&\qquad \qquad \qquad \qquad \qquad \qquad \qquad \qquad \qquad \qquad \qquad
    	+\int_{-\infty}^{+\infty}F(-x-\mu_{A}+\mu_{N})f^{\prime}(x-\mu_{A}+\mu_{N})dx \Big),\\
    	&-\int_{-\infty}^{+\infty}\frac{\partial}{\partial \beta} \Big[G_{7}(-x)g_{8}(x)\Big]dx \Big|_{\beta=\beta_{0}}=\iota_{A}^{2}\ \iota_{N}^{2} \ \Big(\int_{-\infty}^{+\infty}f(x+\mu_{A}-\mu_{N})f(x-\mu_{N}+\mu_{A})dx\\
    	&\qquad \qquad \qquad \qquad \qquad \qquad \qquad \qquad \qquad \qquad \qquad
    	-\int_{-\infty}^{+\infty}F(-x-\mu_{A}+\mu_{N})f^{\prime}(x-\mu_{N}+\mu_{A})dx \Big), \\
    	&-\int_{-\infty}^{+\infty}\frac{\partial}{\partial \beta} \Big[G_{8}(-x)g_{1}(x)\Big]dx \Big|_{\beta=\beta_{0}}=\iota_{C}^{2}\ \iota_{A}\ \iota_{N} \ \Big(-\int_{-\infty}^{+\infty}f(x+\mu_{N}-\mu_{A})f(x)dx\\
    	&\qquad \qquad \qquad \qquad \qquad \qquad \qquad \qquad \qquad \qquad \qquad
    	+\int_{-\infty}^{+\infty}F(-x-\mu_{N}+\mu_{A})f^{\prime}(x)dx \Big),\\
    	&-\int_{-\infty}^{+\infty}\frac{\partial}{\partial \beta} \Big[G_{8}(-x)g_{2}(x)\Big]dx \Big|_{\beta=\beta_{0}}=\iota_{C}\ \iota_{A}\ \iota_{N}^{2} \ \Big(-\int_{-\infty}^{+\infty}f(x+\mu_{N}-\mu_{A})f(x+\mu_{N}-\mu_{C})dx\\
    	&\qquad \qquad \qquad \qquad \qquad \qquad \qquad \qquad \qquad \qquad \qquad
    	+\int_{-\infty}^{+\infty}F(-x-\mu_{N}+\mu_{A})f^{\prime}(x+\mu_{N}-\mu_{C})dx \Big), \\
    	&-\int_{-\infty}^{+\infty}\frac{\partial}{\partial \beta} \Big[G_{8}(-x)g_{3}(x)\Big]dx \Big|_{\beta=\beta_{0}}=\iota_{C}\ \iota_{A} \ \iota_{N}^{2}\ \Big(-\int_{-\infty}^{+\infty}f(x+\mu_{N}-\mu_{A})f(x-\mu_{N}+\mu_{C})dx\Big), \\
    	&-\int_{-\infty}^{+\infty}\frac{\partial}{\partial \beta} \Big[G_{8}(-x)g_{4}(x)\Big]dx \Big|_{\beta=\beta_{0}}=\iota_{C}\ \iota_{A}^{2}\ \iota_{N}\ \Big(-\int_{-\infty}^{+\infty}f(x+\mu_{N}-\mu_{A})f(x-\mu_{A}+\mu_{C})dx\\
    	&\qquad \qquad \qquad \qquad \qquad \qquad \qquad \qquad \qquad \qquad \qquad
    	+\int_{-\infty}^{+\infty}F(-x-\mu_{N}+\mu_{A})f^{\prime}(x-\mu_{A}+\mu_{C})dx \Big),\\
    	&-\int_{-\infty}^{+\infty}\frac{\partial}{\partial \beta} \Big[G_{8}(-x)g_{5}(x)\Big]dx \Big|_{\beta=\beta_{0}}=\iota_{C}\ \iota_{A}^{2} \ \iota_{N} \ \Big(-\int_{-\infty}^{+\infty}f(x+\mu_{N}-\mu_{A})f(x+\mu_{A}-\mu_{C})dx\Big),\\
     	&-\int_{-\infty}^{+\infty}\frac{\partial}{\partial \beta} \Big[G_{8}(-x)g_{6}(x)\Big]dx \Big|_{\beta=\beta_{0}}=(\iota_{A}^{3}\ \iota_{N}+\iota_{A}\ \iota_{N}^{3}) \ \Big(-\int_{-\infty}^{+\infty}f(x+\mu_{N}-\mu_{A})f(x)dx\Big), \\
        &-\int_{-\infty}^{+\infty}\frac{\partial}{\partial \beta} \Big[G_{8}(-x)g_{7}(x)\Big]dx \Big|_{\beta=\beta_{0}}=\iota_{A}^{2}\ \iota_{N}^{2} \ \Big(-\int_{-\infty}^{+\infty}f(x+\mu_{N}-\mu_{A})f(x-\mu_{A}+\mu_{N})dx\\
    	&\qquad \qquad \qquad \qquad \qquad \qquad \qquad \qquad \qquad \qquad \qquad
    	+\int_{-\infty}^{+\infty}F(-x-\mu_{N}+\mu_{A})f^{\prime}(x-\mu_{A}+\mu_{N})dx \Big),\\
    	&-\int_{-\infty}^{+\infty}\frac{\partial}{\partial \beta} \Big[G_{8}(-x)g_{8}(x)\Big]dx \Big|_{\beta=\beta_{0}}=\iota_{A}^{2}\ \iota_{N}^{2} \ \Big(-\int_{-\infty}^{+\infty}f(x+\mu_{N}-\mu_{A})f(x-\mu_{N}+\mu_{A})dx\\
    	&\qquad \qquad \qquad \qquad \qquad \qquad \qquad \qquad \qquad \qquad \qquad
    	-\int_{-\infty}^{+\infty}F(-x-\mu_{N}+\mu_{A})f^{\prime}(x-\mu_{N}+\mu_{A})dx \Big),
     \end{align*}
As discussed in Section~\ref{subsec: Quantitative evaluation of the trade-off between IV strength and sample size}, if the IV is valid, we have $f(x)=f(-x)$ for any $x \in \mathbb{R}$. Invoking Lemma~\ref{lemma: convolution}, we can rewrite these 64 terms as
\allowdisplaybreaks
    \begin{align*}
    	&-\int_{-\infty}^{+\infty}\frac{\partial}{\partial \beta} \Big[G_{1}(-x)g_{1}(x)\Big]dx \Big|_{\beta=\beta_{0}}=\iota_{C}^{4}\ \Big(2 \cdot \{f*f\}(0) \Big),	\\	&-\int_{-\infty}^{+\infty}\frac{\partial}{\partial \beta} \Big[G_{1}(-x)g_{2}(x)\Big]dx \Big|_{\beta=\beta_{0}}=\iota_{C}^{3}\ \iota_{N}\  \Big( 2\cdot \{f*f\}(\mu_{N}-\mu_{C}) \Big), \\
    	&-\int_{-\infty}^{+\infty}\frac{\partial}{\partial \beta} \Big[G_{1}(-x)g_{3}(x)\Big]dx \Big|_{\beta=\beta_{0}}=\iota_{C}^{3}\ \iota_{N} \  \Big( \{f*f\}(\mu_{N}-\mu_{C}) \Big), \\
    	&-\int_{-\infty}^{+\infty}\frac{\partial}{\partial \beta} \Big[G_{1}(-x)g_{4}(x)\Big]dx \Big|_{\beta=\beta_{0}}=\iota_{C}^{3}\ \iota_{A} \  \Big( 2\cdot \{f*f\}(\mu_{A}-\mu_{C}) \Big),\\
    	&-\int_{-\infty}^{+\infty}\frac{\partial}{\partial \beta} \Big[G_{1}(-x)g_{5}(x)\Big]dx \Big|_{\beta=\beta_{0}}=\iota_{C}^{3}\ \iota_{A} \ \Big(\{f*f\}(\mu_{A}-\mu_{C}) \Big),\\
    	&-\int_{-\infty}^{+\infty}\frac{\partial}{\partial \beta} \Big[G_{1}(-x)g_{6}(x)\Big]dx \Big|_{\beta=\beta_{0}}=\iota_{C}^{2}\ (\iota_{A}^{2}+\iota_{N}^{2}) \ \Big( \{f*f\}(0) \Big), \\
      	&-\int_{-\infty}^{+\infty}\frac{\partial}{\partial \beta} \Big[G_{1}(-x)g_{7}(x)\Big]dx \Big|_{\beta=\beta_{0}}=\iota_{C}^{2}\ \iota_{A} \ \iota_{N}  \Big(2\cdot \{f*f\}(\mu_{A}-\mu_{N}) \Big),\\
    	&-\int_{-\infty}^{+\infty}\frac{\partial}{\partial \beta} \Big[G_{1}(-x)g_{8}(x)\Big]dx \Big|_{\beta=\beta_{0}}=0,\\
    	&-\int_{-\infty}^{+\infty}\frac{\partial}{\partial \beta} \Big[G_{2}(-x)g_{1}(x)\Big]dx \Big|_{\beta=\beta_{0}}=\iota_{C}^{3}\ \iota_{N} \ \Big( 2 \{f*f\}(\mu_{N}-\mu_{C}) \Big), \\
    	&-\int_{-\infty}^{+\infty}\frac{\partial}{\partial \beta} \Big[G_{2}(-x)g_{2}(x)\Big]dx \Big|_{\beta=\beta_{0}}=\iota_{C}^{2}\ \iota_{N}^{2} \  \Big(2 \cdot \{f*f\}(2\mu_{N}-2\mu_{C}) \Big),\\
    	&-\int_{-\infty}^{+\infty}\frac{\partial}{\partial \beta} \Big[G_{2}(-x)g_{3}(x)\Big]dx \Big|_{\beta=\beta_{0}}=\iota_{C}^{2}\ \iota_{N}^{2} \  \Big(\{f*f\}(0) \Big),\\
    	&-\int_{-\infty}^{+\infty}\frac{\partial}{\partial \beta} \Big[G_{2}(-x)g_{4}(x)\Big]dx \Big|_{\beta=\beta_{0}}=\iota_{C}^{2}\ \iota_{A} \ \iota_{N}\ \Big(2\cdot \{f*f\}(\mu_{A}-\mu_{N}) \Big),\\
    	&-\int_{-\infty}^{+\infty}\frac{\partial}{\partial \beta} \Big[G_{2}(-x)g_{5}(x)\Big]dx \Big|_{\beta=\beta_{0}}=\iota_{C}^{2}\ \iota_{A}\ \iota_{N} \ \Big(\{f*f\}(\mu_{A}+\mu_{N}-2\mu_{C}) \Big),\\
       	&-\int_{-\infty}^{+\infty}\frac{\partial}{\partial \beta} \Big[G_{2}(-x)g_{6}(x)\Big]dx \Big|_{\beta=\beta_{0}}=(\iota_{C}\ \iota_{A}^{2}\ \iota_{N}+\iota_{C}\ \iota_{N}^{3})\  \Big(\{f*f\}(\mu_{N}-\mu_{C}) \Big),\\
    	&-\int_{-\infty}^{+\infty}\frac{\partial}{\partial \beta} \Big[G_{2}(-x)g_{7}(x)\Big]dx \Big|_{\beta=\beta_{0}}=\iota_{C} \ \iota_{A} \ \iota_{N}^{2} \ \Big(2\cdot \{f*f\}(\mu_{A}-2\mu_{N}+\mu_{C})  \Big), \\
    	&-\int_{-\infty}^{+\infty}\frac{\partial}{\partial \beta} \Big[G_{2}(-x)g_{8}(x)\Big]dx \Big|_{\beta=\beta_{0}}=0, \\
       	&-\int_{-\infty}^{+\infty}\frac{\partial}{\partial \beta} \Big[G_{3}(-x)g_{1}(x)\Big]dx \Big|_{\beta=\beta_{0}}=\iota_{C}^{3}\ \iota_{N} \ \Big(\{f*f\}(\mu_{N}-\mu_{C}) \Big),\\
    	&-\int_{-\infty}^{+\infty}\frac{\partial}{\partial \beta} \Big[G_{3}(-x)g_{2}(x)\Big]dx \Big|_{\beta=\beta_{0}}=\iota_{C}^{2} \ \iota_{N}^{2} \ \Big(\{f*f\}(0) \Big),\\
    	&-\int_{-\infty}^{+\infty}\frac{\partial}{\partial \beta} \Big[G_{3}(-x)g_{3}(x)\Big]dx \Big|_{\beta=\beta_{0}}=0,\\
       	&-\int_{-\infty}^{+\infty}\frac{\partial}{\partial \beta} \Big[G_{3}(-x)g_{4}(x)\Big]dx \Big|_{\beta=\beta_{0}}=\iota_{C}^{2}\ \iota_{A}\ \iota_{N}\ \Big(\{f*f\}(\mu_{A}+\mu_{N}-2\mu_{C}) \Big),\\    
    	&-\int_{-\infty}^{+\infty}\frac{\partial}{\partial \beta} \Big[G_{3}(-x)g_{5}(x)\Big]dx \Big|_{\beta=\beta_{0}}=0, \\
    	&-\int_{-\infty}^{+\infty}\frac{\partial}{\partial \beta} \Big[G_{3}(-x)g_{6}(x)\Big]dx \Big|_{\beta=\beta_{0}}=0,\\
    	&-\int_{-\infty}^{+\infty}\frac{\partial}{\partial \beta} \Big[G_{3}(-x)g_{7}(x)\Big]dx \Big|_{\beta=\beta_{0}}=\iota_{C}\ \iota_{A}\ \iota_{N}^{2}\ \Big(\{f*f\}(\mu_{A}-\mu_{C}) \Big),\\
    	&-\int_{-\infty}^{+\infty}\frac{\partial}{\partial \beta} \Big[G_{3}(-x)g_{8}(x)\Big]dx \Big|_{\beta=\beta_{0}}=\iota_{C}\ \iota_{A}\ \iota_{N}^{2}\ \Big(-\{f*f\}(\mu_{A}-2\mu_{N}+\mu_{C}) \Big),\\
    	&-\int_{-\infty}^{+\infty}\frac{\partial}{\partial \beta} \Big[G_{4}(-x)g_{1}(x)\Big]dx \Big|_{\beta=\beta_{0}}=\iota_{C}^{3}\ \iota_{A}\ \Big(2\cdot \{f*f\}(\mu_{A}-\mu_{C}) \Big),	\\
    	&-\int_{-\infty}^{+\infty}\frac{\partial}{\partial \beta} \Big[G_{4}(-x)g_{2}(x)\Big]dx \Big|_{\beta=\beta_{0}}=\iota_{C}^{2}\ \iota_{A}\ \iota_{N}\  \Big(2\cdot \{f*f\}(\mu_{A}-\mu_{N}) \Big), \\
    	&-\int_{-\infty}^{+\infty}\frac{\partial}{\partial \beta} \Big[G_{4}(-x)g_{3}(x)\Big]dx, \Big|_{\beta=\beta_{0}}=\iota_{C}^{2}\ \iota_{A}\ \iota_{N}\  \Big(\{f*f\}(\mu_{A}+\mu_{N}-2\mu_{C}) \Big), \\
    	&-\int_{-\infty}^{+\infty}\frac{\partial}{\partial \beta} \Big[G_{4}(-x)g_{4}(x)\Big]dx \Big|_{\beta=\beta_{0}}=\iota_{C}^{2}\ \iota_{A}^{2} \  \Big(2\cdot \{f*f\}(2\mu_{A}-2\mu_{C}) \Big),\\
    	&-\int_{-\infty}^{+\infty}\frac{\partial}{\partial \beta} \Big[G_{4}(-x)g_{5}(x)\Big]dx \Big|_{\beta=\beta_{0}}=\iota_{C}^{2}\ \iota_{A}^{2} \ \Big(\{f*f\}(0) \Big),\\
    	&-\int_{-\infty}^{+\infty}\frac{\partial}{\partial \beta} \Big[G_{4}(-x)g_{6}(x)\Big]dx \Big|_{\beta=\beta_{0}}=(\iota_{C}\ \iota_{A}^{3}+\iota_{C}\ \iota_{A}\ \iota_{N}^{2} ) \ \Big(\{f*f\}(\mu_{A}-\mu_{C}) \Big), \\
      	&-\int_{-\infty}^{+\infty}\frac{\partial}{\partial \beta} \Big[G_{4}(-x)g_{7}(x)\Big]dx \Big|_{\beta=\beta_{0}}=\iota_{C}\ \iota_{A}^{2} \ \iota_{N}  \Big(2\cdot \{f*f\}(2\mu_{A}-\mu_{N}-\mu_{C}) \Big),\\
    	&-\int_{-\infty}^{+\infty}\frac{\partial}{\partial \beta} \Big[G_{4}(-x)g_{8}(x)\Big]dx \Big|_{\beta=\beta_{0}}=0, \\
    	&-\int_{-\infty}^{+\infty}\frac{\partial}{\partial \beta} \Big[G_{5}(-x)g_{1}(x)\Big]dx \Big|_{\beta=\beta_{0}}=\iota_{C}^{3}\ \iota_{A} \ \Big(\{f*f\}(\mu_{A}-\mu_{C}) \Big), \\
    	&-\int_{-\infty}^{+\infty}\frac{\partial}{\partial \beta} \Big[G_{5}(-x)g_{2}(x)\Big]dx \Big|_{\beta=\beta_{0}}=\iota_{C}^{2}\ \iota_{A}\ \iota_{N} \  \Big(\{f*f\}(\mu_{A}+\mu_{N}-2\mu_{C}) \Big),\\
    	&-\int_{-\infty}^{+\infty}\frac{\partial}{\partial \beta} \Big[G_{5}(-x)g_{3}(x)\Big]dx \Big|_{\beta=\beta_{0}}=0,\\
    	&-\int_{-\infty}^{+\infty}\frac{\partial}{\partial \beta} \Big[G_{5}(-x)g_{4}(x)\Big]dx \Big|_{\beta=\beta_{0}}=\iota_{C}^{2}\ \iota_{A}^{2}\ \Big(\{f*f\}(0) \Big),\\
    	&-\int_{-\infty}^{+\infty}\frac{\partial}{\partial \beta} \Big[G_{5}(-x)g_{5}(x)\Big]dx \Big|_{\beta=\beta_{0}}=0,\\
       	&-\int_{-\infty}^{+\infty}\frac{\partial}{\partial \beta} \Big[G_{5}(-x)g_{6}(x)\Big]dx \Big|_{\beta=\beta_{0}}=0,\\
    	&-\int_{-\infty}^{+\infty}\frac{\partial}{\partial \beta} \Big[G_{5}(-x)g_{7}(x)\Big]dx \Big|_{\beta=\beta_{0}}=\iota_{C} \ \iota_{A}^{2} \ \iota_{N} \ \Big( \{f*f\}(\mu_{N}-\mu_{C}) \Big), \\
    	&-\int_{-\infty}^{+\infty}\frac{\partial}{\partial \beta} \Big[G_{5}(-x)g_{8}(x)\Big]dx \Big|_{\beta=\beta_{0}}=\iota_{C} \ \iota_{A}^{2} \ \iota_{N} \ \Big(-\{f*f\}(2\mu_{A}-\mu_{N}-\mu_{C}) \Big), \\
       	&-\int_{-\infty}^{+\infty}\frac{\partial}{\partial \beta} \Big[G_{6}(-x)g_{1}(x)\Big]dx \Big|_{\beta=\beta_{0}}=(\iota_{C}^{2}\ \iota_{A}^{2}+ \iota_{C}^{2} \ \iota_{N}^{2}) \Big(\{f*f\}(0) \Big),\\
    	&-\int_{-\infty}^{+\infty}\frac{\partial}{\partial \beta} \Big[G_{6}(-x)g_{2}(x)\Big]dx \Big|_{\beta=\beta_{0}}=(\iota_{C}\ \iota_{A}^{2} \ \iota_{N}+ \iota_{C}\ \iota_{N}^{3}) \ \Big(\{f*f\}(\mu_{N}-\mu_{C}) \Big),\\
    	&-\int_{-\infty}^{+\infty}\frac{\partial}{\partial \beta} \Big[G_{6}(-x)g_{3}(x)\Big]dx \Big|_{\beta=\beta_{0}}=0,\\
       	&-\int_{-\infty}^{+\infty}\frac{\partial}{\partial \beta} \Big[G_{6}(-x)g_{4}(x)\Big]dx \Big|_{\beta=\beta_{0}}=(\iota_{C}\ \iota_{A}^{3}+ \iota_{C}\ \iota_{A}\ \iota_{N}^{2}) \Big(\{f*f\}(\mu_{A}-\mu_{C}) \Big),\\    
    	&-\int_{-\infty}^{+\infty}\frac{\partial}{\partial \beta} \Big[G_{6}(-x)g_{5}(x)\Big]dx \Big|_{\beta=\beta_{0}}=0, \\
    	&-\int_{-\infty}^{+\infty}\frac{\partial}{\partial \beta} \Big[G_{6}(-x)g_{6}(x)\Big]dx \Big|_{\beta=\beta_{0}}=0,\\
    	&-\int_{-\infty}^{+\infty}\frac{\partial}{\partial \beta} \Big[G_{6}(-x)g_{7}(x)\Big]dx \Big|_{\beta=\beta_{0}}=(\iota_{A}^{3} \ \iota_{N}+ \iota_{A}\ \iota_{N}^{3})\Big(\{f*f\}(\mu_{A}-\mu_{N}) \Big),\\
    	&-\int_{-\infty}^{+\infty}\frac{\partial}{\partial \beta} \Big[G_{6}(-x)g_{8}(x)\Big]dx \Big|_{\beta=\beta_{0}}=(\iota_{A}^{3} \ \iota_{N}+ \iota_{A}\ \iota_{N}^{3})\Big(-\{f*f\}(\mu_{A}-\mu_{N}) \Big),\\
        &-\int_{-\infty}^{+\infty}\frac{\partial}{\partial \beta} \Big[G_{7}(-x)g_{1}(x)\Big]dx \Big|_{\beta=\beta_{0}}=\iota_{C}^{2}\ \iota_{A}\ \iota_{N} \ \Big(2\cdot \{f*f\}(\mu_{A}-\mu_{N}) \Big),\\
    	&-\int_{-\infty}^{+\infty}\frac{\partial}{\partial \beta} \Big[G_{7}(-x)g_{2}(x)\Big]dx \Big|_{\beta=\beta_{0}}=\iota_{C}\ \iota_{A}\ \iota_{N}^{2} \ \Big(2\cdot \{f*f\}(\mu_{A}-2\mu_{N}+\mu_{C} \Big), \\
    	&-\int_{-\infty}^{+\infty}\frac{\partial}{\partial \beta} \Big[G_{7}(-x)g_{3}(x)\Big]dx \Big|_{\beta=\beta_{0}}=\iota_{C}\ \iota_{A} \ \iota_{N}^{2}\ \Big(\{f*f\}(\mu_{A}-\mu_{C} \Big), \\
    	&-\int_{-\infty}^{+\infty}\frac{\partial}{\partial \beta} \Big[G_{7}(-x)g_{4}(x)\Big]dx \Big|_{\beta=\beta_{0}}=\iota_{C}\ \iota_{A}^{2}\ \iota_{N}\ \Big(2\cdot \{f*f\}(2\mu_{A}-\mu_{N}-\mu_{C}) \Big),\\
    	&-\int_{-\infty}^{+\infty}\frac{\partial}{\partial \beta} \Big[G_{7}(-x)g_{5}(x)\Big]dx \Big|_{\beta=\beta_{0}}=\iota_{C}\ \iota_{A}^{2} \ \iota_{N} \ \Big(\{f*f\}(\mu_{N}-\mu_{C}) \Big),\\
     	&-\int_{-\infty}^{+\infty}\frac{\partial}{\partial \beta} \Big[G_{7}(-x)g_{6}(x)\Big]dx \Big|_{\beta=\beta_{0}}=(\iota_{A}^{3}\ \iota_{N}+\iota_{A}\ \iota_{N}^{3}) \ \Big(\{f*f\}(\mu_{A}-\mu_{N}) \Big), \\
        &-\int_{-\infty}^{+\infty}\frac{\partial}{\partial \beta} \Big[G_{7}(-x)g_{7}(x)\Big]dx \Big|_{\beta=\beta_{0}}=\iota_{A}^{2}\ \iota_{N}^{2} \ \Big(2\cdot \{f*f\}(2\mu_{A}-2\mu_{N}) \Big),\\
    	&-\int_{-\infty}^{+\infty}\frac{\partial}{\partial \beta} \Big[G_{7}(-x)g_{8}(x)\Big]dx \Big|_{\beta=\beta_{0}}=0, \\
    	&-\int_{-\infty}^{+\infty}\frac{\partial}{\partial \beta} \Big[G_{8}(-x)g_{1}(x)\Big]dx \Big|_{\beta=\beta_{0}}=0,\\
    	&-\int_{-\infty}^{+\infty}\frac{\partial}{\partial \beta} \Big[G_{8}(-x)g_{2}(x)\Big]dx \Big|_{\beta=\beta_{0}}=0, \\
    	&-\int_{-\infty}^{+\infty}\frac{\partial}{\partial \beta} \Big[G_{8}(-x)g_{3}(x)\Big]dx \Big|_{\beta=\beta_{0}}=\iota_{C}\ \iota_{A} \ \iota_{N}^{2}\ \Big(-\{f*f\}(\mu_{A}-2\mu_{N}+\mu_{C})\Big), \\
    	&-\int_{-\infty}^{+\infty}\frac{\partial}{\partial \beta} \Big[G_{8}(-x)g_{4}(x)\Big]dx \Big|_{\beta=\beta_{0}}=0,\\
    	&-\int_{-\infty}^{+\infty}\frac{\partial}{\partial \beta} \Big[G_{8}(-x)g_{5}(x)\Big]dx \Big|_{\beta=\beta_{0}}=\iota_{C}\ \iota_{A}^{2} \ \iota_{N} \ \Big(-\{f*f\}(2\mu_{A}-\mu_{N}-\mu_{C})\Big),\\
     	&-\int_{-\infty}^{+\infty}\frac{\partial}{\partial \beta} \Big[G_{8}(-x)g_{6}(x)\Big]dx \Big|_{\beta=\beta_{0}}=(\iota_{A}^{3}\ \iota_{N}+\iota_{A}\ \iota_{N}^{3}) \ \Big(-\{f*f\}(\mu_{A}-\mu_{N})\Big), \\
        &-\int_{-\infty}^{+\infty}\frac{\partial}{\partial \beta} \Big[G_{8}(-x)g_{7}(x)\Big]dx \Big|_{\beta=\beta_{0}}=0,\\
    	&-\int_{-\infty}^{+\infty}\frac{\partial}{\partial \beta} \Big[G_{8}(-x)g_{8}(x)\Big]dx \Big|_{\beta=\beta_{0}}=\iota_{A}^{2}\ \iota_{N}^{2} \ \Big(-2\cdot \{f*f\}(2\mu_{A}-2\mu_{N}) \Big).
     \end{align*}
\endgroup
The desired result follows from summing up these $64$ terms.

We then find the ARE between the two IVs for the sign test under the same setting. Let  $\mu(\beta)=P_{\beta}(X>0)=1-G_{\beta}(0)$, then we have
	\begin{align*}
	    \mu(\beta)&=1-\iota_{C}^{2} \ F(-\beta+\beta_{0})-\iota_{C}\ \iota_{N} \ F(\mu_{N}-\mu_{C}-\beta+\beta_{0})-\iota_{C}\ \iota_{N}\ F(-\mu_{N}+\mu_{C})\\
		&\quad -\iota_{C}\ \iota_{A}\ F(-\mu_{A}+\mu_{C}-\beta+\beta_{0})-\iota_{C}\ \iota_{A}\ F(\mu_{A}-\mu_{C})-(\iota_{A}^{2}+\iota_{N}^{2}) \ F(0)\\
	&\quad -\iota_{A}\ \iota_{N} \ F(-\mu_{A}+\mu_{N}-\beta+\beta_{0})-\iota_{A}\ \iota_{N} \ F(-\mu_{N}+\mu_{A}+\beta-\beta_{0}). 
	\end{align*}
Then we have
	\begin{align*}
	    \mu^{\prime}(\beta_{0})&=\iota_{C}^{2} \ f(0)+\iota_{C}\ \iota_{N} \ f(\mu_{N}-\mu_{C})\\
	    &\quad +\iota_{C}\ \iota_{A}\ f(-\mu_{A}+\mu_{C})+\iota_{A}\ \iota_{N} \ f(-\mu_{A}+\mu_{N})-\iota_{A}\ \iota_{N} \ f(-\mu_{N}+\mu_{A})\\
	    &=\iota_{C}^{2} \ f(0)+\iota_{C}\ \iota_{N} \ f(\mu_{N}-\mu_{C})+\iota_{C}\ \iota_{A}\ f(-\mu_{A}+\mu_{C}).
	\end{align*}
By Example 3.3.5 in \cite{lehmann2004elements}, the ARE under the sign test is the ratio of $\{\mu^{\prime}(\beta_{0})\}^{2}$ under $\varrho_{1}$ and $\varrho_{2}$. So the desired result follows.
\end{proof}

\subsection*{B.2: Proof of Theorem~\ref{thm: ARE in independent case}}

When $\mu_{A} = \mu_{N} = \mu_{C}$, by Theorem~\ref{thm: complete ARE in general case}, we have 
\begin{align*}
	\psi_{\text{wilc}}(\iota_{C},\iota_{A},\iota_{N},\mu_{C}, \mu_{A}, \mu_{N})&= \{ f*f\}(0) \cdot (2\cdot \iota_{C}^{4}+6\cdot \iota_{C}^{3}\ \iota_{A}+6 \cdot \iota_{C}^{2}\ \iota_{A}^{2}+2\cdot \iota_{C}\ \iota_{A}^{3}\\
	&\quad +6\cdot \iota_{C}^{3}\ \iota_{N}+6\cdot \iota_{C}^{2}\ \iota_{N}^{2}+2\cdot \iota_{C}\ \iota_{N}^{3}\\
	&\quad +12\cdot \iota_{C}^{2}\ \iota_{A}\ \iota_{N}+6\cdot \iota_{C}\ \iota_{A}^{2} \ \iota_{N}+6\cdot \iota_{C}\ \iota_{A}\ \iota_{N}^{2})\\
	&=2 \cdot \{ f*f\}(0) \cdot \iota_{C} \cdot (\iota_{C}+\iota_{A}+\iota_{N})^{3}\\
	&=2 \cdot \{ f*f\}(0) \cdot \iota_{C}.
\end{align*}
Therefore, invoking Theorem~\ref{thm: complete ARE in general case}, we have for the Wilcoxon signed rank test,
\begin{equation*}
    \lim_{n \rightarrow \infty}\frac{I_{n,1}}{I_{n,2}}=\frac{\psi_{\text{wilc}}^{2}(\iota_{C,2},\iota_{A,2},\iota_{N,2},\mu_{C,2}, \mu_{A,2}, \mu_{N,2})}{\psi_{\text{wilc}}^{2}(\iota_{C,1},\iota_{A,1},\iota_{N,1},\mu_{C,1}, \mu_{A,1}, \mu_{N,1})}=\frac{\iota_{C,2}^{2}}{\iota_{C,1}^{2}}.
\end{equation*}
Similarly, when $\mu_{A}=\mu_{N}=\mu_{C}$, by Theorem~\ref{thm: complete ARE in general case}, we have
	\begin{align*}
    \psi_{\text{sign}} &= \iota_{C}^{2} \ f(0)+\iota_{C}\ \iota_{N} \ f(\mu_{N}-\mu_{C})+\iota_{C}\ \iota_{A}\ f(-\mu_{A}+\mu_{C})\\
    &=\iota_{C}^{2} \ f(0)+\iota_{C}\ \iota_{N} \ f(0)+\iota_{C}\ \iota_{A}\ f(0)\\
    &=\iota_{C}\ f(0).
	\end{align*}
Therefore, invoking Theorem~\ref{thm: complete ARE in general case}, we have for the sign test,
\begin{equation*}
    \lim_{n \rightarrow \infty}\frac{I_{n,1}}{I_{n,2}}=\frac{\psi_{\text{sign}}^{2}(\iota_{C,2},\iota_{A,2},\iota_{N,2},\mu_{C,2}, \mu_{A,2}, \mu_{N,2})}{\psi_{\text{sign}}^{2}(\iota_{C,1},\iota_{A,1},\iota_{N,1},\mu_{C,1}, \mu_{A,1}, \mu_{N,1})}=\frac{\iota_{C,2}^{2}}{\iota_{C,1}^{2}}.
\end{equation*}

\subsection*{B.3: Proof of Theorem~\ref{thm: impossible}}
It suffices to show that if $\xi(\widetilde{z}, \mathbf{x}, u)\in \mathcal{G}$, then (S2) implies that (S1) is not true, i.e., $\widetilde{Z}\indep U \mid \mathbf{X}$. Since $\xi (\widetilde{z}, \mathbf{x}, u)\in \mathcal{G}$, there exists nonnegative functions $\eta$, $\zeta$, and $\vartheta$, such that $\xi(\widetilde{z}, \mathbf{x}, u)=\eta(\mathbf{x}, u) \zeta (\widetilde{z}, \mathbf{x}) \vartheta(\widetilde{z},u)$, where $\vartheta$ is a smooth function over the support of $\xi$. Without loss of generality, we assume that the support of $\xi$ is $\mathbb{R}^{p+2}$. Otherwise, we just need to apply our argument over that support. To show that $\widetilde{Z}\indep U \mid \mathbf{X}$, since $\xi \in \mathcal{G}$, it suffices to show that $\vartheta(\widetilde{z}, u)=\vartheta_{1}(\widetilde{z})\vartheta_{2}(u)$ for some functions $\vartheta_{1}$ and $\vartheta_{2}$. This is because if $f(\widetilde{Z}=\widetilde{z}\mid \mathbf{X}=\mathbf{x}, U=u)=\xi (\widetilde{z}, \mathbf{x}, u)=\eta(\mathbf{x}, u) \zeta (\widetilde{z}, \mathbf{x}) \vartheta(\widetilde{z}, u)=\eta(\mathbf{x}, u) \zeta (\widetilde{z}, \mathbf{x})\vartheta_{1}(\widetilde{z})\vartheta_{2}(u)$, we have $f(\widetilde{Z}=\widetilde{z}\mid \mathbf{X}=\mathbf{x}, U=u)\propto \zeta (\widetilde{z}, \mathbf{x})\vartheta_{1}(\widetilde{z})$, that is, $\widetilde{Z}\indep U \mid \mathbf{X}$. Therefore, it suffices to show that if (S2) holds true, there exist two functions $\vartheta_{1}$ and $\vartheta_{2}$ such that $\vartheta(\widetilde{z}, u)=\vartheta_{1}(\widetilde{z})\vartheta_{2}(u)$. Assuming that $\widetilde{Z}_{i1} \neq \widetilde{Z}_{i2}$ in pair $i$, we have 
\begin{align*}
    &P(\widetilde{Z}_{i1}=\widetilde{Z}_{i1}\wedge \widetilde{Z}_{i2}, \widetilde{Z}_{i2}=\widetilde{Z}_{i1}\vee \widetilde{Z}_{i2} \mid \mathcal{F}_{1}, \mathcal{Z}, \widetilde{\mathbf{Z}}_{\vee}, \widetilde{\mathbf{Z}}_{\wedge})\\
    &\quad \quad =\frac{\xi(\widetilde{Z}_{i1}\wedge \widetilde{Z}_{i2}, \mathbf{x}_{i1}, u_{i1})\xi(\widetilde{Z}_{i1}\vee \widetilde{Z}_{i2}, \mathbf{x}_{i2}, u_{i2})}{\xi(\widetilde{Z}_{i1}\wedge \widetilde{Z}_{i2}, \mathbf{x}_{i1}, u_{i1})\xi(\widetilde{Z}_{i1}\vee \widetilde{Z}_{i2}, \mathbf{x}_{i2}, u_{i2})+\xi(\widetilde{Z}_{i1}\vee \widetilde{Z}_{i2}, \mathbf{x}_{i1}, u_{i1})\xi(\widetilde{Z}_{i1}\wedge \widetilde{Z}_{i2}, \mathbf{x}_{i2}, u_{i2})}\\
   &\quad \quad =\frac{\text{density odds ratio}}{\text{density odds ratio}+1},
\end{align*}
where we have 
\begin{align*}
    &\text{density odds ratio}\\
    &\quad \quad =\frac{\xi(\widetilde{Z}_{i1}\wedge \widetilde{Z}_{i2}, \mathbf{x}_{i1}, u_{i1})\xi(\widetilde{Z}_{i1}\vee \widetilde{Z}_{i2}, \mathbf{x}_{i2}, u_{i2})}{\xi(\widetilde{Z}_{i1}\vee \widetilde{Z}_{i2}, \mathbf{x}_{i1}, u_{i1})\xi(\widetilde{Z}_{i1}\wedge \widetilde{Z}_{i2}, \mathbf{x}_{i2}, u_{i2})}\\
    &\quad \quad =\frac{\eta(\mathbf{x}_{i1}, u_{i1}) \zeta (\widetilde{Z}_{i1}\wedge \widetilde{Z}_{i2}, \mathbf{x}_{i1}) \vartheta(\widetilde{Z}_{i1}\wedge \widetilde{Z}_{i2}, u_{i1})\eta(\mathbf{x}_{i2}, u_{i2}) \zeta (\widetilde{Z}_{i1}\vee \widetilde{Z}_{i2}, \mathbf{x}_{i2}) \vartheta(\widetilde{Z}_{i1}\vee \widetilde{Z}_{i2}, u_{i2})}{\eta(\mathbf{x}_{i1}, u_{i1}) \zeta (\widetilde{Z}_{i1}\vee \widetilde{Z}_{i2}, \mathbf{x}_{i1}) \vartheta(\widetilde{Z}_{i1}\vee \widetilde{Z}_{i2}, u_{i1})\eta(\mathbf{x}_{i2}, u_{i2}) \zeta (\widetilde{Z}_{i1}\wedge \widetilde{Z}_{i2}, \mathbf{x}_{i2}) \vartheta(\widetilde{Z}_{i1}\wedge \widetilde{Z}_{i2}, u_{i2})}\\
    &\quad \quad = \frac{\vartheta(\widetilde{Z}_{i1}\wedge \widetilde{Z}_{i2}, u_{i1})}{\vartheta(\widetilde{Z}_{i1}\wedge \widetilde{Z}_{i2}, u_{i2})}\cdot \frac{\vartheta(\widetilde{Z}_{i1}\vee \widetilde{Z}_{i2}, u_{i2})}{\vartheta(\widetilde{Z}_{i1}\vee \widetilde{Z}_{i2}, u_{i1})}. \quad \text{(given $\mathbf{x}_{i1}=\mathbf{x}_{i2}$)}
\end{align*}
If (S2) holds true, we have for all $u_{i1}$ and $u_{i2}$, the density odds ratio does not depend on $\widetilde{Z}_{i1}\wedge \widetilde{Z}_{i2}$ or $\widetilde{Z}_{i1}\vee \widetilde{Z}_{i2}$. Therefore, the function $\vartheta(\widetilde{z}, u_{i1})/\vartheta(\widetilde{z}, u_{i2})$, or equivalently, the function $\ln \vartheta(\widetilde{z}, u_{i1})-\ln \vartheta(\widetilde{z}, u_{i2})$ does not depend on $\widetilde{z}$, for all $u_{i1}$ and $u_{i2}$. Thus, we have for all $\widetilde{z}$, $u_{i1}$, and $u_{i2}$,
\begin{equation*}
    \frac{\partial}{\partial \widetilde{z}}\Big (\frac{\ln \vartheta(\widetilde{z}, u_{i1})-\ln \vartheta(\widetilde{z}, u_{i2})}{u_{i1}-u_{i2}}\Big)\equiv 0.
\end{equation*}
Therefore, for all $\widetilde{z}$ and $U$, by smoothness of $\vartheta$, we have
\begin{align*}
    \frac{\partial^{2} \ln \vartheta(\widetilde{z}, u)}{\partial \widetilde{z} \partial u}= \frac{\partial}{\partial \widetilde{z}}\Big( \lim_{u^{\prime} \rightarrow u}\frac{\ln \vartheta(\widetilde{z}, u^{\prime})-\ln \vartheta(\widetilde{z}, u)}{u^{\prime}-u}\Big)=\lim_{u^{\prime} \rightarrow u}\frac{\partial}{\partial \widetilde{z}}\Big (\frac{\ln \vartheta(\widetilde{z}, u^{\prime})-\ln \vartheta(\widetilde{z}, u)}{u^{\prime}-u}\Big)  \equiv 0.
\end{align*}
Let $h(\widetilde{z}, u)=\ln \vartheta(\widetilde{z}, u)$, we then get for all $\widetilde{z}$ and $u$,
\begin{align}\label{equa: pde}
    \frac{\partial^{2} h(\widetilde{z}, u)}{\partial \widetilde{z} \partial u} \equiv 0.
\end{align}
It is well known that if $h(\widetilde{z}, u)=\ln \vartheta(\widetilde{z}, u)$ is a smooth solution to the second order linear partial differential equation (\ref{equa: pde}), we have $h(\widetilde{z}, u)=\omega_{1}(\widetilde{z})+\omega_{2}(u)$ for some functions $\omega_{1}$ and $\omega_{2}$ (\citealp{strauss2007partial}). Thus, let $\vartheta_{1}(\widetilde{z})=\exp(\omega_{1}(\widetilde{z}))$ and $\vartheta_{2}(u)=\exp(\omega_{2}(u))$, we have $\vartheta(\widetilde{z}, u)=\vartheta_{1}(\widetilde{z})\vartheta_{2}(u)$. Therefore, the desired result follows.

\subsection*{B.4: Proof of Theorem \ref{thm: asymp. Wald}}
Note that 
\begin{align*}
    &\quad  \frac{1}{\sqrt{I}}\cdot  \frac{\widehat{\beta}_{\text{IV}}-\beta- \frac{\sum_{i=1}^{I}(Z_{i1}-Z_{i2})[f(\mathbf{X}_{i1})-f(\mathbf{X}_{i2})]}{\sum_{i=1}^{I}(Z_{i1}-Z_{i2})(D_{i1}-D_{i2})}-\delta \cdot \frac{\sum_{i=1}^{I}(Z_{i1}-Z_{i2})(U_{\text{tot},i1}-U_{\text{tot},i2})}{\sum_{i=1}^{I}(Z_{i1}-Z_{i2})(D_{i1}-D_{i2})}}{\frac{\sqrt{2} \sigma }{\sum_{i=1}^{I}(Z_{i1}-Z_{i2})(D_{i1}-D_{i2})}}\\
    &=\frac{1}{\sqrt{2}\sigma} \frac{1}{\sqrt{I}}\sum_{i=1}^{I}(\epsilon_{Ti}-\epsilon_{Ci}), 
\end{align*}
where $\epsilon_{Ti}=\epsilon_{i1}\cdot \mathbbm{1}(\widetilde{Z}_{i1}<\widetilde{Z}_{i2})+\epsilon_{i2}\cdot \mathbbm{1}(\widetilde{Z}_{i1}>\widetilde{Z}_{i2})$, and $\epsilon_{Ci}=\epsilon_{i1}\cdot \mathbbm{1}(\widetilde{Z}_{i1}>\widetilde{Z}_{i2})+\epsilon_{i2}\cdot \mathbbm{1}(\widetilde{Z}_{i1}<\widetilde{Z}_{i2})$. Note that since $\{\epsilon_{n}: n=1,\dots, N\} \ \indep \ \{(\widetilde{Z}_{n}, \mathbf{X}_{n}, U_{\text{tot},n}): n=1,\dots, N\}$, by Definition~\ref{definition: matching alg is valid}, we have $\{\epsilon_{n}: n=1,\dots, N\} \ \indep \ \mathcal{M}$. Thus, we have $\epsilon_{T1}, \dots, \epsilon_{TI}$, $\epsilon_{C1}, \dots, \epsilon_{CI}$ are i.i.d. with expectation $\mathbb{E}[\epsilon]=0$ and variance $\sigma^{2}$. Thus, $\epsilon_{Ti}-\epsilon_{Ci}, i=1,\dots, N$ are i.i.d. with expectation zero and variance $2\sigma^{2}$. By the central limit theorem, we have $\frac{1}{\sqrt{2}\sigma} \frac{1}{\sqrt{I}}\sum_{i=1}^{I}(\epsilon_{Ti}-\epsilon_{Ci})\xrightarrow{\mathcal{L}} \mathcal{N}(0,1)$. So the desired result follows.

\subsection*{B.5: Proof of Theorem~\ref{thm: asymp bias}}
Under the outcome generating model (\ref{eqn: Y model}), $\widehat{\beta}_{\text{IV}}$ can be decomposed into:
\begin{align}\label{equa: decomposition of the wald estimator}
\widehat{\beta}_{\text{IV}}&=\frac{\sum_{i=1}^{I}(Z_{i1}-Z_{i2})(R_{i1}-R_{i2})}{\sum_{i=1}^{I}(Z_{i1}-Z_{i2})(D_{i1}-D_{i2})\nonumber}\\
	&=\frac{\sum_{i=1}^{I}(Z_{i1}-Z_{i2})[(\beta {D}_{i1}+ f(\mathbf{X}_{i1}) +\delta U_{\text{tot},i1} +\epsilon_{i1})-(\beta {D}_{i2}+ f(\mathbf{X}_{i2})+\delta U_{\text{tot},i2} +\epsilon_{i2})]}{\sum_{i=1}^{I}(Z_{i1}-Z_{i2})(D_{i1}-D_{i2})}\nonumber \\
	&=\beta + \frac{\sum_{i=1}^{I}(Z_{i1}-Z_{i2})(f(\mathbf{X}_{i1})-f(\mathbf{X}_{i2}))}{\sum_{i=1}^{I}(Z_{i1}-Z_{i2})(D_{i1}-D_{i2})}+ \delta \cdot \frac{\sum_{i=1}^{I}(Z_{i1}-Z_{i2})(U_{\text{tot},i1}-U_{\text{tot},i2})}{\sum_{i=1}^{I}(Z_{i1}-Z_{i2})(D_{i1}-D_{i2})}\nonumber \\
	&\quad +\frac{\sum_{i=1}^{I}(Z_{i1}-Z_{i2})(\epsilon_{i1}-\epsilon_{i2})}{\sum_{i=1}^{I}(Z_{i1}-Z_{i2})(D_{i1}-D_{i2})}.
\end{align}
For the second term of (\ref{equa: decomposition of the wald estimator}), by Assumption~\ref{assumption: conv of sample quantities}, we have
 \begin{align*}
	\quad \frac{\sum_{i=1}^{I}(Z_{i1}-Z_{i2})(f(\mathbf{X}_{i1})-f(\mathbf{X}_{i2}))}{\sum_{i=1}^{I}(Z_{i1}-Z_{i2})(D_{i1}-D_{i2})}=\frac{\frac{1}{I}\sum_{i=1}^{I}f(\mathbf{X}_{Ti})-\frac{1}{I}\sum_{i=1}^{I}f(\mathbf{X}_{Ci})}{\frac{1}{I}\sum_{i=1}^{I}D_{Ti}-\frac{1}{I}\sum_{i=1}^{I}D_{Ci}}\xrightarrow{p}\frac{\mathbb{E}_{\mathcal{M},T}[f(\mathbf{X})]-\mathbb{E}_{\mathcal{M},C}[f(\mathbf{X})]}{\mathbb{E}_{\mathcal{M},T}[D]-\mathbb{E}_{\mathcal{M},C}[D]}.
 \end{align*} 
For the third term of (\ref{equa: decomposition of the wald estimator}), by Assumption~\ref{assumption: conv of sample quantities}, we have
 \begin{align*}
	\quad \frac{\sum_{i=1}^{I}(Z_{i1}-Z_{i2})(U_{\text{tot},i1}-U_{\text{tot},i2})}{\sum_{i=1}^{I}(Z_{i1}-Z_{i2})(D_{i1}-D_{i2})}=\frac{\frac{1}{I}\sum_{i=1}^{I}U_{\text{tot},Ti}-\frac{1}{I}\sum_{i=1}^{I}U_{\text{tot},Ci}}{\frac{1}{I}\sum_{i=1}^{I}D_{Ti}-\frac{1}{I}\sum_{i=1}^{I}D_{Ci}}\xrightarrow{p}\frac{\mathbb{E}_{\mathcal{M},T}[U_{\text{tot}}]-\mathbb{E}_{\mathcal{M},C}[U_{\text{tot}}]}{\mathbb{E}_{\mathcal{M},T}[D]-\mathbb{E}_{\mathcal{M},C}[D]}.
 \end{align*}
For the last term of (\ref{equa: decomposition of the wald estimator}), as discussed in the proof of Theorem~\ref{thm: asymp. Wald}, by Definition~\ref{definition: matching alg is valid}, we have $\epsilon_{T1}, \dots, \epsilon_{TI}$, $\epsilon_{C1}, \dots, \epsilon_{CI}$ are i.i.d. with expectation zero and variance $\sigma^{2}$. By Assumption~\ref{assumption: conv of sample quantities}, we have
\begin{align*}
	\quad \frac{\sum_{i=1}^{I}(Z_{i1}-Z_{i2})(\epsilon_{i1}-\epsilon_{i2})}{\sum_{i=1}^{I}(Z_{i1}-Z_{i2})(D_{i1}-D_{i2})}=\frac{\frac{1}{I}\sum_{i=1}^{I}\epsilon_{Ti}-\frac{1}{I}\sum_{i=1}^{I}\epsilon_{Ci}}{\frac{1}{I}\sum_{i=1}^{I}D_{Ti}-\frac{1}{I}\sum_{i=1}^{I}D_{Ci}} \xrightarrow{p} \frac{\mathbb{E}[\epsilon]-\mathbb{E}[\epsilon]}{\mathbb{E}_{\mathcal{M},T}[D]-\mathbb{E}_{\mathcal{M},C}[D]}=0.
 \end{align*}
Putting all the above results together, we have
\begin{align*}
     \widehat{\beta}_{\text{IV}}-\beta &=\frac{\sum_{i=1}^{I}(Z_{i1}-Z_{i2})(f(\mathbf{X}_{i1})-f(\mathbf{X}_{i2}))}{\sum_{i=1}^{I}(Z_{i1}-Z_{i2})(D_{i1}-D_{i2})}+ \delta \cdot \frac{\sum_{i=1}^{I}(Z_{i1}-Z_{i2})(U_{\text{tot},i1}-U_{\text{tot},i2})}{\sum_{i=1}^{I}(Z_{i1}-Z_{i2})(D_{i1}-D_{i2})}\\
     &\quad +\frac{\sum_{i=1}^{I}(Z_{i1}-Z_{i2})(\epsilon_{i1}-\epsilon_{i2})}{\sum_{i=1}^{I}(Z_{i1}-Z_{i2})(D_{i1}-D_{i2})}\\
     &\xrightarrow{p}\frac{\mathbb{E}_{\mathcal{M},T}[f(\mathbf{X})]-\mathbb{E}_{\mathcal{M},C}[f(\mathbf{X})]}{\mathbb{E}_{\mathcal{M}, T}[D]-\mathbb{E}_{\mathcal{M}, C}[D]}+\delta \cdot \frac{\mathbb{E}_{\mathcal{M}, T}[U_{\text{tot}}]-\mathbb{E}_{\mathcal{M},C}[U_{\text{tot}}]}{\mathbb{E}_{\mathcal{M},T}[D]-\mathbb{E}_{\mathcal{M},C}[D]}.
 \end{align*}

\clearpage
\begin{center}
{\large\bf Supplementary Material C: Additional Simulations}
\end{center}

\subsection*{C.1: More simulations on Theorem~\ref{thm: ARE in independent case}}
\label{app: simulation results laplace}
\begin{figure}[ht]
\caption{\small Panels (a), (b), and (c): power against sample size for three pairs of IVs with different strength: $\beta - \beta_0 = 0.1$, Laplace error, $\alpha = 0.05$. Panel (d): sample sizes needed to obtain a fixed power for the stronger and weaker IV in each pair. Lines with slopes equal to $\iota^2_{C, 2}/\iota^2_{C, 1}$ are imposed.}
\label{fig: two pics Laplace error}
\begin{subfigure}{.5\linewidth}
\centering
\includegraphics[width = 6 cm, height = 4.2 cm]{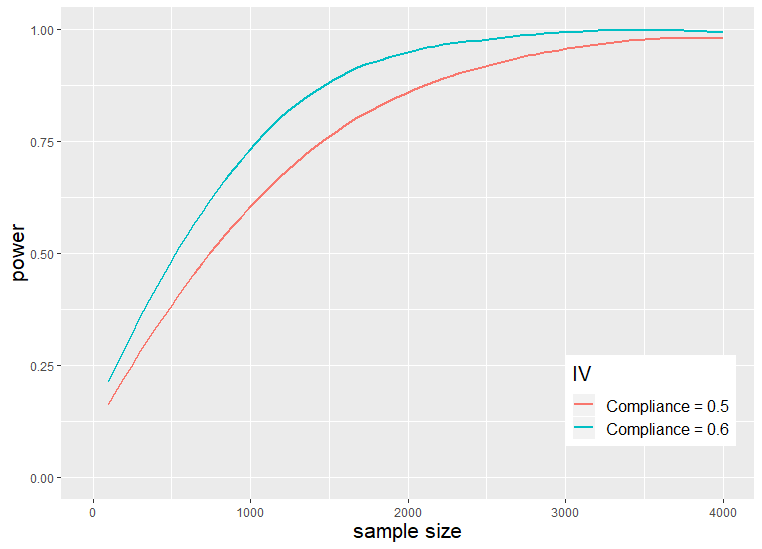}
\caption{Pair 1: $(\iota_{C, 1} = 0.5, ~\iota_{C, 2} = 0.6)$ }\label{fig: power pair 1 laplace}
\end{subfigure}%
\begin{subfigure}{.5\linewidth}
\centering
\includegraphics[width = 6 cm, height = 4.2 cm]{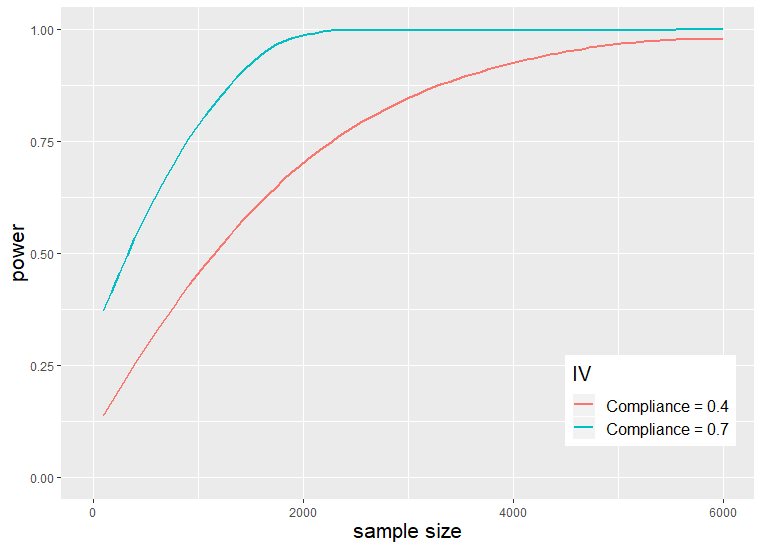}
\caption{Pair 2: $(\iota_{C, 1} = 0.4, ~\iota_{C, 2} = 0.7)$ }\label{fig: power pair 2 laplace}
\end{subfigure}\\[1ex]
\begin{subfigure}{.5\linewidth}
\centering
\includegraphics[width = 6 cm, height = 4.2 cm]{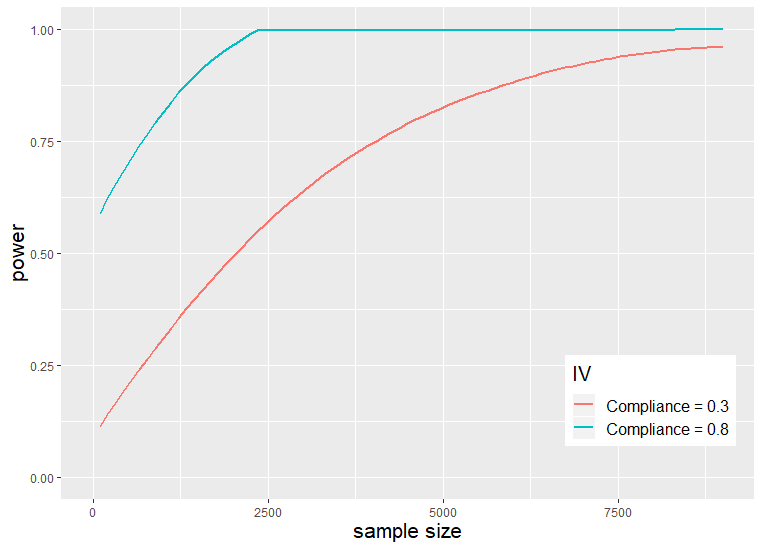}
\caption{Pair 3: $(\iota_{C, 1} = 0.3, ~\iota_{C, 2} = 0.8)$ }\label{fig: power pair 3 laplace}
\end{subfigure}%
\begin{subfigure}{.5\linewidth}
\centering
\includegraphics[width = 7 cm, height = 4.2 cm]{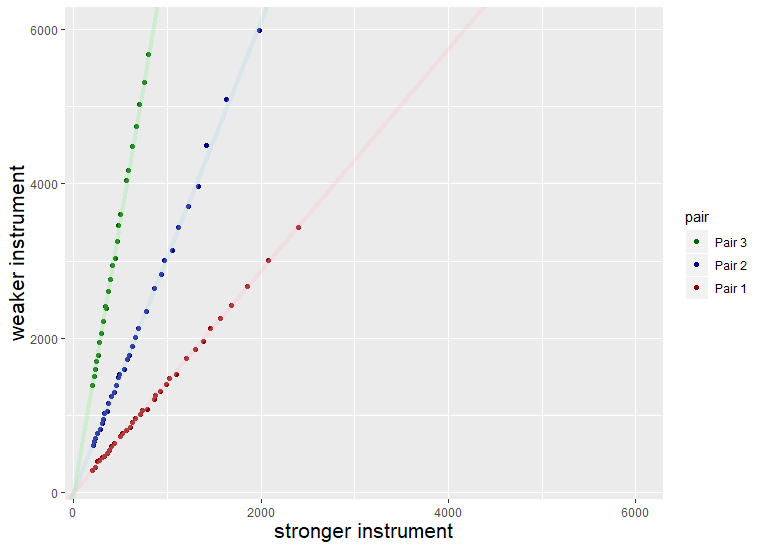}
\caption{Sample size ratio}\label{fig: Laplace sample size ratio}
\end{subfigure}
\end{figure}

\subsection*{C.2: Simulated bias of $\widehat{\beta}_{\text{IV}}$ for two designs}
\label{app: simulation of bias}
\begin{Simulation}[Bias of $\widehat{\beta}_{\text{IV}}$ due to unmeasured confounding]\rm
Let $R$ be the observed outcome, $D$ the binary treatment, $\widetilde{Z}$ the continuous doses of encouragement, $(X_{1}, X_{2})$ two observed covariates, and $U_{\text{tot}}$ an unmeasured confounder. $\{(R_{n}, D_{n}, \widetilde{Z}_{n}, X_{1n}, X_{2n}, U_{\text{tot},n}), n=1,\dots, 400 \}$ represents data before matching and are i.i.d random vectors from the following data generating process: 
\begin{equation*}
\begin{split}
    &R_{n}=\beta \cdot D_{n}+ \sin(X_{1n})+ X_{2n}^{3}+\delta \cdot U_{\text{tot},n} + \epsilon_{1n},\\
    &D_{n}=\mathbbm{1}\{c_{1} \cdot \widetilde{Z}_{n}^{3}+c_{2}\cdot \widetilde{Z}_{n}+0.2\cdot X_{1n}+0.4\cdot X_{2n}+\epsilon_{2n}>c_{3}\},\\
    &U_{\text{tot},n}=f(\widetilde{Z}_{n})+\epsilon_{3n},
\end{split}
\end{equation*}
where $\widetilde{Z}_{n} \overset{\text{i.i.d.}}{\sim} \mathcal{N}(\mu , \sigma_{\text{IV}}^{2})$, $X_{1n} \overset{\text{i.i.d.}}{\sim} \mathcal{N}(0,1)$, $X_{2n} \overset{\text{i.i.d.}}{\sim} \mathcal{N}(0,1)$, $\epsilon_{3n}\overset{\text{i.i.d.}}{\sim} \mathcal{N}(0,1)$ and 
\begin{equation*}
    (\epsilon_{1n}, \epsilon_{2n}) \overset{\text{i.i.d.}}\sim \mathcal{N}\left[\left(\begin{array}{c}
0\\
0
\end{array}\right),\left(\begin{array}{cc}
1 & 0.5 \\
0.5 & 1 
\end{array}\right)\right].\\
\end{equation*} 
Consider the following three specific models:
\begin{itemize}
    \item Model 1: $\delta=1, c_{1}=0, c_{2}=1, c_{3}=0, f(z)=z, \mu=0, \sigma_{\text{IV}}^{2}=1$.
    \item Model 2: $\delta=1, c_{1}=1, c_{2}=1, c_{3}=4, f(z)=1/(z-1), \mu=1, \sigma_{\text{IV}}^{2}=5$.
    \item Model 3: $\delta=10^{-6}, c_{1}=1, c_{2}=1, c_{3}=4, f(z)=\exp(z), \mu=1, \sigma_{\text{IV}}^{2}=5$.
\end{itemize}

We consider two designs $\mathcal{M}_0$ and $\mathcal{M}_1$. $\mathcal{M}_0$ forms $200$ pairs with similar $(X_1, X_2)$ but distinct $\widetilde{Z}$ using optimal non-bipartite matching; $\mathcal{M}_1$ forms $100$ matched pairs that are still similar in $(X_1, X_2)$ but more distinct in $\widetilde{Z}$ by adding $e = 200$ ``sinks'' and a penalty caliper $\Lambda = 8$. See Supplementary Material A.1 for more details on matching. 

In Table~\ref{tab: AOB}, we report the following three quantities averaging over 20,000 replications,  for both designs $\mathcal{M}=\mathcal{M}_0$ and $\mathcal{M}=\mathcal{M}_1$: 1) Estimated compliance rate: $I^{-1}\sum_{i=1}^{I}(Z_{i1}-Z_{i2})(D_{i1}-D_{i2})$; 2) Absolute average encouraged-minus-control difference in $U_{\text{tot}}$: $|I^{-1}\sum_{i=1}^{I}(Z_{i1}-Z_{i2})(U_{\text{tot},i1}-U_{\text{tot},i2})|$; 3) Absolute bias contributed by $U_{\text{tot}}$: $\delta \cdot |\sum_{i=1}^{I}(Z_{i1}-Z_{i2})(U_{\text{tot},i1}-U_{\text{tot},i2})/\sum_{i=1}^{I}(Z_{i1}-Z_{i2})(D_{i1}-D_{i2})|$. We also report the ratio of the magnitude of bias contributed by $U_{\text{tot}}$ for two designs.

\begin{table}[ht]
\footnotesize
    \centering
    \caption{Simulation results: strengthening an IV may amplify the bias}
\begin{tabular}{ c c c c c c c} 
  \hline
  \multirow{3}{*}{} &  \multicolumn{2}{c}{Model 1}  &  \multicolumn{2}{c}{Model 2} & \multicolumn{2}{c}{Model 3}\\ 
  \cmidrule(r){2-3} \cmidrule(r){4-5} \cmidrule(r){6-7}
   & $\mathcal{M}_0$ & $\mathcal{M}_1$ & $\mathcal{M}_0$ & $\mathcal{M}_1$  & $\mathcal{M}_0$ & $\mathcal{M}_1$ \\
\hline
Average compliance rate & 0.32  & 0.37&  0.50 &  0.89  &  0.50 & 0.89   \\
Absolute difference in $U_{\text{tot}}$ & 1.13 & 1.36 & 0.73 & 0.37 & $9.16 \times 10^{5}$ & $4.54 \times 10^{6}$  \\
Absolute bias contributed by $U_{\text{tot}}$ & 3.62 & 3.81& 1.56  & 0.41 & 1.82 & 5.11   \\
Ratio of bias $\Delta$ &\multicolumn{2}{c}{$1.05\approx  1$} &\multicolumn{2}{c}{$0.26 < 1$} &\multicolumn{2}{c}{$2.81 > 1$}\\
 \hline 
 \end{tabular}
 \label{tab: AOB}
\end{table}
\end{Simulation}

Table \ref{tab: AOB} shows that the strengthening-IV design $\mathcal{M}_1$ may render the bias larger ($\Delta > 1$), almost the same $(\Delta \approx 1)$, or smaller $\Delta < 1$, compared to $\mathcal{M}_0$. We see that $\mathcal{M}_1$ always has a larger estimated compliance rate, as expected. Meanwhile, the absolute average encouraged-minus-control difference in $U_{\text{tot}}$ also gets larger in all three data-generating processes, because larger difference in $\widetilde{Z}$ corresponds to larger difference in $U_{\text{tot}}$ according to the data-generating process. Whether or not the magnitude of bias contributed by $U_{\text{tot}}$ would increase or decrease depends on \emph{both} the magnitude of change in the estimated compliance rate and the magnitude of encouraged-minus-control difference in $U_{\text{tot}}$.

\clearpage

\end{document}